\documentclass[USenglish]{article}

\usepackage{fullpage}

\title{\bf Bi-Factor Approximation Algorithms for Hard-Capacitated \texorpdfstring{$k$}{k}-Facility Location Problems\footnote{An extended abstract of this manuscript appeared at SODA 2015.}}

\author{Jaros\l{}aw Byrka
	\thanks{Institute of Computer Science, University of Wroc\l{}aw, Poland} 
	\and Krzysztof~Fleszar
	\thanks{Department of Mathematical Engineering, Universidad de Chile, Chile. Supported by Conicyt PCI PII 20150140 and Millennium Nucleus Information and Coordination in Networks RC130003} 
	\and Bartosz Rybicki
	\thanks{Institute of Computer Science, University of Wroc\l{}aw, Poland. Supported by NCN 2012/07/N/ST6/03068.}
	\and Joachim Spoerhase
	\thanks{Lehrstuhl f\"ur Informatik~I, Universit\"at W\"urzburg, Germany}}
\date{}
\usepackage{graphicx}
\usepackage[utf8]{inputenc}
\usepackage[USenglish]{babel} 
\usepackage{amssymb,amsthm} 
\usepackage{mathtools}
\usepackage[ruled, vlined,nokwfunc]{algorithm2e}
\usepackage{enumerate}
\usepackage{tikz}
\usepackage[caption=false]{subfig}
\usepackage{csquotes}
\usepackage{hyperref}

\usepackage{enumitem}
\newlist{listRoman}{enumerate}{1}
\setlist*[listRoman,1]{label=(\roman*)}
\newlist{inlinelistArabic}{enumerate*}{1}
\setlist*[inlinelistArabic,1]{label=(\arabic*)}
\newlist{inlinelistAlph}{enumerate*}{1}
\setlist*[inlinelistAlph,1]{label=(\alph*)}
\newcommand{\kapitel}{paper\xspace}

\newcommand{\NP}{\ensuremath{\mathsf{NP}}\xspace}

	\newcommand{\PROBkfl}{\textsc{{$k$-Facility} Location}\xspace}
	\newcommand{\PROBfl}{\textsc{Facility Location}\xspace}
	\newcommand{\PROBkmed}{\textsc{{$k$-Median}}\xspace}
	\newcommand{\PROBkcen}{\textsc{{$k$-Center}}\xspace}

	\newcommand{\bigOh}{\mathcal{O}}
	\newcommand{\thSuffix}{-th\xspace}
	\newcommand{\Fig}{Figure}
	\newcommand{\fig}{Fig.}
	\newcommand{\floor}[1]{\left\lfloor #1 \right\rfloor}
	\newcommand{\ceil}[1]{\left\lceil #1 \right\rceil}
	\newcommand{\formulaPunctuationSpace}{~}
	\newcommand{\puthalf}{1/2}
	\newcommand{\elementsetminus}[1]{\setminus\{#1\}}
	\newcommand{\todopic}{
	{\begin{picture}(5,5)\thicklines\put(0,0){\line(1,0){5}}
	\put(0,0){\line(1,2){2.5}}\put(5,0){\line(-1,2){2.5}}
	\put(2.5,2){\makebox(0,0){{\rm !}}}\end{picture}}}
	\newcommand{\todo}[1]{{\setlength{\unitlength}{1mm}\todopic}
	\marginpar{{\setlength{\unitlength}{1mm}\todopic}}
	\footnote{{\setlength{\unitlength}{0.8mm}\todopic} #1}}
	\renewcommand{\todo}[1]{} 
	
	\newcommand\prob[1]{\mathbb{P}\left[{#1}\right]}
	\newcommand\expct[1]{\mathbb{E}\left[{#1}\right]}
	\newcommand{\vol}[1]{\operatorname{vol}(#1)}
	\newcommand{\fvol}[2]{\operatorname{vol}_{#1}(#2)}
	\newcommand{\svol}[1]{\operatorname{vol}(#1)}
	\newcommand{\opt}{\operatorname{OPT}^*\!}
	\newcommand{\copen}[1]{\operatorname{c}_{#1}^1}
	\newcommand{\cclosed}[1]{\operatorname{c}_{#1}^0}

	\newcommand{\metricName}{d} 
	\newcommand{\styleOneDistance}[2]{#1(#2)}
	\newcommand{\dist}[2]{\styleOneDistance{\metricName}{#1,#2}}
	\newcommand{\dav}[1]{\metricName_{\textnormal{av}}(#1)}
	\newcommand{\s}[1]{\operatorname{s}(#1)}
	\newcommand{\numbering}{\operatorname{p}_{\metricName}}

	\newcommand{\ds}[1]{\metricName_{#1}}
	\newcommand{\Cs}[1]{\C_{#1}}
	\newcommand{\Fs}[1]{\F_{#1}}
	\newcommand{\sdist}[3]{\styleOneDistance{\ds{#1}}{#2,#3}}
	\newcommand{\edgeSdist}[2]{\styleOneDistance{\ds{#1}}{#2}}
	
	\newcommand{\C}{\ensuremath{\mathcal{C}}}
	\newcommand{\starCenterSet}{\C_{\textrm{sc}}}
	\newcommand{\F}{\ensuremath{\mathcal{F}}}
	
	\newcommand{\Csbig}[1]{\Cs{#1}^{\textrm{b}}}
	\newcommand{\Cssmall}[1]{\Cs{#1}^{\textrm{s}}}

	\newcommand{\vC}{j} 
	\newcommand{\vSc}{s} 
	\newcommand{\supporting}{supporting\xspace}
	\newcommand{\iraised}{{\skew{2.5}\hat i}}
	\newcommand{\iraisedA}{{\skew{2.5}\hat i_1}}
	\newcommand{\iraisedB}{{\skew{2.5}\hat i_2}}
	\newcommand{\rFS}{r}
	\newcommand{\demand}{d}

\newcommand{\vectorStyle}[1]{\mathbf{#1}}
	\newcommand{\rest}[1]{#1'}
	\newcommand{\assignForST}{a}
	\newcommand{\x}{\vectorStyle{x}}
	\newcommand{\y}{\vectorStyle{y}}
	\newcommand{\xo}{\vectorStyle{x}^*}
	\newcommand{\yo}{\vectorStyle{y}^*}
	\newcommand{\ys}{\vectorStyle{z}}
	\newcommand{\rys}{\bar{\ys}} 
	\newcommand{\ry}{\bar{\y}} 
	\newcommand{\vv}{\vectorStyle{v}}
\newcommand{\componentStyle}[1]{#1}		
		
		\newcommand{\xComp}{\componentStyle{x}}
		\newcommand{\yComp}{\componentStyle{y}}
		\newcommand{\xoComp}{\componentStyle{x}^*}
		\newcommand{\yoComp}{\componentStyle{y}^*}
		\newcommand{\ysComp}{\componentStyle{z}}
		\newcommand{\rysComp}{\bar{\ysComp}} 
		 
		\newcommand{\vvComp}{\componentStyle{v}}

	\newcommand{\ShortTrees}{Short-Trees}
	\newcommand{\BinaryTrees}{Binary-Trees}
	\newcommand{\FacilityStars}{Facility-Stars}

	\newcommand{\lpDemand}{LP~demand}
	\newcommand{\LPfacLoc}{\mbox{Ck-FL~LP}}
	\newcommand{\LPkMed}{\mbox{Ck-MED~LP}}
	\newcommand{\starTail}{star cluster}
	\newcommand{\starTails}{star clusters}
	\newcommand{\StarTails}{Star Clusters}

	\newcommand{\budOpening}[1]{b_{#1}^{\mathrm{f}}}
	\newcommand{\budDirectConnection}[1]{b_{#1}^{\mathrm{d}}}
	\newcommand{\budRelativeConnection}[1]{b_{#1}^{\mathrm{r}}}
	\newcommand{\bsmall}{b^{\textrm{I}}}
	\newcommand{\bbig}{b^{\textrm{II}}}
	\newcommand{\relaxedSolution}{relaxed solution}
	\newcommand{\aRelaxedSolution}{a \relaxedSolution}
	\newcommand{\relaxedSolutions}{\relaxedSolution s}
									
	\newtheorem{theorem}{Theorem}
	\newtheorem{lemma}{Lemma}
	\newtheorem{corollary}{Corollary}
	\newtheorem{definition}{Definition}
	\newtheorem{observation}[theorem]{Observation}

\newcommand{\qedHereInAlign}{\tag*{\qedhere}}

\begin{document}

	\maketitle
					\begin{abstract}
						The~\PROBkfl problem is a generalization of the classical problems~\PROBkmed and \PROBfl.
						The goal is to select a subset
						of at most~${k}$ facilities that minimizes 
						the total cost of opened facilities 
						and 
						established connections between clients and
						opened facilities. 
						We consider the hard-capacitated version of the
						problem, where a single facility may only serve a limited number of
						clients and creating multiple copies of a facility is not allowed.  
						We construct approximation algorithms slightly violating
						the capacities based on rounding a fractional solution to the
						standard LP.
						
						It is well known that the standard LP (even in the case of uniform
						capacities and opening costs) has unbounded integrality gap if we only allow violating
						capacities by a factor smaller than~${2}$, or if we only allow
						violating the number of facilities by a factor smaller than~${2}$. 
						In this \kapitel, we present the first constant-factor approximation algorithms for the hard-capacitated variants of the problem.
						For uniform capacities, we obtain a~${(2+\varepsilon)}$-capacity violating
						algorithm with approximation ratio~${\bigOh(1/\varepsilon^2)}$; our result has not yet been improved.
						Then, for non-uniform capacities, we consider the case of~\PROBkmed, which is equivalent to~\PROBkfl with uniform opening cost of the facilities. 
						Here, we obtain a~${(3+\varepsilon)}$-capacity violating
						algorithm with approximation ratio~${\bigOh(1/\varepsilon)}$.

						Our algorithms first use the clustering of Charikar et al.\ to 
						partition the facilities into sets where the total fractional opening in each set is at
						least~${1-1/\ell}$ for some fixed~${\ell}$.  Then we exploit the
						technique of Levi, Shmoys, and Swamy developed for the capacitated
						\PROBfl problem, which is to locally group the demand from
						clients to obtain a system of single-demand-node instances.  Next,
						depending on the setting, 	
						we either 
						use a dedicated routing tree on 
						the demand nodes (for non-uniform opening cost), or
						we
						work with
						stars of facilities (for non-uniform capacities),
						to redistribute the
						demand that cannot be satisfied locally within the clusters.
					\end{abstract}
					
					\section{Introduction}\label{sec:introduction} 
					In metric location problems, the input consists of a set~${\C}$ of clients, a set~${\F}$ of facilities and a metric distance function~${\metricName}$ on~${\C \cup \F}$. 
					The goal is to select a subset~${\F' \subseteq \F}$ of facilities, 
					and an assignment of clients to the selected 
					facilities, that together minimize a certain problem-specific cost function. 
					One can think of~${\F}$ being a set of potential facility locations, whereas~${\F'}$ contains locations where we decided to open (build) facilities.		
					
					In the~\PROBkmed setting, we search for a subset~${\F' \subseteq \F}$ of cardinality at most~${k}$ and want to minimize the total cost of assigning clients in~${\C}$ to facilities in~${\F'\!}$, where the cost of assigning a client~${\vSc \in \C}$ to a facility~${i \in \F'}$ equals their metric distance~${\dist{\vSc}{i}}$.		
					The~\PROBkmed problem is a classical \NP-hard problem appearing in a number of realistic optimization scenarios. Consider, for example, the location of actual facilities such as voting points during elections, 
					or power plants in an electrical grid. 
					It also appears in the context of clustering data, where one wishes to partition objects into a fixed number of groups containing similar items. 
					
					Similar to~\PROBkmed is the~\PROBkcen problem, 
					where a subset of~${k}$ facilities is selected but the objective is to minimize the maximum distance between a client and its assigned facility.  
					Another related setting is the \PROBfl problem, where instead of the strict constraint of opening at most~${k}$ facilities, we pay a certain cost~${f_i}$ for opening a facility in location~${i \in \F}$.			
					A common generalization of~\PROBkmed and \PROBfl is~\PROBkfl, where there are both, the location specific facility opening cost and the upper bound of~${k}$ on the number of open facilities.
					Note that~\PROBkmed is equivalent to~\PROBkfl with uniform opening costs\footnote{To reduce~\PROBkfl to~\PROBkmed, guess the number of opened facilities in the optimal solution and use this number as~${k}$. In the other direction, set all opening costs to~${0}$.}.
					
					In this \kapitel, we consider the capacitated versions of~\PROBkmed and~\PROBkfl. 
					In this generalization, each facility~${i\in\F}$ has 
					a capacity~${u_i}$ that constrains us 
					to assign at most~${u_i}$ clients to~${i}$.
					If all capacities are the same, we call such a location problem~\emph{uniform}, and, if there are no restrictions on the capacities, we call such problems~\emph{general} or \emph{non-uniform}. 
					We focus on the versions with \emph{hard capacities}, where 
					each facility may be opened at most once, 
					and with \emph{splittable demand}, where a single client may be served from more than one facility. 
					In the simple case of unit demand clients and integral capacity of facilities, the splittability of demands is not important as we discuss in Section~\ref{sec:bundles}.
					The case of unit demand clients carries the essence of 
					capacitated location problems with splittable demand, and, hence, for the simplicity of the argument, we will only consider unit demands.  
					The case of hard capacities is a generalization of the case of \emph{soft capacities}, 
					where one may open multiple copies of the same facility. 
					We will call such location problems \emph{hard-capacitated} and \emph{soft-capacitated}, respectively. 
					In the setting of uniform capacities, the soft- and hard-capacitated versions of~\PROBkmed are equivalent up to a constant factor in the approximation ratio~\cite{ShiLi2015}.
					
										All these mathematical formulations of location problems, although modeling essentially the same clustering task, behave very differently in the context of approximation. 
					
					Best understood is the~\PROBkcen problem, for which a simple and best possible~${2}$-approximation algorithm was given by Hochbaum and Shmoys~\cite{Hochbaum_Shmoys1985}.
					In recent past, Cygan et al.~\cite{Cygan2012} gave a constant-factor approximation algorithm for the capacitated version of the~\PROBkcen problem. 
					The approximation ratio was subsequently improved to~${9}$ by an algorithm of An et al.~\cite{An_Bhaskara_Svensson2015}
					that is based on a natural linear program (LP) relaxation of capacitated~\PROBkcen. 
					This result narrows down the integrality gap of the natural LP relaxation to either one of the three integers\footnote{Cygan et al.~\cite{Cygan2012} give a simple argument that it suffices to consider tree-metrics on unweighted graphs where the optimum solution has length~${1}$. Then any solution has an integral value.}~${7}$,~${8}$, or~${9}$. The best-known lower bound on the approximation factor is~${3}$~\cite{Cygan2012}. 
					
					After a long line of research, the approximability of the uncapacitated \PROBfl problem 
					has been nearly resolved. The~${1.488}$-approximation algorithm of Li~\cite{ShiLi2013} almost closed the gap with the approximability lower bound of~${1.463}$ by Guha and Khuller~\cite{Guha_Khuller1999}. The approximability of the capacitated variant 
					is much less clear. We know that the soft-capacitated problem admits a~${2}$-approximation by Jain et al.~\cite{Jain2003}, which matches the integrality gap of the standard LP. 
					However, the integrality gap of the standard LP 
					for hard-capacitated \PROBfl is unbounded and, for a while, the only successful approach has been local search, 
					which yields a~${3}$-approximation for uniform capacities~\cite{Aggarwal2013} and a~${5}$-approximation for general capacities~\cite{Bansal2012}. 
					Recently, An, Singh and Svensson~\cite{An_LP_Approx_CFL_2014} were successful in obtaining an LP relaxation that yielded a constant-factor approximation algorithm. By this, they answered one of the ten open questions posed in a textbook of Wiliamson and Shmoys~\cite{Williamson_Shmoys2011}.   
					Of interest for our results is 
					an LP-based~${5}$-approximation algorithm for the case with uniform opening costs that was given by Levi~et~al.~\cite{Levi2012}. 
					We will partly 
					build on their techniques in the construction of our algorithm for capacitated~\PROBkmed. 
					
					Despite the simple formulation,~\PROBkmed appears to be the most difficult to handle of the problems above. 
					The first constant-factor approximation algorithm for the uncapacitated~\PROBkmed was achieved by Charikar et al.~\cite{CharikarkMedConst1999} and had an approximation ratio of~${6\frac{2}{3}}$.
					For a long time, the best approximation ratio was~${3+\varepsilon}$ for any positive~${\varepsilon}$, which was obtained by a local-search method~\cite{Arya2001}. 
					Then, not long ago, Charikar and Li~\cite{Charikar2012} gave a~${3.25}$-approximation algorithm by directly rounding the fractional solution to the standard LP. 
					Next, Li and Svensson gave an LP-based algorithm~\cite{Li_Svensson2013} with approximation ratio~${{(1+\sqrt{3} + \varepsilon)}\approx 2.73 + \varepsilon}$, in which they turn a pseudo-approximation algorithm opening a few too many facilities into an algorithm opening at most~${k}$ facilities. Eventually, two ingredients of this algorithm were optimized by Byrka et al.~\cite{Byrka_budgeted2015corrected,Byrka_budgeted2015} to obtain a~${2.675}$-approximation algorithm for~\PROBkmed.  
					
					Until recently, all constant-factor approximation algorithms for capacitated~\PROBkmed were based on the standard LP.
					Since the standard LP 
					has an unbounded integrality gap, it forces to relax some of the constraints. 
					A natural relaxation is to either allow a violation of the capacities by a small factor (we call the factor \emph{capacity violation}), or to allow opening slightly more than~${k}$ facilities. 
					Note that in the well-known integrality gap example~\cite{DemirciL16}, 
					an integral solution must either violate the capacities by at least a factor of~${2-\varepsilon}$ 
					or open at least~${(2-\varepsilon)k}$ facilities 
					in order to have the connection cost within a constant of the optimal solution cost to the standard LP, 
					even for uniform soft capacities. 
					
					The relaxation led to constant-factor approximation algorithms 
					where the factor violating the relaxed constraint is bounded by a constant.
					Charikar et al.~\cite{CharikarkMedConst1999} obtained such a bi-factor approximation algorithm for the setting of uniform soft capacities.
					They presented a~${16}$-approximation algorithm by violating the capacities by a factor of~${3}$.
					Later, Chuzhoy and Rabani~\cite{Chuzhoy2005} gave the first constant-factor approximation algorithm for the non-uniform soft-capacitated case, 
					bounding the capacity violation and the approximation ratio by two-digit constants.			
	Only recently further progress was made. Aardal et al.~\cite{AardalBGL15} designed a~${(7 + \varepsilon)}$-approximation algorithm for the case of 
	general hard capacities using at most~${2k + 1}$ facilities and respecting all capacity constraints.

					\paragraph*{Our results.} 
					We present two algorithms for hard-capacitated~\PROBkfl that are based on the standard LP, one with general opening costs, and one with general capacities. 
					Our aim is to \emph{not} violate the number of open facilities and, simultaneously, to keep the capacity violation as low as possible. 
					
					First, in Section~\ref{sec:k_FL}, we present an algorithm for uniform~\PROBkfl that is still the best known one in its setting. 
					Its capacity violation of at most~$2+\varepsilon$, for any positive~$\varepsilon$, meets the lower bound enforced by the integrality gap example.
					We note that the presentation in our extended abstract~\cite{ByrkaCapKmed2015} had some inaccuracies, as pointed out by Grover et al.~\cite{Grover16,GroverPrivate17}. In parallel to our preparation of this journal version, Grover et al. were able to achieve a slightly higher violation factor of~$3$ avoiding the issues in our extended abstract~\cite{ByrkaCapKmed2015}. Independently of them, we fixed the issue by making a distinction between \emph{strict} and \emph{relaxed} solutions of stars instances in Section~\ref{sec:bundles}. We could also improve the approximation factor by a constant in comparison to the extended abstract. In particular, we obtain the following result: 	 
					\begin{theorem}
						\label{thm:main_2_eps}
						For any~${\ell}$ with~${\ell\geq 2}$,  there is an approximation algorithm 
						for the 
						uniform hard-capacitated~\PROBkfl problem 
						that computes 
						a solution of cost~${8(\ell+1)^2\cdot\opt}$ which violates
						the capacities by a factor at most~${2+3/(\ell-1)}$, where~${\opt}$ is the
						cost of an optimum solution to the standard LP relaxation.
					\end{theorem}
					
					Next, we examine the non-uniform~\PROBkfl problem with uniform opening costs. Recall that this problem is equivalent to non-uniform~\PROBkmed.
					In Section~\ref{sec_ckm}, 
					we 
					describe 
					the first constant-factor approximation algorithm for the hard-capacitated variant of this problem, and
					achieve a capacity violation at most~${3+\varepsilon}$ for any sufficiently small positive~$\varepsilon$.
					More
					specifically, we prove the following.	
			\begin{theorem}\label{thm:main_3_eps}
				For any~${\varepsilon}$ with~${0<\varepsilon\le 1}$, there is an approximation algorithm 
				for the non-uniform hard-capacitated~\PROBkmed problem that computes 
				a solution of cost~${540/\varepsilon \opt +  144\opt}$ which violates 
				the capacities by a factor at most~${3+\varepsilon}$, where~${\opt}$ is the
				cost of an optimum solution to the standard LP relaxation.
			\end{theorem}

					Both our results for~\PROBkfl are built on the idea of Levi et al.~\cite{Levi2012} to decompose the instance into single-demand-node instances.
					We exploit this in Section~\ref{sec:bundles} where we present the tools used by our algorithms.
			
\paragraph*{Subsequent Work.}~ 
Since the publication of our extended abstract~\cite{ByrkaCapKmed2015}, new results were announced. 
Li~\cite{ShiLi2015} introduces a novel LP relaxation for uniform hard-capacitated~\PROBkmed. This allows him to open only~${k (1+\varepsilon)}$ facilities while respecting all capacity constraints. He further develops the LP relaxation and generalizes the result to the case of non-uniform soft capacities~\cite{LiSODA2016}.  
Byrka et al.~\cite{ByrkaIpco2016} use the LP relaxation for uniform hard capacities to open at most~${k}$ facilities and to violate the capacities only by~${1+\varepsilon}$.
The same outcome is achieved by Demirci and Li~\cite{DemirciL16} for the non-uniform hard-capacitated case.
We believe that our results are still of interest as they are based on the substantially simpler standard LP relaxation. Besides that analyzing this relaxation is an interesting question in its own right, the resulting algorithms might also be advantageous in practical applications. Also our approximation ratio has a better asymptotic dependence on $1/\epsilon$, which may lead to better solutions for medium violation factors.

			\section{\StarTails{} and Star Instances}
			\label{sec:bundles}
			Given a capacitated~\PROBkfl instance~${(\C, \F, k, \metricName, u)}$, we will
			partition the facilities of~${\F}$ into \emph{\starTails{}} (similar to
			Charikar and Li~\cite{Charikar2012}). For this, we first solve the
			following natural LP relaxation denoted by~\LPfacLoc{}, where the variable~${\yComp_i}$ encodes the
			opening value (opening) of the facility~${i}$, and the variable~${\xComp_{i\vC}}$ encodes the assignment
			of the client~${\vC}$ to the facility~${i}$.  
			The variable~${\xComp_{i\vC}}$ can also be viewed as the \emph{\lpDemand{}} of the client~${\vC}$ that is send to the facility~${i}$.
			Recall that we consider unit demands, that is, the total \lpDemand{} of the client~${\vC}$ is~${\sum_{i\in \F} \xComp_{i\vC} = 1}$.
			Throughout this \kapitel, we fix an
			integral parameter~${\ell\geq 2}$ and an optimal fractional solution~${(\xo,\yo)}$ to~\LPfacLoc{} and denote its objective value by~${\opt}$. 
			
			\begin{alignat*}{3}
			&\textrm{minimize }&\sum_{i\in \F, \vC\in \C} &\dist{i}{\vC} \xComp_{i\vC} + \sum_{i \in \F} \yComp_i f_i\quad\quad\quad\quad\quad~
			\end{alignat*}
			\begin{alignat}{5}
			&\textrm{subject to }~&\sum_{i\in \F} \yComp_i &\le k; && \label{lp:volume_k}\tag{LP-1}\\
			&&\sum_{i\in \F} \xComp_{i\vC} &= 1  &&\quad\textrm{for each } \vC\in\C; \label{lp:demand_of_client}\tag{LP-2}\\\
			&&\xComp_{i\vC} &\le \yComp_i   &&\quad\textrm{for each } i\in\F, \vC\in\C; \label{lp:opening_ge_demand}\tag{LP-3}\\
			&&\sum_{\vC\in \C} \xComp_{i\vC} &\le u_i \yComp_i  &&\quad\textrm{for each } i\in\F; \label{lp:capacity_ge_total_demand}\tag{LP-4}\\
			&&\xComp_{i\vC},\yComp_i&\ge0  &&\quad\textrm{for each } i\in\F, \vC\in\C\formulaPunctuationSpace.  \nonumber
			\end{alignat}
		
			A \emph{solution with capacity violations} to~\LPfacLoc{} is a solution that satisfies the weaker version of~\LPfacLoc{} where we drop Constraint~\eqref{lp:capacity_ge_total_demand}. 
			In such a solution, the \emph{capacity violation of a facility~${i\in\F}$} 
			is~${\sum_{\vC\in \C} \xComp_{i\vC}}/({u_i \yComp_i})$.  
			We call such a solution also a \emph{solution with capacity violation~${\gamma}$} if~${\gamma\ge \max_{i\in\F} \gamma_i}$. 
	
			As noted in Section~\ref{sec:introduction}, in order to find a solution, it suffices to compute a feasible integral opening vector for the facilities and a possibly fractional assignment of the clients to the open facilities.
			\begin{lemma}\label{lem:min_cost_flow}
				Given a subset~${\F'\subseteq\F}$ of open facilities for which an assignment of the clients exists, we can efficiently compute such an assignment with minimum cost and splittable demands. Moreover, if the capacities are integral, we can obtain a minimum-cost assignment where no demand is split. 
			\end{lemma}
			\begin{proof}
				Given~${\F'\!}$, we fix the corresponding facility openings in~\LPfacLoc{} and solve the LP to obtain a minimum-cost assignment that possibly is fractional.
				If the capacities are integral and we wish to obtain an integral assignment, we model our problem as a minimum-cost flow problem. 
				For this, we take the complete bi-partite graph with the partite sets~${\C}$ and~${\F'\!}$, orient all edges from~${\C}$ to~${\F'}$ and set their capacities to~${1}$ (or any larger integer value) and their costs corresponding to their length in the metric~${\metricName}$.
				Then we introduce a source node that we connect to every client in~${\C}$ via an edge of cost~${0}$ and capacity~${1}$, and, similarly, we introduce a sink node to which we connect every facility~${i\in\F'}$ via an edge of cost~${0}$ and capacity~${u_i}$. 
				We set the required flow to the number of clients.
				 
				Since all capacities and the flow are integral, there is a minimum-cost flow that is integral and we find it efficiently~\cite{Tardos1985}. Hence, each client is \enquote{assigned} by the flow to exactly one facility. 
				
			\end{proof}
			In order to upper bound the connection cost of assignments returned by Lemma~\ref{lem:min_cost_flow}, we will provide possibly suboptimal, fractional assignments of the clients to the open facilities. By upper bounding these, we obtain an upper bound for the assignment obtained by the lemma.
			
			\paragraph{Preliminaries}
			Before obtaining an integral opening value for every facility, our algorithms will 
			operate on smaller subsets of facilities with possibly fractional openings. 
			To ease the description of these procedures, we introduce some helpful notation.	
				\begin{definition}
					An \emph{opening vector}~${\ys}$ for a subset~${\F'\subseteq\F}$ of facilities contains an opening value~${\ysComp_i\in[0,1]}$ 
					for each facility~${i\in\F'}$ and it contains not other values. 
					We say, a facility~${i\in\F'}$ is 
					\begin{itemize}
						\item \emph{closed in~${\ys}$} if~${\ysComp_i=0}$, 
						\item \emph{\supporting in~${\ys}$} if~${\ysComp_i\in(0,1]}$,
						\item \emph{fractional in~${\ys}$} if~${\ysComp_i\in(0,1)}$,
						\item and \emph{open in~${\ys}$} if~${\ysComp_i=1}$. 
					\end{itemize}
					We define the \emph{volume~${\vol{\ys}}$} of~${\ys}$ as~${\sum_{i\in\F'}\ysComp_i}$, 
					and, for any~${\F''\subseteq\F'\!}$, we use~${\fvol{\ys}{\F''}}$ to denote~${\sum_{i\in\F''}\ysComp_i}$.
					We call~${\ys}$ \emph{almost integral} if at most one~${i\in \F'}$ is fractional in~${\ys}$.
					
					Let~${\F',\F''\subseteq \F}$ be two disjoint sets and~${\ys'}$ an opening vector for~${\F'}$ and~${\ys''}$ an opening vector for~${\F''}$. 
					The \emph{union} of~${\ys'}$ and~${\ys''}$ is an opening vector~${\ys}$ for~${\F'\cup\F''}$ with~${\ysComp_i=\ysComp'_i}$ for each~${i\in\F'\!}$, 
					and~${\ysComp_i=\ysComp''_i}$ for each~${i\in\F''}$.
				\end{definition}

				\begin{definition}\label{def:facVol}
					For any set~${\F'\subseteq\F}$ of facilities, we define its \emph{volume} as~${\fvol{\yo}{\F'}}$. 
				\end{definition}

			In the metric~${\metricName}$, a node can have the same distance to multiple nodes. To avoid ambiguity, 
			we could arbitrarily define one of the multiple nodes to be its \emph{closest} node.
			However, our algorithms will need a stronger property: 
			We will have to avoid cycles of length more than two where, for each node of the cycle, its closest node is its neighbor in a fixed orientation. 
			We can achieve this by assigning to every edge~${\{\vSc,\vC\}\subseteq\C\cup\F}$ a distinct priority~${\numbering(\{\vSc,\vC\})}$.
			Now, informally speaking, the closest node~${\vC}$ to a node~${\vSc}$ is the node with the smallest distance to~${\vSc}$ and, among all the nodes with the smallest distance  to~${\vSc}$, it is the node whose edge connecting to~${\vSc}$ has the smallest value in~${\numbering}$. 
			Suppose there is a cycle as described above, then all its edges have the same distance and exactly one of the edges has the smallest priority value. Both its endpoints are thus closest to each other, implying that the cycle is of length~${2}$; a contradiction.
			We define the notion of~\emph{closeness} more precisely.
			\begin{definition}\label{def:closeness}

				Let~${A\subseteq\C\cup\F}$ be a non-empty set and let~${\vSc\in\C\cup\F}$. 
				If~${\vSc\in A}$, then the \emph{closest} node \emph{in~${A}$ to~${\vSc}$} is~${\vSc}$.
				Otherwise,
				let \[{A^{\textrm{min}}_\vSc = \{\vC\in A\mid\ \nexists \vC'\in A : \dist{\vSc}{\vC'}<\dist{\vSc}{\vC}\}}\formulaPunctuationSpace.\]
				The \emph{closest} node \emph{in~${A}$ to~${\vSc}$} is~${\arg\min_{\vC\in A^{\textrm{min}}_\vSc} \numbering(\{\vSc,\vC\})}$.
If~${\vC}$ is the closest node in~${A}$ to~${\vSc}$, we also say: \emph{$\vC\in A$ is closest to~${\vSc}$}.
			\end{definition}

			 			\paragraph{Graphs on Clients and Facilities.}
			 			To simplify the description of our algorithms, we will build directed acyclic graphs based on either the clients or the facilities. 
			 			First, we fix some notations and then present a procedure that we will use to construct forests of rooted in-trees.
			 			
			 			\begin{definition}
			 				For any graph~${G}$, its node set is denoted by~${V(G)}$ and its edge set is denoted by~${E(G)}$. 
			 				Let~${(\vSc,\vC)}$ be an edge of a directed acyclic graph~${G}$.
			 				We call~${(\vSc,\vC)}$ an \emph{outgoing edge of~${\vSc}$} and an \emph{incoming edge of~${\vC}$}. We also call~${\vSc}$ a \emph{son} of~${\vC}$, 
			 				and~${\vC}$ a \emph{father} of~${\vSc}$. 
			 				Sons of the same father are called~\emph{brothers}.
			 				The \emph{indegree} of a node is the number of its incoming edges, and the \emph{outdegree} of a node is the number of its outgoing edges.
			 				Moreover, any node in~${G}$ with outdegree~${0}$ is called a \emph{root}.
			 			\end{definition}
			 			
			 			Below, we present a procedure that, given two disjoint subsets~${A,B\subseteq \F\cup\C}$, 
			 			constructs a directed forest, where each node in~${A}$ has either a directed edge to its closest distinct node in~${A \cup B}$, or is a root (recall Definition~\ref{def:closeness}). We will show that its components are in-trees. The procedure assumes that~${A}$ is not empty and~${A\cup B}$ contains at least two elements. 
			 			
			 			\begin{procedure}
			 				\caption{\ShortTrees($A$,~${B}$)}
			 				Create~${G}$ with initially~${V(G)=A}$ and~${E(G)=\emptyset}$\; 
			 				\ForEach{$\vSc \in A$}{
			 					select~${\vSc'\in A \cup B \elementsetminus{\vSc}}$ closest to~${\vSc}$\; 
			 					add~${\vSc'}$ to~${V(G)}$ if not already contained\;
			 					add the directed edge~${(\vSc,\vSc')}$ to~${E(G)}$\;
			 				}
			 				\ForEach{$(\vSc, \vSc'), (\vSc', \vSc) \in E(G)$}{
			 					remove the edge~${(\vSc, \vSc')}$ from~${G}$\;
			 				}
			 				\Return~${G}$\;
			 			\end{procedure}
			 			
			 			Let~${G}$ be the output of Procedure \ShortTrees($A$,~${B}$) on any disjoint subsets~${A,B\subseteq \F\cup\C}$.
			 			We establish the following properties of~${G}$.

			 			\begin{lemma}\label{decrease_edge_length}
			 				Let~${\vSc,\vSc',\vSc'' \in V(G)}$. If~${(\vSc, \vSc'), (\vSc', \vSc'') \in E(G)}$, 
			 				then~${\dist{\vSc}{\vSc'} \ge \dist{\vSc'}{\vSc''}}$. 
			 			\end{lemma}
			 			\begin{proof}
			 				We have~${\vSc\ne \vSc''}$ because the algorithm removes one edge from
			 				each cycle of length~${2}$.  By the construction of the edge set,~${\vSc''}$ is the closest node 
			 				in~${V(G)}$ 
			 				to~${\vSc'}$. 
			 				Thus, we have~${\dist{\vSc'}{\vSc''}\leq \dist{\vSc}{\vSc'}}$.  
			 			\end{proof}
			 			
			 			\begin{lemma}\label{lem:forest-intrees}
			 				The graph~${G}$ 
			 				is a forest
			 				of in-trees and~${A\subseteq V(G)}$.
			 			\end{lemma}
			 			\begin{proof}
			 				By construction,~${A\subseteq V(G)}$.
			 				Furthermore, every node in the graph 
			 				has outdegree at most~${1}$.
			 				Suppose there is a cycle. Then,
			 				by Lemma~\ref{decrease_edge_length}, all edges have the same length.
			 				Recall Definition~\ref{def:closeness} and consider the edge~${(\vSc, \vSc')}$ on the cycle that yields the smallest priority value~${\numbering(\{\vSc, \vSc'\})}$. 
			 				By the definition of closeness, both endpoints~${\vSc}$ and~${\vSc'}$ are the closest nodes to each other. By the construction of the edge set, the cycle must contain~${(\vSc, \vSc')}$ and~${(\vSc', \vSc)}$;
			 				a contradiction as cycles of length~${2}$ have been removed by the procedure.
			 			\end{proof}
			 			
			 			\begin{lemma}\label{lem:root-son-distance}
			 				Every root in~${G}$ has at least one son, and one of its sons is the closest node in~${A \cup B}$ to the root.
			 				Furthermore, all nodes in~${B\cap V(G)}$ are roots.
			 			\end{lemma}
			 			\begin{proof}
			 				During the construction of~${G}$, every node receives an outgoing edge to its closest node. 
			 				Since the root of an in-tree does not contain any outgoing edge in the returned graph, its outgoing edge is removed at the end of the construction. 
			 				Given the condition on which an edge is deleted, the node closest to the root is one of its sons. 
			 				
			 				The second claim follows from the fact that~${A}$ and~${B}$ are disjoint and that we therefore never add an outgoing edge to a node in~${B}$.
			 			\end{proof}

			\paragraph{\StarTails{}.} 	
			Central for our algorithms is the following quantity.
			For every client~${\vC\in\C}$, we define~${\dav{\vC} = \sum_{i\in\F} \xoComp_{i\vC} \dist{i}{\vC}}$ as \emph{its connection cost in~${(\xo,\yo)}$}. In general, connection cost refers to the cost~${w \cdot l}$ of sending~${w}$ units of demand along a distance~${l}$. 
				
			We select a
			subset~${\starCenterSet\subseteq\C}$ of clients that are \emph{far away} from each
			other with respect to their connection cost. Beginning
			with~${\starCenterSet=\emptyset}$ and~${\C'=\C}$, we select a client~${\vSc\in\C'}$
			with minimum~${\dav{\vSc}}$ (ties are broken arbitrarily) and remove it
			from~${\C'}$ along with every client~${\vC}$ with~${\dist{\vSc}{\vC}\le2\ell\dav{\vC}}$, 
			and add~${\vSc}$ to~${\starCenterSet}$. We repeat this procedure until~${\C'}$
			is empty. We call the clients in~${\starCenterSet}$ \emph{star centers}.

			\begin{lemma}\label{lem:center-distance}
				The following holds:
				\begin{enumerate}
					\item \label{lem:center-distances:center-center}
					For any~${\vSc,\vSc'\in\starCenterSet}$,~${\vSc\not=\vSc'}$, it holds~${\dist{\vSc}{\vSc'}> 2\ell\max\{ \dav{\vSc},\dav{\vSc'}\}}$. 
					\item \label{lem:center-distances:client-center} 
					For any~${\vC\in\C\setminus\starCenterSet}$, there is a client~${\vSc\in\starCenterSet}$ with~${\dist{\vSc}{\vC} \le 2\ell\dav{\vC}}$.
					\item \label{lem:center-distances:dav} 
					For any~${\vSc\in\starCenterSet}$ and any~${\vC\in\C}$, it holds~${\dav{\vSc}\le 2\dav{\vC}}$ if either~${\dist{\vSc}{\vC}\le \min\{\dist{\vSc}{\vSc'}\mid \vSc'\in\starCenterSet,\vSc'\ne \vSc\}/2}$ or~${|\starCenterSet|=1}$.
				\end{enumerate}
			\end{lemma}
			\begin{proof}
				Let~${\vSc,\vSc'\in\starCenterSet}$. 
				Without loss of generality,~${\vSc}$ was added before~${\vSc'}$ into~${\starCenterSet}$ and thus we have~${\dav{\vSc}\le\dav{\vSc'}}$. Since~${\vSc'}$ was not removed from~${\C'}$ when~${\vSc}$ was added to~${\starCenterSet}$, we also have~${2\ell\dav{\vSc'}<\dist{\vSc}{\vSc'}}$, and Claim~\ref{lem:center-distances:center-center} follows.
				
				Next, let~${\vC\in\C\setminus\starCenterSet}$. As~${\vC\not\in\starCenterSet}$, it was removed from~${\C'}$ when a client~${\vSc}$ was added to~${\starCenterSet}$. Thus, we have~${\dist{\vSc}{\vC}\le2\ell\dav{\vC}}$ and Claim~\ref{lem:center-distances:client-center} follows.
				
				Now, let~${\vSc\in\starCenterSet}$ and~${\vC\in\C}$. 
				If~${|\starCenterSet|=1}$, then~${\vSc}$ is the only star center in~${\starCenterSet}$. Then we know that~${\vC}$ was removed from~${\C'}$ when~${\vSc}$ was added to~${\starCenterSet}$ (and possibly~${\vSc=\vC}$).
				Hence, our greedy construction of~${\starCenterSet}$ implies~${\dav{\vSc}\le\dav{\vC}}$, and the claim follows.
				Otherwise, $|\starCenterSet|\ge 2$ and let~${R=\min\{\dist{\vSc}{\vSc'}\mid \vSc'\in\starCenterSet,\vSc'\ne \vSc\}/2}$. 
				By assumption,~${\dist{\vSc}{\vC}\le R}$.
				As a consequence of Claim~\ref{lem:center-distances:center-center},~${R>0}$. 
				Thus,~${\dist{\vSc}{\vC}<2R}$, and either~${\vC=\vSc}$ or~${\vC\in\C\setminus\starCenterSet}$.
				In the first case, the claim follows immediately.
				In the second case, Claim~\ref{lem:center-distances:client-center} gives us a client~${\vSc'\in\starCenterSet}$ with~${\dist{\vSc'}{\vC}\le2\ell\dav{\vC}}$. 
				If~${\vSc=\vSc'}$ and~${\vSc}$ is the only such client for~${\vC}$, then~${\vC}$ was removed from~${\C'}$ when~${\vSc}$ was added to~${\starCenterSet}$.
				As discussed above, this event implies the claim. 
				Otherwise, there is a star center~${\vSc'\in\starCenterSet}$ with~${\dist{\vSc'}{\vC}\le2\ell\dav{\vC}}$ and~${\vSc'\ne\vSc}$.
				By the minimality of~${R}$ and the triangle inequality, we have 
				\[{2R\le \dist{\vSc'}{\vSc}\le \dist{\vSc'}{\vC}+\dist{\vC}{\vSc}\le 2\ell\dav{\vC} + R}\formulaPunctuationSpace.\]	
				Consequently,~${R\le 2\ell\dav{\vC}}$. 
				Claim~\ref{lem:center-distances:center-center} implies~${\ell\dav{\vSc}<R}$ and, hence,  
				altogether, we get~${\ell\dav{\vSc} < 2\ell\dav{\vC}}$ and Claim~\ref{lem:center-distances:dav} follows.
			\end{proof} 
			
			\begin{definition}
				For each star center~${\vSc\in\starCenterSet}$, we define the \starTail{}~${\F_\vSc\subseteq\F}$ as the set of all facilities that are closest to~${\vSc}$, that is,
				\[{\F_\vSc = \{i\in\F \mid \vSc \textrm{ is the closest node in } \starCenterSet \textrm{ to } i \}}\formulaPunctuationSpace.\] 
			\end{definition}
			Recall that by our definition of closeness (Definition~\ref{def:closeness}), each facility has a unique closest star center. Therefore, the \starTails{} partition all facilities in~${\F}$. 

	\begin{figure}
\centering
		\includegraphics{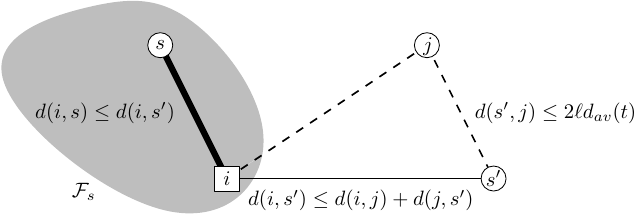}
		\caption{
			Let~${\vSc\in\starCenterSet}$, $\vC\in\C$, $i\in\F_\vSc$, and~${\vSc'\in\starCenterSet}$ be the closest star center to~${i}$.
			The distance~${\dist{i}{\vSc}}$ (thick edge) is bounded by the distance~${\dist{i}{\vSc'}}$ (solid edge), which in turn is bounded by the detour over~${\vC}$ (dashed path).}
		\label{fig:triangle_inequalities} 
	\end{figure}
			\begin{lemma}\label{lem:distance-facility}
				For any~${\vSc\in\starCenterSet}$, ${\vC\in\C}$ and~${i\in\F_\vSc}$, the following holds:
				\begin{enumerate}
					\item \label{lem:distance-facility:fac-center} $\dist{i}{\vSc}~\le~ \dist{i}{\vC}+2\ell\dav{\vC}$
					\item \label{lem:distance-facility:center-center} $\dist{\vSc}{\vC}\le2\dist{i}{\vC}+2\ell\dav{\vC}$
				\end{enumerate}
			\end{lemma}		
			\begin{proof}
				Let~${\vSc'\in\starCenterSet}$ be the star center closest to~${\vC}$ (possibly~${\vSc'\in\{\vSc,\vC\}}$). 
				If~${\vC\not\in\starCenterSet}$, then, by Lemma~\ref{lem:center-distance}.\ref{lem:center-distances:client-center}, we have~${\dist{\vSc'}{\vC}\le2\ell\dav{\vC}}$. 
				Otherwise,~${\vSc'=\vC}$, and the same inequality holds. 
				Consider \fig~\ref{fig:triangle_inequalities}. Since~${i}$ belongs to~${\F_{\vSc}}$, we have  
				\[\dist{i}{\vSc}\leq
				\dist{i}{\vSc'}\leq 
				\dist{i}{\vC}+\dist{\vC}{\vSc'}\leq 
				\dist{i}{\vC}+2\ell \dav{\vC} \] 
				and
				\[\dist{\vSc}{\vC}\leq
				\dist{\vSc}{i}+\dist{i}{\vC}\leq 2\dist{i}{\vC}+2\ell \dav{\vC}\formulaPunctuationSpace.\qedHereInAlign\] 
			\end{proof}

						Next, we bound the volume of facilities that are close to their star centers.
						\begin{lemma}\label{lem:minvolume-R}
							For any positive~${R}$ and any~${\vSc\in\starCenterSet}$, let~${\F_\vSc^R=\{i\in\F_\vSc\mid\dist{i}{\vSc}\le R\}}$.
							We have \[{\fvol{\yo}{\F_\vSc^R}\ge 1-\dav{\vSc}/R}\formulaPunctuationSpace.\]
						\end{lemma}
						\begin{proof}
							Note that at most a portion~${\dav{\vSc}/R}$ of the \lpDemand{}~${\sum_{i\in \F} \xoComp_{i\vSc}}$ of~${\vSc}$ can be served by the facilities in~${\F\setminus\F_\vSc^R}$. 
							Otherwise,
							\[{\dav{\vSc}=\sum_{i\in\F} \xoComp_{i\vSc} \dist{i}{\vSc}\ge R \cdot\!\! \sum_{i\in\F\setminus\F_\vSc^R} \xoComp_{i\vSc}  > R \cdot \dav{\vSc}/R = \dav{\vSc}}\formulaPunctuationSpace;\] 
							a contradiction. 
							Hence at least a portion~${1-\dav{\vSc}/R}$ of the
							\lpDemand{} of~${\vSc}$ is served by the facilities in~${\F_\vSc^R}$. 
							Hence, by Constraint~\eqref{lp:opening_ge_demand},
							\[{\fvol{\yo}{\F_\vSc^R} = \sum_{i\in\F_\vSc^R} \yoComp_i \ge \sum_{i\in\F_\vSc^R} \xoComp_{i\vSc} \ge 1-\dav{\vSc}/R}\formulaPunctuationSpace.\] 
						\end{proof}

						Our last result implies the following bound on the volume of \starTails{}, which has also been shown by Charikar et al.~\cite{CharikarkMedConst1999}.
						\begin{corollary}[\cite{CharikarkMedConst1999}]\label{cor:minvolume-bundle}
							The volume of every \starTail{} is at least~${1-1/\ell}$.
						\end{corollary}	
						\begin{proof}
							By Lemma~\ref{lem:center-distances:center-center}, any facility~${i\in\F}$ lying within the radius~${\ell\dav{\vSc}}$ around a star center~${\vSc\in\starCenterSet}$ belongs to~${\F_\vSc}$. 
							If~${\dav{\vSc}=0}$, then, by the definition of~${\dav{\vSc}}$, all facilities serving~${\vSc}$ must have distance~${0}$ to~${\vSc}$. 
							Thus,~${\fvol{\yo}{\F_\vSc^R}\ge 1}$. 
							If~${\dav{\vSc}>0}$, then we set~${R=\ell\dav{\vSc}}$ and apply Lemma~\ref{lem:minvolume-R}.
						\end{proof}

				In our algorithms and in their analyses, we will charge portions of~${\opt}$ to the star centers. These portions are described by the following quantities.
				\begin{definition}
					For every star center~${\vSc\in\starCenterSet}$, we define the \emph{budgets}
					\begin{itemize}
						\item~${\budOpening{\vSc} = \sum_{i\in\F_{\vSc}} \yoComp_i f_i}$ as the \emph{opening cost budget},
						\item~${\budDirectConnection{\vSc} = \sum_{i\in\F_{\vSc}} \sum_{\vC\in\C} \xoComp_{i\vC} \dist{i}{\vC}}$ as the \emph{direct connection cost budget}, and
						\item~${\budRelativeConnection{\vSc} = \sum_{i\in\F_{\vSc}} \sum_{\vC\in\C} \xoComp_{i\vC} \dav{\vC}}$ as the \emph{relative connection cost budget}.
					\end{itemize}
				\end{definition}
				The idea for these quantities is to distribute the total opening and connection cost 
				among the star centers.
				Then these cost portions can be used by each star center to upper bound some local solutions. Consequently, the total cost of all these solutions will be bounded from above by a function of~${\opt}$. 
				For the first two budgets, we split the total cost evenly between the facilities.
				For a star center~${\vSc\in\starCenterSet}$,
				the budget~${\budOpening{\vSc}}$ 
				is the total opening cost of the facilities in~${\F_{\vSc}}$, 
				and~${\budRelativeConnection{\vSc}}$ 
				is the total connection cost of transporting demand to facilities in~${\F_{\vSc}}$.
				For the third budget, we split the total connection cost proportional to the amount of demand that is served by the facilities of the \starTails{}. To see this, consider the contribution of any client~${\vC\in\C}$ to~${\budRelativeConnection{\vSc}}$. It is~${\dav{\vC} \sum_{i\in\F_{\vSc}} \xoComp_{i\vC}}$. This quantity corresponds to the connection cost~${\dav{\vC}}$ of~${\vC}$ weighted by the \lpDemand{} that~${\vC}$ sent to~${\F_{\vSc}}$.
				Our observations about the quantities are confirmed by the following lemma. 
				\begin{lemma}\label{lem:budget-equations}
					The following holds:
					\begin{enumerate}
						\item \label{lem:budget-equations-opt}  $\sum_{\vSc\in\starCenterSet} \budOpening{\vSc} + \budDirectConnection{\vSc} = \opt$.
						\item \label{lem:budget-equations-connection}  $\sum_{\vSc\in\starCenterSet} \budRelativeConnection{\vSc} \le \opt$.
					\end{enumerate}
				\end{lemma}
				\begin{proof} Both statements follow directly from the definitions of the budgets and the fact that the \starTails{} partition~${\F}$: 
					\begin{enumerate}
						\item
						\hfill$\begin{aligned}[t]
						\sum_{\vSc\in\starCenterSet} \budOpening{\vSc} + \budDirectConnection{\vSc}\enskip&=&& \sum_{\vSc\in\starCenterSet}\sum_{i\in\F_{\vSc}} \left( \yoComp_if_i + \sum_{\vC\in \C}\xoComp_{i\vC} \dist{i}{\vC}\right)\\ 
						&=&&\sum_{i \in \F} \yoComp_i f_i + \sum_{\vC\in \C}\sum_{i\in\F}\xoComp_{i\vC}\dist{i}{\vC}\\ 
						&=&& \opt\\~
						\end{aligned}$\hfill\null
						\item
						\hfill$\begin{aligned}[t]
						\sum_{\vSc\in\starCenterSet} \budRelativeConnection{\vSc}&=&\enskip& \sum_{\vSc\in\starCenterSet}\sum_{i\in\F_{\vSc}}\sum_{\vC\in \C}\xoComp_{i\vC} \dav{\vC}\\ 
						&=&&\sum_{\vC\in\C}\sum_{i\in\F}\xoComp_{i\vC}\dav{\vC} \\ 
						&=&& \sum_{\vC\in\C}\dav{\vC}\underbrace{\sum_{i\in\F}\xoComp_{i\vC}}_{=1}\\
						&=&& \sum_{\vC\in \C}\sum_{i\in\F}\xoComp_{i\vC}\dist{i}{\vC}\\
						&\le&& \opt
						\end{aligned}$\hfill\null
					\end{enumerate}
				\end{proof}

			\paragraph{Star instances.}
			We will now introduce a notion that will help us to locally modify
			facility openings. Following the approach of Levi et
			al.~\cite{Levi2012}, we ``move'' the total demand served by
			facilities of a \starTail{} $\F_\vSc$ to its center~${\vSc}$. We call the
			so-obtained single-demand-node instances \emph{star instances}.  We
			will see that extreme point solutions to such star instances
			have a particular nice structure. This will enable us later to round
			these solutions to integral solutions with small constant capacity
			violation at the expense of only a constant factor in the connection and opening cost.  On the other hand, we will show that we can
			decompose the original instance into a collection of star instances
			and that we pay only a constant factor in approximation when
			applying this reduction.

			\begin{definition}
			Every star center~${\vSc\in\starCenterSet}$ defines one star instance~${S_\vSc}$.
			It consists of 
			\begin{itemize}
				\item the set~${\F_\vSc}$ that we obtained when computing the star centers,
				\item the \emph{demand} $w_\vSc = \sum_{i\in\F_{\vSc}}\sum_{\vC\in\C}\xoComp_{i\vC}$\,,
				\item the \emph{strict budget}~${\bsmall_\vSc = \budOpening{\vSc}+\budDirectConnection{\vSc}+2\ell \budRelativeConnection{\vSc}}$\,, 
				\item and the \emph{relaxed budget}~${\bbig_\vSc = 2\budOpening{\vSc}+2\budDirectConnection{\vSc}+2\ell \budRelativeConnection{\vSc}}$\,. 
			\end{itemize}

			A star instance asks for a \emph{solution}, which is either a \emph{strict solution} or a \emph{\relaxedSolution}. 
			
		A \emph{strict solution} to the star instance~${S_\vSc}$ is an opening vector~${\ys}$ for the facilities in~${\F_\vSc}$ that satisfies the following constraints:

						\begin{alignat}{3}
						\label{con:star_demand}
						\sum_{i \in \F_\vSc} u_i \ysComp_i&\geq&~& w_\vSc\\
						\label{con:star_strict_budget}
						\sum_{i \in \F_\vSc} (f_i + \dist{i}{\vSc}u_i)\ysComp_i&\leq && \bsmall_\vSc
						\end{alignat}  
						
		A \emph{\relaxedSolution{}} to the star instance~${S_\vSc}$ is an opening vector~${\ys}$ for the facilities in~${\F_\vSc}$ that contains exactly one \supporting facility~${\iraised\in\F_\vSc}$ and satisfies Constraint~\eqref{con:star_demand} as well as the following constraint:
						\begin{alignat}{3}
							\label{con:star_budget} 
						f_\iraised\ysComp_\iraised + \dist{\iraised}{\vSc} w_\vSc &\leq &~&  \bbig_\vSc
						\end{alignat}  
		\end{definition}
		
We will use \relaxedSolutions{} in the context of almost integral opening vectors of volume at most~${1}$. In such cases, any restricted solution is also a relaxed solution. 
		
		We defined star instances only on star centers.
		Therefore, when referring to a star instance~${S_\vSc}$, or to one of its quantities~${w_\vSc}$,~${\bsmall_\vSc}$, or~${\bbig_\vSc}$, we will implicitly assume that~${\vSc\in\starCenterSet}$.
		We will also call~${w_\vSc}$ the \emph{demand} of the star center~${\vSc}$, opposed to the \lpDemand{} of~${\vSc}$ that refers to the unit demand of~${\vSc}$ in the context of~\LPfacLoc{}.
		
		Throughout the \kapitel, Constraints~\eqref{con:star_demand}--\eqref{con:star_budget} will refer to the constraints defined by a star instance, whereas Constraints~\eqref{lp:volume_k}--\eqref{lp:capacity_ge_total_demand} refer to the constraints defined by~\LPfacLoc{}.
		
			As a corollary of Lemma~\ref{lem:budget-equations}, we immediately get the following bounds on the total budgets of all star centers. 

			\begin{corollary}\label{cor:total_budget}
				The following holds:
				\begin{enumerate}
					\item~${\sum_{\vSc\in\starCenterSet} \bsmall_\vSc \le (1+2\ell)\opt}$.
					\item~${\sum_{\vSc\in\starCenterSet} \bbig_\vSc \le (2+2\ell)\opt}$.
				\end{enumerate}
			\end{corollary}

			In the following central lemma, we show how to compute a strict solution to a star instance of bounded volume.	
			\begin{lemma}\label{lem:compute_strict_solution}
				For any star instance~${S_{\vSc}}$, we can compute a strict solution~${\ys}$ to~${S_\vSc}$ 
				such that~${\vol{\ys} \leq \fvol{\yo}{\F_{\vSc}}}$ holds.
			\end{lemma}
			\begin{proof}
				The demand~${w_{\vSc}}$ is 
				given by~${\sum_{i\in\F_{\vSc}}\sum_{\vC\in\C}\xoComp_{i\vC}}$. 
				For each facility~${i\in\F_{\vSc}}$, 
				we charge~${\sum_{\vC\in\C}\xoComp_{i\vC}}$ units of~${w_\vSc}$ to~${i}$
				by
				opening it by an amount of~${\ysComp_i=(\sum_{\vC\in\C}\xoComp_{i\vC})/u_i}$.  
				Then we have~${\ysComp_i \le \yoComp_i}$ by Constraint~\eqref{lp:capacity_ge_total_demand}. 
				Thus~${\ys}$ is an opening vector satisfying 
				Constraint~\eqref{con:star_demand}
				and~${\vol{\ys} \leq \fvol{\yo}{\F_{\vSc}}}$.

				Now we prove that also Constraint~\eqref{con:star_strict_budget} is
				satisfied. 
				By Lemma~\ref{lem:distance-facility}.\ref{lem:distance-facility:fac-center}, for every client~${\vC\in\C}$ and every facility~${i\in\F_{\vSc}}$, it holds~${\dist{i}{\vSc}\le \dist{i}{\vC}+2\ell \dav{\vC}}$. So we have  
				\begin{alignat*}{2}  
				&&~&\sum_{i\in\F_{\vSc}}(f_i + \dist{i}{\vSc}u_i)\ysComp_i 
				\\&=&& \sum_{i\in\F_{\vSc}}\ysComp_i f_i   + \sum_{i\in\F_{\vSc}}\sum_{\vC\in\C} \dist{i}{\vSc}\xoComp_{i\vC}
				\\&\leq&& \sum_{i\in\F_{\vSc}} \yoComp_if_i + \sum_{i\in\F_{\vSc}}\sum_{\vC\in\C}(\dist{i}{\vC}+2\ell \dav{\vC})\xoComp_{i\vC}
				\\&=&& \bsmall_{\vSc} \formulaPunctuationSpace .
				\end{alignat*}
			\end{proof}
			
			Note that the strict solution provided by Lemma~\ref{lem:compute_strict_solution}
			may have volume strictly smaller than that of the underlying \starTail{}
			and consequently also smaller than~${1-1/\ell}$ (see Corollary~\ref{cor:minvolume-bundle}). 
			However, in our algorithm for~\PROBkfl with uniform capacities, we will be interested in solutions of volume at least~${1-1/\ell}$.  
			Therefore, we prove the following lemma.
			
			\begin{lemma} \label{lem:relaxedSolution} 
				For any a star instance~${S_\vSc}$ with uniform capacities that admits a strict solution~${\ys}$ with volume~${\vol{\ys}\le\min\{1,\fvol{\yo}{\F_\vSc}\}}$, 
				we can construct \aRelaxedSolution{}~${\ys'}$ to~${S_\vSc}$ 
				with~${\vol{\ys'}=\min\{1,\fvol{\yo}{\F_\vSc}\}}$ where the only \supporting facility~${\iraised\in\F_\vSc}$ satisfies~${\dist{\iraised}{\vSc}\le\ell\dav{\vSc}}$. 
			\end{lemma}
			\begin{proof}
				We first construct an opening vector~${\ys'}$ with the required volume and distance properties and then show that it is \aRelaxedSolution{} to~${S_\vSc}$. 
				
				Let~${R=\ell\dav{\vSc}}$. 
				Let~${\F_\vSc^R}$ be the set~${\{i \in \F \mid \dist{i}{\vSc}\le R \}}$ of facilities lying within the radius~${R}$ around~${\vSc}$.
				First, we show that~${\F_\vSc^R}$ is not empty, which is always implied by~${\fvol{\yo}{\F_\vSc^R}> 0}$.
				If~${\dav{\vSc}=0}$, then we have~${\fvol{\yo}{\F_\vSc^R}\ge 1}$ as discussed before. 
				Otherwise,~${R>0}$ and Lemma~\ref{lem:minvolume-R} assures
				\[{\fvol{\yo}{\F_\vSc^R}\ge 1-\dav{\vSc}/R=1-1/\ell \ge 1/2}\formulaPunctuationSpace;\] the last inequality follows from~${\ell\ge 2}$. 
				Consequently~${\F_\vSc^R}$ is not empty.
				
				Among the facilities in~${\F_\vSc^R}$, pick a facility~${\iraised}$ that minimizes the opening cost~${f_{\iraised}}$. 
				We define the following opening vector~${\ys'}$ for~${\F_\vSc}$.
				We set~${\ysComp'_\iraised=\min\{1,\fvol{\yo}{\F_\vSc}\}}$
				and, for every other facility~${i\in\F_\vSc}$, we set~${\ysComp'_i=0}$.
				Thus, we have the two properties~${\vol{\ys'}=\min\{1,\fvol{\yo}{\F_\vSc}\}}$ and~${\dist{\iraised}{\vSc}\le\ell\dav{\vSc}}$. 
				
				To this end, we show that~${\ys'}$ is \aRelaxedSolution{}.
				Since the solution~${\ys}$ fulfills Constraint~\eqref{con:star_demand}, we have
				\[{u \vol{\ys}= \sum_{i \in \F_\vSc} u \ysComp_i\geq w_\vSc}\formulaPunctuationSpace.\]
				Thus, also~${\ys'}$ fulfills Constraint~\eqref{con:star_demand}, 
				as
				\[{\vol{\ys'}=\min\{1,\fvol{\yo}{\F_\vSc}\}\ge \vol{\ys}}\formulaPunctuationSpace.\]
				
				To show Constraint~\eqref{con:star_budget},
				we first show~${\ysComp'_\iraised \cdot f_\iraised \le \budOpening{\vSc}}$ and then~${\dist{\iraised}{\vSc} w_\vSc \leq 2\budDirectConnection{\vSc}+2\ell \budRelativeConnection{\vSc}}$.
					
				Recall that~${\budOpening{\vSc}=\sum_{i\in\F}{f_i \yoComp_i}}$. 
				By using the minimality of~${\iraised}$, we infer
				\[{\budOpening{\vSc}\ge \sum_{i\in\F_\vSc^R}{f_i \yoComp_i} \ge f_{\iraised} \sum_{i\in\F_\vSc^R}{\yoComp_i} = f_\iraised \cdot \fvol{\yo}{\F_\vSc^R}}\formulaPunctuationSpace.\] 
				Above we have shown~${\fvol{\yo}{\F_\vSc^R} \ge 1/2}$. 
				Thus,~${\budOpening{\vSc} \ge f_\iraised/2}$ and our bound on~${\ysComp'_\iraised \cdot f_\iraised}$ holds as~${f_\iraised \ge \ysComp'_\iraised \cdot f_\iraised}$.
				
				Next, we show~${\dist{\iraised}{\vSc} w_\vSc \leq 2\budDirectConnection{\vSc}+2 \ell\budRelativeConnection{\vSc}}$. 	
				In fact, we will show the stronger inequality~${R w_\vSc \leq 2\budDirectConnection{\vSc}+2 \ell\budRelativeConnection{\vSc}}$, which implies the first one as~${\dist{\iraised}{\vSc}\le R}$.
				Let~${P}$ denote the set~${\F_\vSc \times \C}$ of pairs. 
				By our definitions,
				 \[{R w_\vSc = \sum_{(i,\vC)\in P} \xoComp_{i\vC}R }\] 
				and \[{2\budDirectConnection{\vSc}+2 \ell\budRelativeConnection{\vSc} = \sum_{(i,\vC)\in P} \xoComp_{i\vC} (2\dist{i}{\vC}+2\ell\dav{\vC})}\formulaPunctuationSpace.\] 
				Observe that each pair~${(i,\vC)\in P}$ contributes an amount~${\xoComp_{i\vC}R}$ to the left side of the inequality and an amount~${\xoComp_{i\vC} (2\dist{i}{\vC}+2\ell\dav{\vC})}$ to the right side of the inequality.
				Therefore, to show the inequality, it suffices to show
				\[{R \le 2\dist{i}{\vC}+2\ell\dav{\vC}}\] for each~${(i,\vC)\in P}$.
				Consider any~${(i,\vC)\in P}$. We distinguish two cases.
				
				In the first case,~${\dist{\vSc}{\vC}< R}$.  If~${|\starCenterSet|\ge 2}$, then, by Lemma~\ref{lem:center-distance}.\ref{lem:center-distances:center-center},~${R\le \dist{\vSc'}{\vSc}/2}$ for any~${\vSc'\in\starCenterSet}$ with~${\vSc'\ne \vSc}$.
				Hence, we have \[{\dist{\vSc}{\vC}< \frac{1}{2} \cdot \min\{\dist{\vSc}{\vSc'}\mid \vSc'\in\starCenterSet,\vSc'\ne \vSc\}}\formulaPunctuationSpace.\]
				Thus, independently of the size of~${|\starCenterSet|}$, Lemma~\ref{lem:center-distance}.\ref{lem:center-distances:dav} implies~${\dav{\vSc}\le2\dav{\vC}}$. Thus,~${R=\ell\dav{\vSc}\le 2\ell\dav{\vC}}$.
				
				In the second case,~${R\le\dist{\vSc}{\vC}}$. 
				By Lemma~\ref{lem:distance-facility}.\ref{lem:distance-facility:center-center},~${\dist{\vSc}{\vC}\le 2\dist{i}{\vC}+2\ell \dav{\vC}}$ and the claim follows.				
			\end{proof}
				
	The next lemma shows that we can always assume that 
	a strict solution to a star instance has at most two fractional facilities. 
	\begin{lemma}\label{lem:at_most_two_fractional}
		For any strict solution~${\ys}$ to a star instance~${S_\vSc}$, we can construct a strict solution~${\ys'}$ to~${S_\vSc}$ 
		which has at most \emph{two} fractional facilities and satisfies~${\vol{\ys'}\leq \vol{\ys}}$. 
	\end{lemma}
	\begin{proof}    
		Let~${S_\vSc}$ be any star instance with a strict solution~${\ys}$.
		Consider the LP for the star instance with Constraints~\eqref{con:star_demand},~\eqref{con:star_strict_budget} and the \emph{additional constraints}
		\[{0\le\ysComp_i} \textrm{ and } {\ysComp_i\le 1} \textrm{ for each } {i\in\F_\vSc}\] (which implicitly hold for opening vectors), and the objective 
		\[\textrm{minimize }{\sum_{i \in \F_\vSc} \ysComp_i}\formulaPunctuationSpace.\]
		Clearly~${\ys}$ is a feasible solution to this LP 
		with objective~${\vol{\ys}}$.  
		Now consider an optimal extreme point
		solution~${\ys'}$ to this LP. Of course,~${\vol{\ys'}\leq\vol{\ys}}$ is satisfied. 
		The number of variables in the LP is~${|\F_\vSc|}$.  Since~${\ys'}$ is an extreme point solution,  
		at least~${|\F_\vSc|}$ many of the LP constraints are tight.
		This means that at least~${|\F_\vSc|-2}$ of the additional constraints are tight.
		As for each~${i\in\F_\vSc}$, at most one of its two additional constraints are tight ($\ysComp_i=0$ and~${\ysComp_i=1}$ are mutually exclusive events), 
		there are at most two facilities in~${\F_\vSc}$ whose both additional constraints are not tight.
	\end{proof}
	
	For the case of uniform capacities, we can even assume that a strict solution contains at most one fractional facility. 
	
	\begin{lemma} \label{lem:almostIntegral} 
		For any strict solution~${\ys}$ to a star instance~${S_\vSc}$ with uniform capacities, we can construct an almost integral strict solution~${\ys'}$ to~${S_\vSc}$ 
		with~${\vol{\ys'}=\vol{\ys}}$.
	\end{lemma}
	\begin{proof} 
		Let~${S_\vSc}$ be any star instance with a strict solution~${\ys}$.
		The idea is to take the volume of~${\ys}$ and transfer it greedily to the facility~${i\in\F_\vSc}$
		that minimizes~${\dist{i}{\vSc}u + f_i}$ among all facilities that are not yet open.
		
		In more detail, we compute an opening vector~${\ys'}$ for~${\F_\vSc}$ of volume equal to~${\vol{\ys}}$  
		that minimizes~\[{\sum_{i \in \F_\vSc} (f_i + \dist{i}{\vSc}u_i)\ysComp'_i}\formulaPunctuationSpace.\]
		For this, we define for each facility~${i\in\F_\vSc}$ its \emph{weight} as~${\dist{i}{\vSc}u + f_i}$
		and order the facilities non-decreasingly by their weights.
		Then we create an opening vector~${\ys'}$ for~${\F_\vSc}$, where we set the first~${\floor{ \vol{\ys} }}$ facilities to~${1}$, the~${(\floor{ \vol{\ys} }+1)}$\thSuffix facility (if it exists) to~${\vol{\ys}-\floor{ \vol{\ys} }}$, and all remaining facilities to~${0}$.
		Thus, we have an almost integral opening vector with~${\vol{\ys'}=\vol{\ys}}$ and 
		\[\sum_{i \in \F_\vSc} (f_i + \dist{i}{\vSc}u_i)\ysComp'_i\le \sum_{i \in \F_\vSc} (f_i + \dist{i}{\vSc}u_i)\ysComp_i \le \bsmall_\vSc\formulaPunctuationSpace.\]
		Thus, Constraints~\eqref{con:star_demand} and~\eqref{con:star_strict_budget} are fulfilled and~${\ys'}$ is a strict solution to~${S_\vSc}$.
	\end{proof}

			For technical reasons, we
			prove the following two lemmas which are crucial for our later analysis.
			The statement of the first lemma is a little surprising. For any star instance with \aRelaxedSolution{} of volume~$z$, we can take a fraction~$z$ of its demand and send it to its closest star center. Independently of the distance, we can bound the connection cost with the budget of the star instance. We show this claim by observing that, in the~\LPfacLoc{} solution, at least a fraction~$z$ of the \lpDemand{} of the corresponding star center was sent to facilities that were very far away. The corresponding connection cost sufficiently sized the budget of the star instance. 
			The proof follows the ideas of the proof of Lemma~\ref{lem:relaxedSolution}.

			\begin{lemma}\label{lem:technical-transportation}
				Assume~${|\starCenterSet|\ge 2}$.
				Consider any star instance~${S_{\vSc}}$ with uniform capacities and any \relaxedSolution{} $\ys$ to~${S_{\vSc}}$ with~${\vol{\ys}=\fvol{\yo}{\F_{\vSc}}}$. Let~${\iraised}$ be the single \supporting facility in~${\ys}$ and let~${\vSc'\in\starCenterSet\elementsetminus{\vSc}}$ be the distinct star center closest to~${\vSc}$.
				We have \[{(1-\ysComp_{\iraised})w_{\vSc}\dist{\vSc}{\vSc'}\leq 4(\budDirectConnection{\vSc}+\ell\budRelativeConnection{\vSc})}\formulaPunctuationSpace.\]
			\end{lemma}
			\begin{proof} 
				Let~${P=\F_{\vSc} \times \C}$. 
				Then \[{(1-\ysComp_{\iraised})w_{\vSc}\dist{\vSc}{\vSc'} = \sum_{(i,\vC)\in P}\xoComp_{{i}\vC}(1-\ysComp_{\iraised})\dist{\vSc}{\vSc'}}\] and \[{4(\budDirectConnection{\vSc}+\ell\budRelativeConnection{\vSc}) = \sum_{({i},\vC)\in P} \xoComp_{{i}\vC} 4(\dist{i}{\vC}+\ell \dav{\vC})}\formulaPunctuationSpace.\] 
				To show the inequality, it suffices to show~${(1-\ysComp_\iraised)\dist{\vSc}{\vSc'} \le 4(\dist{i}{\vC}+\ell \dav{\vC})}$ for each~${(i,\vC)\in P}$.
				Consider any~${(i,\vC)\in P}$.
				Similarly, as in the proof of Lemma~\ref{lem:relaxedSolution},
				we distinguish the two cases
				where~${\dist{\vSc}{\vC}}$ is smaller or larger than a value~${R}$, respectively. 
				We set~${R = \dist{\vSc}{\vSc'}/2}$, thus~${R>0}$.
				
				First, assume~${\dist{\vSc}{\vC}\leq R}$. Then, by Lemma~\ref{lem:center-distance}.\ref{lem:center-distances:dav},~${\dav{\vSc}\le2\dav{\vC}}$.
				Note that every facility within radius~${R}$ around~${\vSc}$ belongs to~${\F_{\vSc}}$. 
				Applying Lemma~\ref{lem:minvolume-R}, we obtain that the volume of~${\F_{\vSc}}$ is at least~${1-\dav{\vSc}/R}$. 
				Given~${\ysComp_\iraised=\vol{\ys}}$ and our assumption~${\vol{\ys}=\fvol{\yo}{\F_{\vSc}}}$, 
				we have~${1-\ysComp_\iraised\le\dav{\vSc}/R}$.
				Hence,~${(1-\ysComp_\iraised)\dist{\vSc}{\vSc'}}$ is bounded from above by~${\dav{\vSc}/R\cdot 2R=2\dav{\vSc}\leq 4\dav{\vC}}$.

				Secondly, suppose~${R \le \dist{\vSc}{\vC}}$. By Lemma~\ref{lem:distance-facility}.\ref{lem:distance-facility:center-center}, $\dist{\vSc}{\vC}\le2\dist{\iraised}{\vC}+2\ell \dav{\vC}$.
				We therefore have 
				\begin{alignat*}{2}
				(1-\ysComp_\iraised)\dist{\vSc}{\vSc'}&\leq 2R\leq 2\dist{\vSc}{\vC}
				\\&\leq 4(\dist{\iraised}{\vC}+\ell \dav{\vC})\formulaPunctuationSpace.\qedHereInAlign
				\end{alignat*}
			\end{proof}
			
			The next lemma is essential to our decomposition of the problem to star instances. It will justify our assumption that
			 all the demand of the clients is accumulated in the star centers.
			
			\begin{lemma}\label{lem:move-to-bundle-centers}
				We can distribute 
				the \lpDemand{} of the clients among the star 
				centers such that each star 
				center~${\vSc}$ 
				receives precisely~${w_{\vSc}}$ units of demand and such that the total connection cost
				is at most~${(2\ell+2)\opt}$.  
			\end{lemma}
			\begin{proof}
				
				Let~${\vC\in\C}$ be an
				arbitrary client and~${i}$ be a facility lying in a star instance~${S_{\vSc}}$ for
				some~${\vSc\in\starCenterSet}$.  
				We ship precisely~${\xoComp_{i\vC}}$
				units of flow from~${\vC}$ to~${\vSc}$. 
				By Lemma~\ref{lem:distance-facility}.\ref{lem:distance-facility:center-center}, we upper bound the distance~${\dist{\vSc}{\vC}}$ by~${2\dist{i}{\vC} + 2\ell\dav{\vC}}$.
				 Performing this operation for any
				client-facility pair~${\vC\in\C}$,~${i\in\F_{\vSc}}$, we ensure that the star center~${\vSc}$ 
				collects~${\sum_{\vC\in\C}\sum_{i\in\F_{\vSc}}\xoComp_{i\vC}}$ units of demand, which is precisely~${w_{\vSc}}$.  
				
				The total cost of transporting~${w_{\vSc}}$ to the star center~${\vSc}$ is at most
				
				\begin{alignat*}{2}
					&&~&\sum_{\vC\in\C}\sum_{i\in\F_{\vSc}}\xoComp_{i\vC}(2\dist{i}{\vC}+2\ell  \dav{\vC})\\
					&=&&2\budDirectConnection{\vSc}+2 \ell\budRelativeConnection{\vSc}\formulaPunctuationSpace.
				\end{alignat*}
			
				Thus, the total cost of this flow over all star instances is bounded from above by
				\[2\sum_{\vSc\in\starCenterSet} \budDirectConnection{\vSc}+ 2\ell\sum_{\vSc\in\starCenterSet} \budRelativeConnection{\vSc}\]
				which is upper bounded by~${(2\ell+2)\opt}$ according to Lemma~\ref{lem:budget-equations}. 
				
			\end{proof}

						\subsection{The Dependent Rounding Approach}
						\label{app:dep-round}
						In our algorithm for hard-capacitated~\PROBkfl, we will apply
						the dependent rounding approach of Gandhi~et~al.~\cite{Gandhi2006} that is based on pipage rounding~\cite{Ageev2004}.
						For the sake of completeness, we give now an overview of this approach and state some properties that we will use.
						
						The dependent rounding procedure
						iteratively rounds a given vector~${\y\in [0,1]^N\!}$, for any number~${N}$ of components, 
						until all components are in~${\{0,1\}}$. 
						It works as follows.
						Suppose
						the current version of the rounded vector is~${\vv}$; initially,~${\vv}$ is~${\y}$.
						We have~${\vv=(\vvComp_1, \ldots, \vvComp_N)\in[0,1]^N}$.
						When we describe the random choice made in a step below, this choice
						is made independently of all such choices made thus far.
						If every component of~${\vv}$ lies in~${\{0,1\}}$, we are done, so let us assume that
						there is at least one component~${\vvComp_i\in(0,1)}$. The first (simple) case
						is that there is exactly
						one fractional component~${\vvComp_i}$; we round~${\vvComp_i}$ in the natural way---to~${1}$ 
						with probability~${\vvComp_i}$, and to~${0}$ with complementary probability~${1 - \vvComp_i}$; 
						letting~${V_i}$ denote the rounded version of~${\vvComp_i}$, we note
						that 
						\begin{equation*}
						\expct{V_i} = \vvComp_i
						\end{equation*}
						holds.
						This simple step is called a \emph{Type I iteration}, and it
						completes the rounding process. The remaining case is that of a
						\emph{Type II iteration}: there are at least two components of~${\vv}$ that
						lie in~${(0,1)}$. In this case, we choose two such components,~${\vvComp_i}$ and~${\vvComp_j}$, 
						in an arbitrary manner. 
						Let~${\varepsilon}$ and~${\delta}$ be the positive constants 
						such that:
						\begin{inlinelistAlph}~\item $\vvComp_i + \varepsilon$ and~${\vvComp_j - \varepsilon}$ lie in~${[0,1]}$, with at least
						one of these two quantities lying in~${\{0,1\}}$, 
						and~\item~${\vvComp_i - \delta}$ and~${\vvComp_j + \delta}$ lie in~${[0,1]}$, with at least
						one of these two quantities lying in~${\{0,1\}}$. 
						\end{inlinelistAlph}
						Such
						strictly-positive~${\varepsilon}$ and~${\delta}$ exist and are trivial to
						compute. We then update~${(\vvComp_i, \vvComp_j)}$ to a random pair~${(V_i, V_j)}$
						as follows:
						\begin{itemize}
							\item with probability~${\delta/(\varepsilon + \delta)}$, set~${(V_i, V_j) = (\vvComp_i + \varepsilon, ~\vvComp_j - \varepsilon)}$;
							\item with the complementary probability~${\varepsilon/(\varepsilon + \delta)}$, set~${(V_i, V_j) = (\vvComp_i - \delta, ~\vvComp_j + \delta)}$.
						\end{itemize}
						
						The main properties of \emph{Type II iteration} 
						that we need are:
						\begin{listRoman}
							\item~${\prob{V_i + V_j = \vvComp_i + \vvComp_j}  =  1}$;
							\item~${\expct{V_i} = \vvComp_i}$  and~${\expct{V_j} = \vvComp_j}$. 
						\end{listRoman}
						
						We iterate the iteration above until we obtain a rounded vector with at most one component in~${(0,1)}$. Since each
						iteration rounds at least one additional variable, we need at most~${N}$
						iterations.
						
						Note that the description above does not specify the order in which
						the elements are rounded. Observe that we may use a predefined laminar family
						of subsets to guide the rounding procedure. That is, we may first apply
						Type II iterations to elements of the smallest subsets, then continue
						applying Type II iterations for smallest subsets among those still containing more than
						one fractional entry, and eventually round the at most one remaining fractional entry
						with a Type I iteration. The following lemma shows that by executing the dependent rounding
						procedure in this manner, we almost preserve the sum of entries within each of the subsets
						of our laminar family. 
						
						For any subset~${L\subseteq\{1,\ldots,N\}}$ of component indices and any vector~${\vv\in [0,1]^N\!}$, with~${\vv=(\vvComp_1, \ldots, \vvComp_N) }$, we let~${\fvol{\vv}{L}}$ denote the total volume~${\sum_{j\in L} \vvComp_j}$ of the components of~${\vv}$ whose indices are in~${L}$. 
						\begin{lemma}\label{lem:dependent-rounding-sum-preservation}
							Let~${\y=(\yComp_1,\ldots,\yComp_N)\in [0,1]^N\!}$, and let~${L_1\subset \dots \subset L_l}$ be any laminar sequence of subsets 
							of the component indices~${\{1, \dots, N\}}$ of~${\y}$.
							In the order~${i=1,\dots,l}$, we repeatedly run the Type II iteration on the components of~${\y}$ that are given by~${L_i}$ until at most one 
							element of~${L_i}$ points to a fractional component.
							Let~${\ry\in [0,1]^N\!}$ be the resulting rounded vector. 
							
							The following holds: 
							For each~${i}$ with~${1 \le i \le l}$, 
there are at least~${\floor{\fvol{\y}{L_i}}}$ elements of~${L_i}$ that have value~${1}$ in~${\ry}$.
						\end{lemma}
						\begin{proof}
							Consider any~${i\in\{1,\dots,l\}}$ and the moment when the Type II iteration has just finished rounding elements in~${L_i}$. Let~${\vv\in [0,1]^N\!}$, with~${\vv=(\vvComp_1,\ldots,\vvComp_N)}$, be the current version of the rounded vector.
							Until this moment, only elements of~${L_i}$ had been considered, given the laminar property~${L_j \subseteq L_i}$ for every~${j}$ and~${i}$ with~${1\le j\le i}$. Thus, by the main properties of Type II iteration, the volume of~${L_i}$ preserved completely and we have~${\fvol{\vv}{L_i}=\fvol{\y}{L_i}}$. Recall that at most one element~${j}$ in~${L_i}$ has value~${\vvComp_j\in(0,1)}$. Thus there are  exactly~${\floor{\fvol{\y}{L_i}}}$ elements of~${L_i}$ that have value~${1}$ in~${\vv}$. Since the dependent rounding procedure never changes integral values, at least a portion~${\floor{\fvol{\y}{L_i}}}$ of the volume of~${L_j}$ remains until the end.
						\end{proof}

			\section{Algorithm for Uniform Hard-Capacitated~\PROBkfl}
			\label{sec:k_FL}
			
			In this section, we prove Theorem~\ref{thm:main_2_eps}. 
			For any positive~${\varepsilon}$, we show that we can 
			efficiently compute a solution to the hard-capacitated~\PROBkfl problem with capacity violation~${2+\varepsilon}$ and cost bounded by a factor~${\bigOh(1/\varepsilon^2)}$ of the cost of an optimum LP solution.
			
			In what follows, we fix the parameter~${\ell}$ to any value at least~${2}$ (see
			Section~\ref{sec:bundles}) 
			and let~${u}$ denote the uniform capacity of the facilities.
			Given~${\ell}$ and the optimum solution~${(\xo,\yo)}$ to~\LPfacLoc{}, we obtain the set~${\starCenterSet}$ of star centers and, for each star center~${\vSc\in\starCenterSet}$, its corresponding star instance~${S_\vSc}$, as described in Section~\ref{sec:bundles}. 
			By the following lemma, we will assume that there are at least two star centers. 
			
				\begin{lemma}
					If~${\starCenterSet}$ contains only one star center~${\vSc}$, we can efficiently compute an approximation with capacity violation~${2}$ and connection cost~${(4\ell+6)\opt}$ to the underlying~\PROBkfl instance.
				\end{lemma}
				\begin{proof}
	
				By Lemmas~\ref{lem:compute_strict_solution} and~\ref{lem:almostIntegral}, we compute an almost integral strict solution~${\ys}$ to~${S_\vSc}$. If its volume is at least~${1}$, we round any fractional facility down and take the rounded vector as our solution. Otherwise, by Lemma~\ref{lem:relaxedSolution}, we compute \aRelaxedSolution{}~${\ys'}$ and claim that the single \supporting facility is open. We take it as our solution. Lemma~\ref{lem:min_cost_flow} gives us an assignment of the clients to the open facilities.
				
				We show that our solution is feasible and satisfies the bound on the connection cost.
				First, assume~${\vol{\ys}\ge 1}$. Then our solution contains~${\floor{\vol{\ys}}}$ open facilities and we have~${\floor{\vol{\ys}}\ge 1}$. 
				If we scale the capacities up by~${2}$, the total capacity of our solution will be~${2u \cdot \floor{\ys}}$. 
				This is enough to serve the total demand, which is at most~${u \vol{\ys}}$; see Constraint~\eqref{con:star_demand}.
				Next, assume~${\vol{\ys}< 1}$. 
				By Lemma~\ref{lem:relaxedSolution}, our \relaxedSolution{}~${\ys'}$ has volume~${\vol{\ys'}=\min\{1,\fvol{\yo}{\F_\vSc}\}}$.
				Since~${\vSc}$ is the only star center, we have~${\F_\vSc=\F}$. 
				By Constraints~\eqref{lp:demand_of_client} and~\eqref{lp:opening_ge_demand},~${\fvol{\yo}{\F_\vSc}\ge 1}$. 
				Thus, we obtain~${\vol{\ys'}=\min\{1,\fvol{\yo}{\F_\vSc}\}=1}$. 
				Hence, there is a single open facility and it can serve the total demand without capacity violation.
				In the latter case, the cost of distributing~${w_\vSc}$ and opening the facilities is at most~${\bbig_\vSc}$ by Constraint~\eqref{con:star_budget}.
				In the former case, it is at most~${2\bsmall_\vSc}$ given the capacity blow-up of~${2}$ and Constraint~\eqref{con:star_strict_budget}.
				Hence, observing~${\bbig_\vSc\le 2\bsmall_\vSc}$ and using Corollary~\eqref{cor:total_budget}, the distribution and opening cost is bounded by~${(2\ell+4)\opt}$. 
				The cost to move all the demand of the clients to~${\vSc}$ is at most~${(2\ell+2)\opt}$ by Lemma~\ref{lem:move-to-bundle-centers}.
				The claim follows.
			\end{proof}
			\subsection{Constructing a Star Forest.}  
			In the first step of our algorithm
			for the uniform capacitated~\PROBkfl problem,
			we decomposed the instance into a set of star instances.
			Subsequently, we
			introduce the concept of a \emph{star tree} which imposes a suitable
			structure on the star instances and facilitates our description of the
			algorithm as the demand will be routed only along edges of the star
			tree.  We show that in order to obtain a bi-factor approximation
			algorithm for capacitated~\PROBkfl it suffices to
			appropriately ``round'' a star tree.

			A star tree is build up of so called~\emph{stars}, which are star instances with solutions. 
			We will distinguish between stars whose solutions have small or big volume and consider only those that have a lower bound on their volume.
			\begin{definition}\label{def:star}
				A~star~${(S_\vSc,\ys)}$ consists of a star instance~${S_\vSc}$ and an almost integral solution~${\ys}$ to~${S_\vSc}$.
				The \emph{demand of a star}~${(S_\vSc,\ys)}$ is the demand~${w_\vSc}$ of~${S_\vSc}$, and
				the \emph{volume of a star}~${(S_\vSc,\ys)}$ is the volume of~${\ys}$. 
				A star is~\emph{small} if it has volume is is at most~1 and at least~${1-1/\ell}$, and~${\ys}$ is \aRelaxedSolution{}. 
				A star is~\emph{big} if it is has volume greater than~${1}$ and~${\ys}$ is a strict solution. 
			\end{definition}
			
					Note that any small star has exactly one \supporting facility.  
						Also note that each star contains at most one fractional facility.

			\begin{definition}
				\label{def:star_tree_definition}
				A star tree is any rooted in-tree~${T}$ whose node set~${\Cs{T}}$ is a subset of~${\starCenterSet}$ and that is associated with following components satisfying the following properties.
				The components are  
				the set~${\Fs{T}}$ defined as~${\bigcup_{\vSc\in\Cs{T}} \F_\vSc}$,
				and 
				a metric~${\ds{T}}$ on~${\Cs{T}\cup\Fs{T}}$
				where~${\vSc\in\Cs{T}}$ and~${i\in\F_\vSc}$ imply~${\sdist{T}{i}{\vSc}=\dist{i}{\vSc}}$. 
				The properties are:
				\begin{listRoman}
					\item\label{prop:small-or-big-star} Each~${\vSc\in \Cs{T}}$ is associated with a small or a big star~${(S_\vSc,\ys)}$. 
					\item \label{prop:in-degree-bound} Each~${\vSc\in \Cs{T}}$ has indegree at least~${2}$ and the root~${r}$ has indegree exactly~${1}$.
					\item\label{prop:dec-length} For any consecutive edges~${(\vSc,\vSc'),(\vSc',\vSc'')}$, 
					we have~${\sdist{T}{\vSc}{\vSc'}\geq\sdist{T}{\vSc'}{\vSc''}}$.
					\item \label{prop:small-dist}
						Consider any~${\vSc,\vSc'\in\Cs{T}}$, $\vSc\ne\vSc'$, where~${\vSc}$ is associated with a small star. 
						Let~${\iraised}$ be the single \supporting facility of~${\vSc}$.
						We have~${\sdist{T}{\iraised}{\vSc} \le \sdist{T}{\vSc}{\vSc'}/2}$. 
					\item\label{prop:rerouting} 
						Consider any~${\vSc\in\Cs{T}}$ associated with a small star~${(S_{\vSc},\ys)}$.
						Let~${\iraised}$ be the single \supporting facility of~${\vSc}$.
						If~${\vSc}$ is the root of~${T}$, let~${\vSc'}$ be the single son of~${\vSc}$, otherwise, let~${\vSc'}$ be the father of~${\vSc}$.
					We have
					\[(1-\ysComp_{\iraised})w_{\vSc}\sdist{T}{\vSc}{\vSc'}\leq 8(\budDirectConnection{\vSc}+\ell \budRelativeConnection{\vSc})\formulaPunctuationSpace.\]

				\end{listRoman}
				The \emph{budget} $b(T)$ of the star tree~${T}$ is~${\sum_{\vSc \in \Cs{T}} \bbig_\vSc}$. 
				The \emph{volume} $\svol{T}$ of the star tree is given by the sum of the volumes of its stars.

			A \emph{solution to~${T}$ with capacity violation~${\gamma}$} is a set~${\F'\subseteq \Fs{T}}$ and an assignment~${(\assignForST_{i\vSc})_{i\in\F',\vSc\in\Cs{T}}}$ 
			satisfying the following constraints:
			\begin{eqnarray*}
				\sum_{i \in \F'} \assignForST_{i\vSc}&\geq& w_\vSc ~~~ \textrm{for each } \vSc\in\Cs{T} \\ 
				\sum_{\vSc\in\Cs{T}} \assignForST_{i\vSc}&\leq & \gamma u ~~ \textrm{for each }  i\in\F' \\ 
				\nonumber
				\assignForST_{i\vSc} & \geq & 0 ~~~~~\textrm{for each } i\in\F',\vSc\in \Cs{T} \formulaPunctuationSpace. 
			\end{eqnarray*}
			The \emph{cost} of the solution is \[{\sum_{i\in\F'} f_i + \sum_{i\in\F'}\sum_{\vSc\in\Cs{T}}\assignForST_{i\vSc}\sdist{T}{i}{\vSc}}\formulaPunctuationSpace.\]  
		\end{definition}
		
						Summarized, a star tree is a binary tree whose nodes are associated with small or big stars.
						The edge lengths are non-increasing towards the root which has degree~${1}$.
						Furthermore, every small star~${(S_\vSc,\ys)}$ has its single \supporting facility~${\iraised}$ relatively close to its center.
						Moreover, we can afford sending a fraction~${1-\ysComp_\iraised}$ of the demand of~${\vSc}$ to its closest neighbor. 
						We can interpret the last property also as follows:
						If the probability is~${1-\ysComp_\iraised}$ that~${\vSc}$ sends all its demand to its closest neighbor, and otherwise it does not send any demand, 
						then we can bound the expected connection cost by a constant multiple of its budget. 
						A star tree also defines a set of facilities which is the set of all facilities of its stars.
						By the construction of the \starTails{}, 
						each facility of a star tree appears in exactly one of stars of the tree.
						
\begin{definition}
A \emph{star forest} $H$		
is a collection of disjoint star trees. The \emph{budget}~${b(H)}$ and the \emph{volume}~${\svol{H}}$ 
of a star forest~${H}$ are given by the sum of budgets and the sum volumes of its star trees, respectively.  
A \emph{solution} to a star forest provides a solution to each of its star trees and
additionally satisfies the constraint that the total number of open facilities is no
more than~${\ceil{\svol{H}}}$. 
The \emph{cost} of a solution to a star forest~${H}$ is the total cost of the solutions that it provides to each star tree of~${H}$. 
\end{definition}

		    \paragraph{Creating Star Trees.}
		    The motivation for considering star trees is given by the following theorem.
		    It states that in order to get a constant-factor
		    approximation algorithm for uniform capacitated~\PROBkfl, 
		    it is sufficient to appropriately ``round'' a star forest.
		    \begin{theorem}
		    	\label{thm:central_theorem}
		    	If there is an efficient algorithm that computes for a given star
		    	forest~${H}$ a solution of cost at most~${c\cdot b(H)}$
		    	for some constant~${c>0}$ with capacity violation~${\gamma}$, 
		    	then there is a~${(2\ell+2)(c+1)}$-approximation algorithm for capacitated~\PROBkfl with
		    	capacity violation~${\gamma}$.
		    \end{theorem}
		    
			Before we prove the theorem, we first 
			describe how to build a star forest~${H}$ from a
			solution to~\LPfacLoc{} where the total volume~${\svol{H}}$ is bounded from above by~${k}$. 
			We begin by defining the node set and proving Property~\ref{prop:small-or-big-star}. 
			
			As the node set of the forest,
			we take the set~${\starCenterSet}$ of all star centers. 
			Since the \starTails{} defined by the star centers partition the set~${\F}$ of all facilities,
			our forest will contain all facilities of~${\F}$.
			Recall that each~${\vSc\in\starCenterSet}$ defines the star
			instance~${S_{\vSc}}$. 
			By Lemma~\ref{lem:compute_strict_solution}, we compute a strict solution~${\ys}$ to~${S_{\vSc}}$ with~${\vol{\ys}\le\fvol{\yo}{\F_{\vSc}}}$.
			If~${\vol{\ys}> 1}$, we apply Lemma~\ref{lem:almostIntegral} to obtain an almost integral strict solution~${\ys'}$ with~${\vol{\ys'}=\vol{\ys}}$. 
			Thus,~${(S_{\vSc},\ys')}$ is a big star and we associate it with~${\vSc}$. 
			Otherwise, if~${\vol{\ys}\le 1}$, then we have~${\vol{\ys}\le\min\{1,\fvol{\yo}{\F_{\vSc}}\}}$.
			We apply Lemma~\ref{lem:relaxedSolution} to obtain \aRelaxedSolution{}~${\ys'}$ with volume~${\vol{\ys'}=\min\{1,\fvol{\yo}{\F_{\vSc}}\}}$ where the only \supporting facility~${\iraised\in\F_{\vSc}}$ satisfies~${\dist{\iraised}{\vSc}\le\ell\dav{\vSc}}$. 
			By Corollary~\ref{cor:minvolume-bundle},~${\fvol{\yo}{\F_{\vSc}}\ge1-1/\ell}$ and so~${\vol{\ys'}\ge1-1/\ell}$. Consequently~${(S_\vSc,\ys')}$ is a small star. We associate it with~${\vSc}$ and conclude that Property~\ref{prop:small-or-big-star} holds.
			
			Note that every star~${(S_\vSc,\ys)}$ constructed above has a volume at most~${\fvol{\yo}{\F_\vSc}}$. 
			Thus, the total volume in the star forest that we are constructing is at most~${k}$, in particular,~${\ceil{\vol{H}}\le k}$. 
			In the following, 
			we describe how to connect the nodes together and how to compute a metric in each obtained star tree so that the remaining properties
			are satisfied.

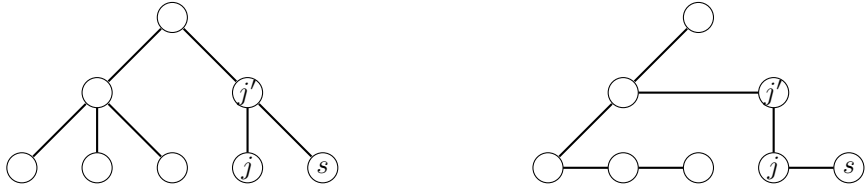
\begin{figure}
	\centering
	\hfill
	\begin{tikzpicture}[scale=1]
	\draw [fill=white] (-3.15,3.15) circle (0.2);
	
	\draw [thick] (-4,2.3) -- (-3.3,3); 
	\draw [thick] (-2.3,2.3) -- (-3,3); 
	
	\draw [fill=white] (-4.15,2.15) circle (0.2);
	\draw [fill=white] (-2.15,2.15) circle (0.2);
	
	\draw [thick] (-5,1.3) -- (-4.3,2); 
	\draw [thick] (-4.15,1.35) -- (-4.15,1.95); 
	\draw [thick] (-3.3,1.3) -- (-4,2); 
	
	\draw [thick] (-2.15,1.35) -- (-2.15,1.95); 
	\draw [thick] (-1.3,1.3) -- (-2,2); 
	
	\draw [fill=white] (-5.15,1.15) circle (0.2);
	\draw [fill=white] (-4.15,1.15) circle (0.2);
	\draw [fill=white] (-3.15,1.15) circle (0.2);
	\draw [fill=white] (-2.15,1.15) circle (0.2);
	\draw [fill=white] (-1.15,1.15) circle (0.2);
	
	\path (-2.15,2.15) node {$\vC'$};
	\path (-2.15,1.15) node {$\vC$};
	\path (-1.15,1.15) node {$\vSc$};
	
	\end{tikzpicture}
	\hfill
	\begin{tikzpicture}[scale=1]
	\draw [fill=white] (2.45,3.15) circle (0.2);
	
	\draw [thick] (1.6,2.3) -- (2.3,3); 
	\draw [thick] (3.25,2.15) -- (1.65,2.15); 
	
	\draw [fill=white] (1.45,2.15) circle (0.2);
	\draw [fill=white] (3.45,2.15) circle (0.2);
	
	\draw [thick] (0.6,1.3) -- (1.3,2); 
	\draw [thick] (1.25,1.15) -- (0.65,1.15); 
	\draw [thick] (2.25,1.15) -- (1.65,1.15); 
	
	\draw [thick] (3.45,1.35) -- (3.45,1.95); 
	\draw [thick] (4.25,1.15) -- (3.65,1.15); 
	
	\draw [fill=white] (0.45,1.15) circle (0.2);
	\draw [fill=white] (1.45,1.15) circle (0.2);
	\draw [fill=white] (2.45,1.15) circle (0.2);
	\draw [fill=white] (3.45,1.15) circle (0.2);
	\draw [fill=white] (4.45,1.15) circle (0.2);
	
	\path (3.45,2.15) node {$\vC'$};
	\path (3.45,1.15) node {$\vC$};
	\path (4.45,1.15) node {$\vSc$};
	\end{tikzpicture}
	\hfill\,
	\caption{Making a tree binary: The left side shows one of the trees returned by Procedure \ShortTrees($\starCenterSet$,~${\emptyset}$); the right side shows the corresponding tree after modification by Procedure \BinaryTrees($G$). 
		Observe that the father~${\vC'}$ of~${\vSc}$ has been replaced by the left brother~${\vC}$ of~${\vSc}$.}
	\label{fig:trees_algorithms} 
	
\end{figure}
				As a first step, we build a directed forest~${G}$ on the node set~${\starCenterSet}$
				by running Procedure \ShortTrees($\starCenterSet$,~${\emptyset}$) (see Section~\ref{sec:bundles}).
				By Lemma~\ref{lem:forest-intrees}, its components are all in-trees, which we will call \emph{short center trees}.				
			We cannot take the short center trees as our star trees, as the indegrees of their nodes may be unbounded.  Therefore,
			we change the structure of each short center tree to obtain a
			\emph{binary center tree}, in which the indegree of each node is at most~${2}$. 
			\Fig~\ref{fig:trees_algorithms} depicts this process that we describe below.
			By showing that all remaining properties of Definition~\ref{def:star_tree_definition} are fulfilled, we will prove that the constructed binary center tree is a star tree.
			
			Consider the Procedure \BinaryTrees$(G)$
			that takes as input the forest~${G}$ of
			short center trees returned by \ShortTrees($\starCenterSet$,~${\emptyset}$). 
			Each short center tree~${T}$ in~${G}$ is separately modified as follows. For each node~${\vC \in V(T)}$, we
			sort all incoming edges of~${\vC}$ from left to right by non-decreasing length.
			In the case that~${\vC}$ is the root of~${T}$, the leftmost edge is set to be the incoming edge from its closest son. (Note that there might exist multiple incoming edges with the same smallest length, therefore, we use the uniqueness of closeness; see Definition~\ref{def:closeness}.)
			Having defined the order, we remove all incoming edges of~${\vC}$
			except the leftmost one. 
			In the next step, the procedure adds
			for each son~${\vSc}$ of~${\vC}$ an edge from~${\vSc}$ to its left brother (if there
			exists one).  Note that no edge is added to the leftmost son of~${\vC}$.  The resulting forest of binary center trees is denoted by~${H}$.
		\begin{procedure}
			\caption{\BinaryTrees($G$)} $H = \emptyset$\; 
			\ForEach{short tree~${T}$ in~${G}$} {
				\ForEach{node~${\vC \in V(T)}$} {
					sort all sons of~${\vC}$ from left to right by
					non-decreasing distance to~${\vC}$\;
					\If{$\vC$ is the root of~${T}$}{place the son of~${\vC}$ closest to~${\vC}$ on the leftmost position in the ordering\;}
					remove all incoming edges of~${\vC}$ from~${T}$ except the leftmost one\;
					add a directed edge from each
					son of~${\vC}$ to its left brother (if there exists one)\;
				}
				add~${T}$ to~${H}$\;
			}
			\Return~${H}$\;	
		\end{procedure}
		
			Every binary center tree satisfies the following useful property. 
			\begin{lemma}\label{lem:root-cycles}
				The root of a binary center tree has exactly one son and its son is its closest star center.
			\end{lemma}
			\begin{proof}
				By Lemma~\ref{lem:root-son-distance}, the node closest to the root is one of its sons in the short center tree.
				During the construction of the binary center tree, this node becomes the only son of the root.
			\end{proof}
						
			Consider any binary center tree~${T'}$. Let~${\Cs{T'}}$ be its node set and let~${\Fs{T'}}$ be the set~${\bigcup_{\vSc\in\Cs{T'}} \F_\vSc}$ of all facilities in~${T'}$.
			Let~${T}$ be the short center tree from which~${T'}$ was derived.
			We show that~${T'}$ is a star tree. 
			For technical reasons, we first define a new metric~${\ds{T'}}$ on~${\Fs{T'} \cup \Cs{T'}}$ for~${T'}$, before we consider the remaining properties.

			The metric~${\ds{T'}}$ will be a tree metric where the underlying tree is the binary center tree~${T'}$ (for this purpose considered undirected) with an additional edge for each facility that connects it to its star center. 
			Recall that, in a tree metric, the distance between any two nodes is the weight of the path connecting the two nodes.
			For each additional edge between a facility~${i}$ and its star center~${\vSc}$, we set its weight to~${\dist{i}{\vSc}}$; hence~${\sdist{T'}{i}{\vSc}=\dist{i}{\vSc}}$ and we fulfill the requirement imposed by Definition~\ref{def:star_tree_definition} on~${\ds{T'}}$. 
			For each edge~${(\vSc,\vC)}$ of the binary center tree~${T'}$, we set its weight to~${2\dist{\vSc}{\vC'}}$, where~${\vC'}$ is the closest node in~${\starCenterSet\elementsetminus{\vSc}}$ to~${\vSc}$, 
			or stated equivalently, where~${\vC'}$ is the father of~${\vSc}$ in the short center~${T}$; see \fig~\ref{fig:trees_algorithms}. Hence,~${\sdist{T'}{\vSc}{\vC}=2\dist{\vSc}{\vC'}}$.
			The new metric~${\ds{T'}}$ is never smaller than the underlying metric~${\metricName}$, as the next lemma shows.
			\begin{lemma}\label{lem:ds-metric}
				The metric~${\ds{T'}}$ satisfies~${\sdist{T'}{\vSc}{\vC}\ge\dist{\vSc}{\vC}}$ for every~${\vSc,\vC\in\Cs{T'}\cup\Fs{T'}}$.
			\end{lemma}
			\begin{proof}
				It suffices to show the claim for each edge of the tree that underlies the metric.
				For each edge between a facility and its star center, the claim directly holds by definition. 
				Thus, consider any edge~${(\vSc, \vC)\in E(T')}$. 
				By definition of~${\ds{T'}}$, we have~${\sdist{T'}{\vSc}{\vC}=2\dist{\vSc}{\vC'}}$, where~${\vC'}$ is the father of~${\vSc}$ in the short center tree~${T}$.
				To prove the claim, we show~${\dist{\vSc}{\vC}\leq 2 \dist{\vSc}{\vC'}}$.
			    There are two cases: Either~${\vC}$
								is the father of~${\vSc}$ in~${T}$, or it is the left brother of~${\vSc}$ in~${T}$. 
								The first case is trivial as~${\dist{\vSc}{\vC} \leq 2 \cdot \dist{\vSc}{\vC}}$.  In the second case, the node~${\vC'}$ is the common father of~${\vSc}$ and~${\vC}$ in~${T}$; see \fig~\ref{fig:trees_algorithms}. 
								Since~${\vC}$ is the left brother of~${\vSc}$ in~${T}$, we have~${\dist{\vC}{\vC'} \leq \dist{\vSc}{\vC'}}$.  
								Hence~${\dist{\vSc}{\vC}\leq \dist{\vSc}{\vC'} + \dist{\vC'}{\vC} \leq 2 \cdot \dist{\vSc}{\vC'}}$.
			\end{proof}

			Now that the metric~${\ds{T'}}$ is defined, we prove that~${T'}$ satisfies all properties of a star tree.
			\begin{lemma}\label{lem:binary-center-tree-is-star-tree}
				The binary center tree~${T'}$ with the metric~${\ds{T'}}$ is a star tree. 
			\end{lemma}
			\begin{proof}
				Recall that our construction of the node set and their associated stars already implies Property~\ref{prop:small-or-big-star} 
				for~${T'}$. 
					
				The next two properties 
				are related to the tree structure of a binary center tree. We provide bounds on the degree of the nodes and show that edge lengths towards the root are non-increasing. 
					
				We show Property~\ref{prop:in-degree-bound}.
				Any node~${\vSc}$ in~${T'}$ has at most two incoming edges: One from
				its closest son in~${T}$ and one from its right brother in~${T}$.
				By Lemma~\ref{lem:root-cycles}, the root has exactly one son in~${T'}$.
					
				Next, we show Property~\ref{prop:dec-length}.
				Let~${(\vSc,\vSc')}$ and~${(\vSc',\vSc'')}$ be any consecutive edges in~${T'}$.  
				We have to prove~${\sdist{T'}{\vSc}{\vSc'}\geq \sdist{T'}{\vSc'}{\vSc''}}$.
				In the tree~${T}$, let~${\vC}$ be the father of~${\vSc}$, and  let~${\vC'}$ be the father of~${\vSc'}$.
				Recall that, by the definition of~${\ds{T'}}$, we have~${\sdist{T'}{\vSc}{\vSc'}=2\dist{\vSc}{\vC}}$ and~${\sdist{T'}{\vSc'}{\vSc''}=2\dist{\vSc'}{\vC'}}$. 
				Therefore, it suffices to show~${\dist{\vSc}{\vC}\geq \dist{\vSc'}{\vC'}}$.  
				If~${\vSc'}$ is the father of~${\vSc}$ in~${T}$, then~${\vSc'=\vC}$ and the claim holds by
				Lemma~\ref{decrease_edge_length}.  If~${\vSc'}$ is the left brother of~${\vSc}$ in~${T}$, then~${\vC'=\vC}$ and the claim holds by the construction of binary center trees.

				The remaining two properties 
				are related to small stars.
				We show how to bound the distance of a node to its only \supporting facility and how to bound the cost of transporting its demand to the next node.
				
			    Consider any small star~${(S_{\vSc},\ys)}$ with~${\vSc\in\Cs{T'}}$. 
				Let~${\iraised}$ be the only \supporting facility in~${(S_{\vSc},\ys)}$.
				We show Property~\ref{prop:small-dist}. 
				Let~${\vSc'\in\Cs{T'}\elementsetminus{\vSc}}$ be any node distinct to~${\vSc}$. 
				By the definition of~${\ds{T'}}$,~${\sdist{T'}{\iraised}{\vSc}=\dist{\iraised}{\vSc}}$.
				Recall that our construction of small stars guarantees that~${\iraised}$ satisfies~${\dist{\iraised}{\vSc}\le\ell\dav{\vSc}}$. 
				By Lemmas~\ref{lem:center-distance}.\ref{lem:center-distances:center-center} and~\ref{lem:ds-metric},
				we obtain~${\ell\dav{\vSc} < \dist{\vSc'}{\vSc}/2 \le \sdist{T'}{\vSc'}{\vSc}/2}$.				
				
				Next, we show Property~\ref{prop:rerouting}. 
				If~${\vol{\ys}=1}$, then we are done. Therefore, we assume~${\vol{\ys}<1}$ which implies~${\vol{\ys}=\fvol{\yo}{\F_{\vSc}}}$ by our construction of small stars.
				Let~${\vSc'\in\starCenterSet\elementsetminus{\vSc}}$ be the star center distinct from~${\vSc}$ that is closest to~${\vSc}$. 
				If~${\vSc}$ is the root of~${T'}$, then, by Lemma~\ref{lem:root-cycles},~${\vSc'}$ is the son of~${\vSc}$.
				Thus, by the definition of~${\ds{T'}}$,~${\sdist{T'}{\vSc'}{\vSc}=2\dist{\vSc'}{\vSc}}$
				and Lemma~\ref{lem:technical-transportation} implies the claim.
				If~${\vSc}$ is not the root, let~${\vC}$ be the father of~${\vSc}$.
				Then, by the definition of~${\ds{T'}}$,~${\sdist{T'}{\vSc}{\vC}=2\dist{\vSc}{\vSc'}}$.
			    Therefore, it suffices to show the inequality~${(1-\ysComp_{\iraised})w_{\vSc}\dist{\vSc}{\vSc'}\leq 4(\budDirectConnection{\vSc}+\ell\budRelativeConnection{\vSc})}$ to prove the claim. 
				The inequality holds by Lemma~\ref{lem:technical-transportation}.

			\end{proof}
			
			We have now shown that all properties of star trees as required in
			Definition~\ref{def:star_tree_definition} are actually satisfied by the binary center tree~${T'}$. 
			This fact implies that all our binary center trees are star trees, and, hence, that our constructed forest~${H}$ is in fact a star forest. 
			Thus, we are ready to prove Theorem~\ref{thm:central_theorem}.
				\begin{proof}[Proof of Theorem~\ref{thm:central_theorem}] 
					Given a capacitated~\PROBkfl instance, we compute the corresponding star forest~${H}$ and use the black box algorithm to obtain a solution to~${H}$ of cost~${c\cdot b(H)}$ and capacity violation~${\gamma}$. Note that the solution opened at most~${\ceil{\svol{H}}\le k}$ facilities.
					Then, by solving a
					minimum-cost flow problem (see Lemma~\ref{lem:min_cost_flow}), we efficiently compute an optimal assignment of
					the clients to facilities opened by the solution. 
					
					To bound the cost of the solution, we give a suboptimal
					fractional flow of demand that uses the edges in the star forest and that
					satisfies the claimed cost bound.  The flow is constructed in two
					steps.  First the demand of the clients is transported to the 
					star centers
					so that each node~${{\vSc}\in\starCenterSet}$ collects precisely~${w_{\vSc}}$ units of demand.  By Lemma~\ref{lem:move-to-bundle-centers}, this can be accomplished at cost at most~${(2+2\ell)\cdot\opt}$.  To
					transport the demand collected at the star centers to
					the actual facilities, we use the assignment provided by the
					solution to the star forest. By definition, this assignment
					transports for each~${{\vSc}\in\starCenterSet}$ precisely~${w_{\vSc}}$ units of demand to the facilities
					opened by the solution. 
					The cost of this assignment, together with the costs for opening the facilities, 
					is~${c\cdot b(H)=c\cdot\sum_{\vSc\in\starCenterSet} \bbig_\vSc \leq c\cdot (2\ell+2)\opt}$   
					by Corollary~\ref{cor:total_budget}.
					Altogether, we obtain~${(c+1)\cdot (2\ell+2)\opt}$ as an upper bound of the solution cost.

				\end{proof}

			\subsection{Solving a Star Forest.} 
			\label{sec:twoApproAlgo}
			We now show how to solve any given star forest~${H}$. We describe a deterministic rounding procedure which is tuned to minimize the capacity violation while allowing a large  
			connection and opening cost that is still bounded by a constant multiple of the budget~${b(H)}$.  
			Together with Theorem~\ref{thm:central_theorem}, this will imply, for any positive~${\varepsilon}$, an~${\bigOh(1/\varepsilon^2)}$-approximation algorithm 
			for the uniform hard-capacitated~\PROBkfl problem with capacity violation~${2+\varepsilon}$. 
						
			In the first step, the algorithm 
		    forms groups of~${\ell \geq 2}$ nodes in each star tree. In the next step, at the cost of loosing some accuracy with distances, we simplify the graph structure within each of the groups. Eventually, we use a dependent rounding routine to decide the actual openings of facilities, 
			and argue that there is sufficient capacity open up every tree to serve all demand coming from below, and, hence, on every star tree there exists a feasible routing of its total demand.
			
			To the end of this section, let~${\ys}$ denote the union of the solutions of all stars in the star forest.
			For the sake of easier presentation, we will refer to each star~${(S_\vSc,\ys')}$ just by its instance~${S_\vSc}$ and use~${\ys}$ 
			when referring to the opening values of its facilities.
			Also to the end of this section, we fix any star tree~${T}$ of our forest with a root~${r}$. Let~${\Cs{T}}$ be the set of its nodes,~${\Fs{T}}$ the set of its facilities, and~${\ds{T}}$ its metric on~${\Cs{T}\cup\Fs{T}}$.

			\paragraph{Building groups.}
			The nodes of the star tree~${T}$ will be grouped by a top-down greedy procedure starting from the root~${r}$; see \fig~\ref{fig:groups:grouping}. 
			When forming a new group, a single node~${\vSc}$ (having all its descendants yet not grouped)
			will be selected as a root of the new group. Then new nodes will be added to the group in a greedy fashion until either the group has reached the size of~${\ell}$ nodes, or all descendants of~${\vSc}$ are already included. The greedy choice of the next node to include will be to take one which is connected to the already included nodes by a cheapest tree edge. When a group is complete, we exclude the selected nodes from participating in the later formed groups. As long as not all nodes of the tree are grouped, we select a top-most one~${\vSc}$ and build a group~${G_\vSc}$ rooted at~${\vSc}$. 
			\begin{definition}
			  	A group~${G_{\vSc}}$ is a \emph{child} of a group~${G_\vC}$ and~${G_\vC}$ is a \emph{parent} of~${G_{\vSc}}$ if there is a directed edge in~${T}$ from~${\vSc}$ to a node in~${G_\vC}$. The group~${G_r}$ is the \emph{root group} and every other group is a \emph{non-root group}.
			\end{definition}
			
			\begin{observation}
				If~${G_\vC}$ has at least one child, then it contains exactly~${\ell}$ nodes, otherwise (if it has no children) $G_\vC$ may have less nodes. 
				Moreover, each group has at most~${\ell + 1}$ children.
			\end{observation}
			The next lemma 
			is implied by Property~\ref{prop:dec-length}
			and the way in which the algorithm selects nodes to a group.
			\begin{lemma} \label{lem:group_edge_monotonicity}
				Consider any group~${G_\vC}$ that has a child group~${G_{\vSc}}$. Let~${e_{\vSc}\in E(T)}$ be the edge from~${\vSc}$ (the root of~${G_{\vSc}}$) to its father in~${G_\vC}$. For any edge~${e}$ in~${G_\vC}$ and any edge~${e'}$ in~${G_{\vSc}}$, we have~${\edgeSdist{T}{e} \leq \edgeSdist{T}{e_{\vSc}} \leq \edgeSdist{T}{e'}}$.
			\end{lemma}
			\begin{figure}[htb]
				\captionsetup[subfloat]{captionskip=5pt}
				\centering
				\subfloat[Building groups: Edges highlighted in gray form a new group~${G_{r}}$ with the root node~${{r}}$; nodes~${{\vSc}}$ and~${{\vC}}$ are roots of groups~${G_{\vSc}}$ and~${G_{\vC}}$, respectively. $G_{r}$ is parent group of~${G_{\vSc}}$ and~${G_{\vC}}$. The numbers indicate a possible order in which the nodes have been added to~${G_r}$. \label{fig:groups:grouping}]{
				\quad\quad\enspace\begin{tikzpicture}
				\draw [lightgray, line width=8pt,cap=round] (-0.85,0.65) -- (-0.15,1.35);
				\draw [lightgray, line width=8pt,cap=round] (0.85,0.65) -- (0.15,1.35);
				\draw [lightgray, line width=8pt,cap=round] (-1.45,-0.3) -- (-1.15,0.35);
				\draw [lightgray, line width=8pt,cap=round] (0.6,-0.3) -- (0.9,0.3);
				\draw [lightgray, line width=8pt,cap=round] (-1.6,-0.7) -- (-1.9,-1.3);
				\draw [lightgray, line width=8pt,cap=round] (-1.4,-0.66) -- (-1.1,-1.3);

				\draw [fill=white] (0,1.5) circle (0.2);
				
				\path (0,1.5) node {$\boldsymbol{r}$};

				\draw [thick] (-0.85,0.65) -- (-0.15,1.35);
				\draw [thick] (0.85,0.65) -- (0.15,1.35);
				
				\draw [fill=white] (-1,0.5) circle (0.2);
				\draw [fill=white] (1,0.5) circle (0.2);
				
				\draw [thick] (-1.45,-0.3) -- (-1.15,0.35);
				\draw [thick] (-0.6,-0.3) -- (-0.9,0.3); 
				
				\draw [thick] (0.6,-0.3) -- (0.9,0.3);
				\draw [thick] (1.45,-0.3) -- (1.1,0.35); 
				
				\draw [fill=white] (-1.5,-0.5) circle (0.2);
				\draw [fill=white] (-0.5,-0.5) circle (0.2);
				\path (-0.5,-0.5) node {$\boldsymbol{\vSc}$};
				
				\draw [fill=white] (0.5,-0.5) circle (0.2);
				\draw [fill=white] (1.5,-0.5) circle (0.2);
				\path (1.5,-0.5) node {$\mathbf{\vC}$};
				
				\draw [thick] (-1.6,-0.7) -- (-1.9,-1.3);
				\draw [thick] (-1.4,-0.66) -- (-1.1,-1.3);
				
				\draw [thick] (-0.4,-0.7) -- (-0.1,-1.3); 
				\draw [thick] (1.45,-0.7) -- (1.1,-1.3); 
				\draw [thick] (1.6,-0.7) -- (1.9,-1.3); 
				
				\draw [fill=white] (-2,-1.5) circle (0.2);
				\draw [fill=white] (-1,-1.5) circle (0.2);
				
				\draw [fill=white] (0,-1.5) circle (0.2);
				\draw [fill=white] (1,-1.5) circle (0.2);
				\draw [fill=white] (2,-1.5) circle (0.2);
				
				\path (-1,0.5) node {$1$};
				\path (-1.5,-0.5) node {$2$};
				\path (1,0.5) node {$3$};
				\path (-2,-1.5) node {$4$};
				\path (0.5,-0.5) node {$5$};
				\path (-1,-1.5) node {$6$};
				
				\end{tikzpicture}\quad\quad\enspace
			}\hfill
				\subfloat[The group~${G_r}$ after the modification step: The group is now a chain and nodes are ordered from top to bottom in the order they were added to the group, that is, with non-decreasing distances to their fathers in the star tree. \label{fig:groups:modifying}]{
					\quad\begin{tikzpicture}
				
					\foreach \nummer [evaluate=\nummer as \x using 0.75*\nummer] [evaluate=\nummer as \y using -\nummer*0.5] in {1,2,3,4,5,6}{
						\draw [lightgray, line width=8pt,cap=round] (\x-0.75,\y+0.5) -- (\x,\y);
						\draw [thick] (\x-0.75,\y+0.5) -- (\x,\y);
					}
					\draw [fill=white] (0,0) circle (0.2);
					\path (0,0) node {$\boldsymbol{r}$};
					\foreach \nummer [evaluate=\nummer as \x using 0.75*\nummer] [evaluate=\nummer as \y using -\nummer*0.5]  in {1,2,3,4,5,6}{
						\draw [fill=white] (\x,\y) circle (0.2);
						\path (\x,\y) node {$\nummer$};
					}
					
					\end{tikzpicture}\quad
				}%
				\caption{Building groups and the modification step.}  
				\label{fig:groups}
			\end{figure}
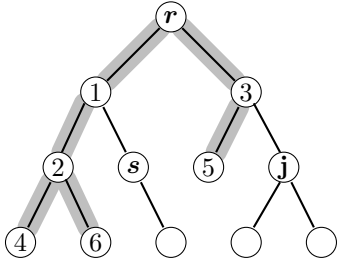
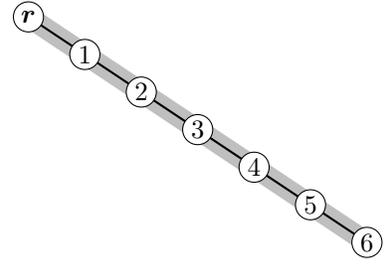
			
			\paragraph{Group modification.} 
			To facilitate rounding of facility openings within groups, we will modify the tree structure within groups to obtain a new in-tree~${T'}$ from the initial star tree~${T}$.
			The partition of nodes into groups will stay unchanged and the parent-child relation between groups will also be preserved. The modification within a single group is as follows. 
			
			Consider a group of nodes~${G_\vSc}$ and the order in which the nodes were added to the group by the greedy procedure. In the modified tree~${T'}$, the
			group~${G_\vSc}$ will form a chain graph directed towards its root~${\vSc}$, with the nodes closer to~${\vSc}$ being those selected earlier by the group forming algorithm; see \fig~\ref{fig:groups:modifying}.
			Finally, for any group~${G_{\vSc}}$ which is a child of a group~${G_\vC}$, let the edge outgoing from~${\vSc}$ point to the lowest vertex in~${G_\vC}$ in~${T'}$.
			
			Clearly, such modification of the tree structure may interfere with routing demand along edges of the used tree. Nevertheless, we will argue that we may bound
			this influence to only a constant multiplicative growth in the routing distance. 
			
			Recall that the lengths of edges of~${T}$ were monotone non-increasing on any directed path towards the root node~${r}$. We will no longer have this property in~${T'}$,
			but we will now exploit the monotonicity of~${\ds{T}}$ on edges of directed paths in~${T}$ to bound distances on~${T'}$.

			\begin{lemma}\label{lem:routing-in-group}
				Let~${\vSc}$ and~${\vC}$ be any nodes of the same group such that~${\vC}$ lies above~${\vSc}$ in~${T'}$.
				Let~${\vC'}$ be the father of~${\vSc}$ in~${T}$.	
				It holds~${\sdist{T}{\vSc}{\vC}\leq(\ell-1)\cdot\sdist{T}{\vSc}{\vC'}}$.
			\end{lemma}
			\begin{proof}
				Since~${\vC}$ lies above~${\vSc}$ 
				in the same group as~${\vSc}$ in~${T'}$, we
				have that~${\vSc}$ was added later to this group than~${\vC}$.  Hence every
				edge on the path (ignoring edge directions) from~${\vSc}$ to~${\vC}$ in~${T}$
				has length at most~${\sdist{T}{\vSc}{\vC'}}$; see \fig~\ref{fig:groups} and consider~${\vSc=5}$,~${\vC=2}$ and~${\vC'=3}$.  
				Since no more than~${\ell-1}$ edges lie
				on this path and since~${\ds{T}}$ is a metric, the claim follows.
			\end{proof}
			
			\begin{lemma}\label{lem:routing-between-group}
					Let~${\vSc}$ be the root of any non-root group, let~${\vC'}$ be its father in~${T}$, and let~${\vC}$ be any node 
					in the parent group of~${\vSc}$.
					It holds~${\sdist{T}{\vSc}{\vC}\leq\ell\cdot\sdist{T}{\vSc}{\vC'}}$.
			\end{lemma}
			\begin{proof}
				By triangle inequality,
				\[{\sdist{T}{\vSc}{\vC}\le\sdist{T}{\vSc}{\vC'}+\sdist{T}{\vC'}{\vC}}\formulaPunctuationSpace.\]
				If~${\vC'=\vC}$, we are done. Otherwise, there is an edge~${(\vC',\vSc')}$ in~${T}$. 
				By Lemma~\ref{lem:routing-in-group}, \[{\sdist{T}{\vC'}{\vC}\leq(\ell-1)\cdot\sdist{T}{\vC'}{\vSc'}}\] and,
				by Lemma~\ref{lem:group_edge_monotonicity},~${\sdist{T}{\vC'}{\vSc'}\le\sdist{T}{\vSc}{\vC'}}$. 
			\end{proof}
			
			\begin{lemma}\label{lem:dec-distance-modified-tree}
				Let~${(\vSc,\vSc')}$ and~${(\vC,\vC')}$ be any two edges in~${T}$.
				If~${\vC}$ lies above~${\vSc}$ in~${T'}$, 
				 then~${\sdist{T}{\vSc}{\vSc'}\ge\sdist{T}{\vC}{\vC'}}$.
			\end{lemma}
			\begin{proof}
				If~${\vSc}$ and~${\vC}$ do not belong to the same group in~${T}$, then the claim follows by Lemma~\ref{lem:group_edge_monotonicity}.
				Otherwise, our greedy choice in the construction of the group implies the claim.
			\end{proof}
			
			\paragraph{Rounding the facility openings.}
			In the previous step of our algorithm, we have computed, for~${T}$ as well as every other star tree of our star forest~${H}$, a
			new in-tree with modified groups.
			Now, to decide the eventual openings of facilities, we use the dependent rounding procedure 
			described in Section~\ref{app:dep-round} 
			that we apply on all facilities of the star forest together. We will refer to the randomized rounding algorithm by \emph{rounding procedure}.
			Later, we show how to derandomize it.

			The rounding procedure starts with the opening vector~${\ys}$. In each step, it computes a new opening vector that is used as the input for the next step. 
				In the first iterative phase (Type II iteration), the procedure considers, step by step, pairs of still fractional facilities.  
				In such a pair of facilities, the procedure pumps one of the openings up and the other one down randomly choosing the one to increase. As a result of this step, at least one of the two facilities  
				becomes either closed or open. 
				The Type II iteration phase ends when at most one fractional facility is left in~${H}$. Based on its current opening value, we randomly decide whether to close or to open it (Type I iteration).

				In the first phase, the procedure preserves the sum of facility openings (hence, their volume).
				Therefore, we will open either~${\floor{ \svol{H}}}$ or~${\ceil{ \svol{H}}}$ facilities at the end.  
			 Moreover, the probability of eventually opening the facility~${i}$ equals the opening value that it had at the beginning of the random procedure, that is, it equals~${\ysComp_i}$. 
			
			On top of these standard properties, we will also exploit that we may guarantee to almost preserve the volume in a number of chosen subsets of facilities,
			provided that the subsets form a laminar family (see Lemma~\ref{lem:dependent-rounding-sum-preservation}). 
			Here, rather than explicitly defining the family of subsets, we will directly say in which order the pairs of
			fractional facilities should be chosen.

			The rounding will proceed first within the groups until at most one fractional facility is left in each of the groups.
			Within each group, the rounding procedure will always select the top-most pair of currently fractional facilities.  
			Please note that modifying the shape of the tree inside the groups into chain graphs we have made the choice of the top-most pair unambiguous.
			When there is at most one fractional facility left in each group, the rounding may be continued in an arbitrary order. 
			At the end, there is at most one fractional facility left in the whole star forest. We open it with probability equal to its current opening.

						In more detail, consider any group~${G_\vSc}$. In each step, we select among all the facilities of~${G_\vSc}$ the two top-most fractional ones. 
						Here, we say a facility lies above another one, if the node it belongs to lies above the node of the other facility in the chain graph of~${G_\vSc}$. Then we round this pair as described in the Type~II iteration (see Section~\ref{app:dep-round}). As a result, one of both facilities gets either open or closed and their total volume remains the same. 
						Eventually, we are left with at most one fractional facility in~${G_\vSc}$.

			During the next discussions, the (unchanged) opening vector on which the rounding procedure started is still denoted by~${\ys}$. Also the volume of a star remains defined by~${\ys}$. However, when referring to the opening value of a facility without specifying the opening vector, we will refer to its opening in the opening vector returned by the rounding procedure.

			The following lemma shows that, independently of the outcome of the rounding procedure, the distribution of the volume has a  nice property: For any node, the total volume of the stars above the node is almost entirely preserved. Later, this will ensure that enough facilities are open above every node without open facilities to serve its demand.  

			\begin{lemma} \label{lem:sum-preservation}
			    Consider any group and let~${l}$ be the number of it nodes.
				For each~${m}$ with~${1\le m\le l}$, let~${\F_m}$ denote the set of all facilities of the first~${m}$ nodes from the top, and let~${v_m}$ denote the volume of~${\F_m}$ in~${\ys}$.
				
				The following is true: For each~${m}$ with~${1\le m\le l}$, the rounding procedure opens at least~${\floor{v_m}}$ facilities in~${\F_m}$.
			\end{lemma}
			\begin{proof}
			 
			  By our greedy choice of always selecting the top-most pair of fractional facilities, we exhaustingly applied Type II iteration on each set~${\F_m\in\F_1, \dots,~\F_l}$ in this order until at most one fractional facility remained in~${\F_m}$. 
			  Note that~${\F_1 \subset \dots \subset \F_l}$ forms a laminar sequence. 
			  Thus, the claim follows directly by the sum preservation shown in Lemma~\ref{lem:dependent-rounding-sum-preservation} and the fact that once opened facilities will not be closed at any step later by the rounding procedure.
			\end{proof}

			\paragraph{Routing and analysis.}
			
			Once the facilities are opened, a minimum-cost assignment of clients to facilities can be found, for example, by a minimum-cost flow computation in the original graph. Nevertheless, for the purpose of the analysis, 
			we will consider a suboptimal assignment where the demand is routed along the edges of the in-tree~${T'}$.
			
								  \begin{figure}[htb]
								  	\centering
								  	\includegraphics{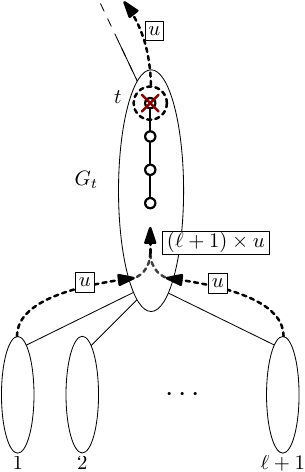}
								  	\caption{A parent group~${G_t}$ with~${\ell+1}$ children groups. 
								  		The groups are depicted as ovals with edges connecting them.
								  		Each child group sends at most~${u}$ units of demand to its parent group (dashed arrows). Thus,~${G_t}$ receives at most~${(\ell+1)u}$ units of demand from its children. 
								  		Inside~${G_t}$, the root~${t}$ is a small star whose \supporting facility has been closed.
								  		The demand of~${t}$, which is at most~${u}$, gets routed to the parent of~${G_t}$ (assuming that it exists).
								  		Any other demand of~${G_t}$ and all demand coming from its children is served within~${G_t}$.}
								  	\label{fig:demand_assignment}
								  \end{figure}
			We will make sure that the demand of any node~${\vSc\in\Cs{T}}$ will be satisfied not farther than at the root of the group of the parent of~${\vSc}$, which is not too far by Lemmas~\ref{lem:group_edge_monotonicity} and~\ref{lem:routing-in-group}. 
			We will show that after scaling up the capacity of each facility by a factor of~${2 + 3/(\ell-1)}$, each non-root group will send up at most~${u}$ units of demand to its parent group; see \fig~\ref{fig:demand_assignment}.
			Then we will argue that the excess capacity of at least~${(\ell + 1)u}$ in a group is sufficient to serve the demand sent up from its child groups.
			The tricky part is to control the demand transportation within groups.
			Note that once we let a unit of demand travel along an edge of~${T'}$, then, by paying only~${\ell}$ times more, we may let it travel further as long as it stays within the same group. 
			Therefore, it is essential to make sure that inside each group sufficient capacity is provided by the open facilities above a node to collect its demand.

									\begin{lemma}\label{lem:group-assignment}
										For any group~${G_\vSc}$,
										we can assign 
										its total demand 
										such that each open facility receives at most~${2u}$ units of demand, 
										the demand of each big star gets completely assigned to its own open facilities 
										and the following holds for every small star~${S_{\vSc'}}$:
										
										\begin{listRoman}
											\item If~${S_{\vSc'}}$ has an open facility, its demand gets completely assigned to its open facility.	
											\item Otherwise, if~${\vSc'}$ is the root of~${T'}$, its demand gets completely assigned to open facilities that belong to its son. (Note that the son exists.)			 	
											\item Otherwise, if~${\vSc'}$ is the root of~${G_{\vSc}}$, its demand remains completely unassigned.
											\item Otherwise, 
											the demand of~${S_{\vSc'}}$ 
											gets completely assigned to open facilities that belong to nodes lying above~${\vSc'}$ in~${G_{\vSc}}$.
										\end{listRoman}
										
									\end{lemma}
					\begin{proof}
						In the following, we call a small star~\emph{closed} if it contains no open facility in the opening vector returned by the rounding procedure.
									
						We set the capacity of each open facility to~${2u}$, hence, each open facility can serve a demand up to~${2u}$.
						Consider any big star. By definition, it has some volume~${v}$ larger than~${1}$ and an almost integral solution. 
						Thus, it contains~${\floor{v}}$ open facilities in~${\ys}$ and at least the same number in the output of the rounding procedure. By Constraint~\eqref{con:star_demand}, its demand is at most~${vu}$. 
						Using the fact~${2u\cdot\floor{v}\geq uv}$, we assign the demand of every big star to its own open facilities such that each facility receives at most~${2u}$ units of demand. 
						Next, consider any small star. It has some volume~${v}$ and a demand of size at most~${uv \le u}$. Hence, if a small star is not closed, then we assign its demand to its own open facility. 
						
						We now show how we assign the demand of closed small stars.
						Let~${\vSc_1,\dots,\vSc_l}$ be the nodes as they are ordered in the
						group~${G_\vSc}$ from top to bottom.
						First assume that at least one of the two statements is true: (a)~${S_{\vSc_1}}$ is not a small star, or (b)~${\vSc_1}$ is not the root of~${T'}$.
						We prove the claim of the lemma by induction for every subgroup~${\vSc_1,\dots,\vSc_m}$ where~${m=1,\dots,l}$.
						
						Consider first the base case~${m=1}$. 
						If~${S_{\vSc_1}}$ is a closed small star, then, by assumption,~${S_{\vSc_1}}$ is not the root of~${T'}$. 
						We don't route its demand and the claim holds. Otherwise the star serves its own demand as shown above and the claim also holds. 
				
						Next, consider the inductive step with~${m\geq 2}$. 
						By induction hypothesis, we compute an assignment~${\sigma}$ that satisfies the claim for the subgroup~${\vSc_1,\dots,\vSc_{m-1}}$.  We will extend this assignment to~${S_{\vSc_m}}$.
						Let~${v}$ be the total volume of the stars~${S_{\vSc_1},\dots,S_{\vSc_m}}$.
						By Lemma~\ref{lem:sum-preservation}, 
						the number of open facilities in the stars~${S_{\vSc_1},\dots,S_{\vSc_m}}$ 
						is at least~${\floor{v}}$, and, hence, their total capacity is at least~${2u\cdot\floor{v}}$.  
						By Constraint~\eqref{con:star_demand}, the total demand located at these stars is
						at most~${u v}$. Since~${m\geq 2}$ and since each star has volume at least~${1 - 1/\ell \ge 1/2}$ (Definition~\ref{def:star}), we have~${v\geq 1}$. Consequently,~${2u\cdot\floor{v}\geq uv}$.  
						Let~${v'}$ be the volume in the stars~${S_{\vSc_1},\dots,S_{\vSc_{m-1}}}$.  
						Since the assignment~${\sigma}$ distributes at most~${uv'}$ units of demand, 
						the leftover capacity in the stars~${S_{\vSc_1},\dots,S_{\vSc_m}}$ is at least~${2u\cdot\floor{v}-uv'}$. This amount is larger
						than the demand of~${S_{\vSc_m}}$, which is at most~${u(v-v')=uv-uv'}$.  
						Therefore, if~${S_{\vSc_m}}$ is a
						small star whose fractional facility has been closed, we assign the demand of~${S_{\vSc_m}}$ in an arbitrary
						manner to the leftover capacity.  
						In all other cases, the star serves its demand itself as argued above.
						Thus, we get a new assignment for the subgroup~${\vSc_1,\dots,\vSc_{m}}$ satisfying the claim.
					
						Now, consider the case that~${S_{\vSc_1}}$ is a small star and~${\vSc_1}$ is the root of~${T'}$. By Property~\ref{prop:in-degree-bound}, the root of~${T'}$ has a son and thus~${l\ge 2}$.
						Again, we prove the claim of the lemma by induction for every subgroup~${\vSc_1,\dots,\vSc_m}$ where~${m=2,\dots,l}$.
						
						For the base case~${m=2}$, the claim immediately holds if neither~${S_{\vSc_1}}$ nor~${S_{\vSc_2}}$ is a closed small star as then the stars serve their demands by themselves. Otherwise, at least one of the stars is a closed small star. 
						Let~${v}$ be the total volume of~${S_{\vSc_1}}$ and~${S_{\vSc_2}}$ in~${\ys}$. 
						Recall, by Lemma~\ref{lem:sum-preservation}, that there are at least~${\floor{ v }}$ open facilities in the two stars. 
						Thus, the open facilities in the stars provide a capacity of at least~${2u \cdot \floor{ v }}$.
						As every star has volume at least~${1-1/\ell\ge 1/2}$, we have~${\floor{v}\ge1}$.	
						This fact implies two observations.					
						First, at least one facility is open. Hence, exactly one of the two stars is a closed small star. 
						Secondly, the capacity provided by the open facilities 
						is enough to serve their total demand, which is at most~${vu}$.			
					    Thus, we obtain a feasible assignment by assigning the demand of the closed small star in an arbitrary manner to the leftover capacity of the other star.
					
						For the inductive step with~${m\ge 3}$, we apply the same arguments as in the inductive step described above. 
						Hence, we obtain an assignment for the whole group~${G_\vSc}$ that satisfies the claim.

					\end{proof}

			The lemma above shows that, within each group, we can satisfy the demand of all its nodes, 
			except perhaps the root node of the group.
			If the demand of the root of a group~${G_\vSc}$ is not satisfied, then, by Lemma~\ref{lem:group-assignment}, the root~${\vSc}$ is associated with a small star and~${\vSc}$ is not the root of~${T'}$. Consequently, the unsatisfied demand is at most~${u}$ and there exists a parent group for~${G_\vSc}$. 
			To satisfy the demand of~${\vSc}$, we will forward it to the parent group.
			Thus, 
			we have to show that the parent group~${G_{\vC}}$ has enough capacity left to serve the demand sent from~${G_\vSc}$ and from every other of its children. 
			As~${G_{\vC}}$ is a non-leaf group, it contains~${\ell}$ stars. Since each star has volume at least~${1-1/\ell}$ (see Definition~\ref{def:star}), the volume~${v}$ of~${G_{\vC}}$ is at least~${\ell-1}$. By Lemma~\ref{lem:sum-preservation}, there are~${\floor{v}}$ open facilities in~${G_{\vC}}$. After scaling the capacities with~${(2+3/(\ell-1))}$ and using~${\floor{v}\ge \ell-1}$, we can lower bound the total capacity in~${G_{\vC}}$ by 
			\begin{alignat*}{2}
			&&~&(2+3/(\ell-1))\floor{v}u
			\\&=&&\floor{v}u + (1+3/(\ell-1))\floor{v}u
			\\&\ge&&(v-1)u + (1+3/(\ell-1))(\ell-1)u
			\\&\ge&& vu + (\ell + 1)u\formulaPunctuationSpace.
			\end{alignat*}
			From this capacity, at most~${vu}$ is used for the demand from~${G_\vSc}$. Consequently, at least~${(\ell+1)u}$ capacity remains to be potentially used by demand forwarded from the child groups. Since there are at most~${\ell +1}$ child groups and each of them forwards at most~${u}$ units of demand, the remaining capacity is sufficient. Thus, we assign the demand coming from the children arbitrarily on the facilities in~${G_{\vC}}$ that still have some capacity left; see \fig~\ref{fig:demand_assignment}.

			We summarize the properties of our assignment as follows.
			\begin{lemma}\label{lem:complete-assignment}
				We can assign
				the total demand in~${T'}$ 
				such that each open facility receives at most~${2+3/(\ell-1)u}$ units of demand, 
				the demand of each big star gets completely assigned to its own open facilities 
				and the following holds for every small star~${S_\vSc}$:
				
				\begin{listRoman}
					\item If~${S_\vSc}$ has an open facility, its demand gets completely assigned to its open facility.	
					\item Otherwise, if~${\vSc}$ is the root of~${T'}$, its demand gets completely assigned to open facilities that belong to its son.
					\item Otherwise, if~${\vSc}$ is the root of a group, its demand gets completely assigned to open facilities that belong to  nodes lying in the parent group of~${\vSc}$. 
					\item Otherwise, the demand of~${S_\vSc}$ gets completely assigned to open facilities that belong to nodes lying above~${\vSc}$ in the same group as~${\vSc}$.
				\end{listRoman}
				
			\end{lemma}
		
		\paragraph{Expected Cost.}	
		Now, we examine the expected cost of the assignment above, that is, the cost of placing facilities in~${T'}$ and routing demand within~${T'}$. 
			
			Consider any outcome of the rounding procedure together with our assignment of Lemma~\ref{lem:complete-assignment}.
			We will split the cost over all facilities such that each facility is charged with some cost portion: 			
			Every facility pays its own opening cost
			as well as the connection cost of the demand that it receives directly from its star center.
			If the single facility~${i}$ of a small star is closed, it pays the connection cost of distributing the demand of its star~${S_\vSc}$ on other star centers. 
			If a receiving star center~${\vC}$ belongs to a small star~${S_{\vC}}$, then~${i}$ even pays for forwarding the demand to the single open facility in~${S_{\vC}}$.
			If~${S_{\vC}}$ is big, then each facility~${i'}$ receiving a demand portion from~${\vSc}$ pays for the last stretch of moving the demand from~${\vC}$ to~${i'}$.
			Note that by this charging schema, we have completely split the cost over all facilities.

			In what follows, we upper bound the cost contribution of each facility~${i}$ depending on whether~${i}$ is open or closed in the outcome of the rounding procedure.
			For the case that~${i}$ is open, we use~${\copen{i}}$ to denote its upper bound, otherwise, we use~${\cclosed{i}}$.			
			Thus, for any opening vector~${\rys}$ returned by the rounding procedure, the total cost of the solution to~${T}$ 
			will be bounded from above by \emph{the gross cost of~${T}$} that we define as follows. 	
			\begin{definition} 
				The gross cost of~${T}$ is~${\sum_{i\in\Fs{T}} \rysComp_i \copen{i} + (1-\rysComp_i)\cclosed{i}}$.
			\end{definition}
			
			Consider any facility~${i}$ and assume that it belongs to a big star~${S_\vSc}$.
			For the case that~${i}$ is closed,~${i}$ is not charged and we set~${\cclosed{i}=0}$.
			For the case that~${i}$ is open, it has to pay for its opening cost~${f_i}$ as well as for moving from its star center all the demand that it receives.
			By Lemma~\ref{lem:complete-assignment},~${i}$ receives at most~${(2+3/(\ell-1))u}$ units of demand.  
			Thus, we set~${\copen{i} = f_i + (2+3/(\ell-1))u \sdist{T}{i}{\vSc}}$.
			
			Now, assume that~${i}$ belongs to a small star~${S_\vSc}$.
			If~${i}$ is not the single \supporting facility~${\iraised}$ of~${S_\vSc}$ in~${\ys}$, it will remain closed and therefore not charged. We set~${\copen{i}=\cclosed{i}=0}$.
			Otherwise, consider~${i=\iraised}$.
			For the case that~${\iraised}$ is open, it has to pay for its opening cost as well as for moving the demand that it receives directly from its star center.
			Since a small star has at most one open facility,~${\iraised}$ has to pay for the full demand~${w_\vSc}$ of~${\vSc}$. Recall that it is not charged for receiving demand originating from other stars. Thus, we set~${\copen{\iraised} = f_\iraised + w_\vSc \sdist{T}{\iraised}{\vSc}}$. 
			
			For the case that~${\iraised}$ is closed, it has to pay for distributing the demand~${w_\vSc}$ on other stars. We distinguish two cases.
		
			First, assume that~${\vSc}$ is the root of~${T}$ and let~${\vSc'}$ be its single son in~${T}$ (see Property~\ref{prop:in-degree-bound}), which is the same in~${T'}$. 
			By Lemma~\ref{lem:complete-assignment}, all the demand of~${S_\vSc}$ gets assigned to facilities in~${S_{\vSc'}}$. 
			Thus, we move the demand to~${\vSc'}$ over the distance~${\sdist{T}{\vSc}{\vSc'}}$.
			If~${S_{\vSc'}}$ is a small star, we further route the demand to the single open facility in~${S_{\vSc'}}$.  
			By Property~\ref{prop:small-dist}, we traversed in total a distance at most~${(1+1/2)\sdist{T}{\vSc}{\vSc'}}$.
			Generously, we set~${\cclosed{\iraised} = w_{\vSc}(\ell+1/2)\sdist{T}{\vSc}{\vSc'}}$.
			
			Next, assume that~${\vSc}$ is not the root of~${T}$.
			Let~${\vSc'}$ be the father of~${\vSc}$ in~${T}$. 
			Fix a portion~${\hat{w_\vSc}}$ of the demand that is routed to some star~${S_{\vC}}$.
			By Lemma~\ref{lem:complete-assignment},~${\vC}$ has to lie above~${\vSc}$ in~${T'}$. If~${\vSc}$ is the root of its group, then~${\vC}$ belongs to the parent group of~${G_\vSc}$. 
			Otherwise, if~${\vSc}$ is not the root of its group,  
			then~${\vC}$ belongs to the same group as~${\vSc}$. 
			By Lemmas~\ref{lem:routing-in-group} and~\ref{lem:routing-between-group}, we know that the distance between~${\vSc}$ and~${\vC}$ is
			at most~${\ell\sdist{T}{\vSc}{\vSc'}}$.  
			If~${S_{\vC}}$ is a big star, we just route~${\hat{w_\vSc}}$ to~${\vC}$.  
			If~${S_{\vC}}$ is a small star, we further route~${\hat{w_\vSc}}$ to its single open facility~${\iraised'}$. 
			
			Assume that~${\vC}$ is not the root of~${T}$. Let~${\vC'}$ be the father of~${\vC}$ in~${T}$. 
			Since~${\vC}$ lies above~${\vSc}$ in~${T'}$, we have~${\sdist{T}{\vC}{\vC'}\le\sdist{T}{\vSc}{\vSc'}}$ by Lemma~\ref{lem:dec-distance-modified-tree}.
			Thus, by Property~\ref{prop:small-dist}, the distance from~${\vC}$ to~${\iraised'}$ is at most \[{\sdist{T}{\vC}{\vC'}/2\le \sdist{T}{\vSc}{\vSc'}/2}\formulaPunctuationSpace.\] 
	
			Now, if~${\vC}$ is the root of~${T}$, then let~${\vC'}$ be the single son of~${\vC}$ in~${T}$ and~${T'}$. If~${\vC'}$ lies above~${\vSc}$ in~${T'}$, then Lemma~\ref{lem:dec-distance-modified-tree} implies~${\sdist{T}{\vC}{\vC'}\le\sdist{T}{\vSc}{\vSc'}}$. Otherwise,~${\vC'=\vSc}$, and thus~${(\vC',\vC)=(\vSc,\vSc')}$.   
			Again, by Property~\ref{prop:small-dist}, we have \[{\sdist{T}{\iraised}{\vC}\le \sdist{T}{\vC}{\vC'}/2\le \sdist{T}{\vSc}{\vSc'}/2}\formulaPunctuationSpace.\] 
			
			Hence, the total distance for~${\hat{w_\vSc}}$ is at most~${(\ell+1/2)\sdist{T}{\vSc}{\vSc'}}$.
			Since this bound holds for any portion of the demand~${w_\vSc}$, 
			we set~${\cclosed{\iraised} = w_{\vSc}(\ell+1/2)\sdist{T}{\vSc}{\vSc'}}$.

		We summarize our bounds:
		\begin{definition}
			Let~${\Csbig{T}}$ denote the set of all star centers of big stars in~${\Cs{T}}$, 
		    and let~${\Cssmall{T}}$ denote the set of all star centers of small stars in~${\Cs{T}}$.
		\begin{itemize}
			\item For every~${\vSc\in\Csbig{T}}$ and~${i\in\F_\vSc}$, we set 
			\[\copen{i}=f_i + (2+3/(\ell-1))u \sdist{T}{i}{\vSc}
			\textrm{\enspace\enspace and \enspace\enspace}
			\cclosed{i}=0\formulaPunctuationSpace.\]
			\item For every~${\vSc\in\Cssmall{T}}$ and the only \supporting facility~${\iraised}$ of~${S_\vSc}$ in~${\ys}$, we set 
			\[\copen{\iraised} = f_\iraised + w_\vSc \sdist{T}{\iraised}{\vSc} 	\textrm{\enspace\enspace and \enspace\enspace}
			\cclosed{\iraised}=w_{\vSc}(\ell+1/2)\sdist{T}{\vSc}{\vSc'} \formulaPunctuationSpace,\]
			where~${\vSc'}$ is the son of~${\vSc}$ if~${\vSc}$ is the root, 
			and the father of~${\vSc}$ in~${T'}$ otherwise.
			\item For every other facility~${i\in\Fs{T}}$, we set
			\[\copen{i}=0
			\textrm{\enspace\enspace and \enspace\enspace}
			\cclosed{i}=0\formulaPunctuationSpace.\]
		\end{itemize}
		\end{definition}

			Now we are ready to bound the expected gross cost. 
						\begin{lemma}\label{lem:expected-cost}
							The expected gross cost of our solution to~${T}$ is at most~${(4\ell+3)b(T)}$.
						\end{lemma}
						\begin{proof} 
							
							For every~${\vSc\in\Cssmall{T}}$, we define~${\iraised(\vSc)}$ as the single \supporting facility of~${S_\vSc}$ in~${\ys}$.
							Recall that every facility~${i\in\Fs{T}}$ is opened with probability equal to~${\ysComp_i}$ by the rounding procedure.
							Using our definitions, the expected gross cost is
							\begin{alignat*}{2}
							&&~&\sum_{i\in\Fs{T}}  \ysComp_i \copen{i} + (1-\ysComp_i)\cclosed{i}
							\\&=&&\sum_{\vSc\in\Csbig{T}}\sum_{i\in\F_\vSc}  \ysComp_i \copen{i} 
							+  \sum_{\vSc\in\Cssmall{T}} \ysComp_{\iraised(\vSc)} \copen{{\iraised(\vSc)}} + (1-\ysComp_{\iraised(\vSc)})\cclosed{{\iraised(\vSc)}}\formulaPunctuationSpace.
							\end{alignat*}
						
						    Next, we separately bound the cost contribution of facilities belonging to big and small stars.
						    Thereby, we will use~${\dist{i}{\vSc}=\sdist{T}{i}{\vSc}}$ for every~${\vSc\in\Cs{T}}$ and~${i\in\F_\vSc}$, which holds by the definition of~${\ds{T}}$.
							
							For each~${\vSc\in\Csbig{T}}$, we have
							\begin{alignat*}{2}
							&&~&\sum_{i\in\F_\vSc}  \ysComp_i \copen{i}  
							\\&=&& \sum_{i\in\F_\vSc} \ysComp_i f_i + (2+3/(\ell-1))u\sum_{i\in\F_\vSc}\ysComp_i\sdist{T}{i}{\vSc} 
							\\&\le&&(2+3/(\ell-1)) \bsmall_\vSc
							\\&\le&&5 \bsmall_\vSc \formulaPunctuationSpace,
							\end{alignat*}
							where the first inequality follows from Constraint~\eqref{con:star_strict_budget} for solutions to star instances, and the second one follows from~${\ell\ge 2}$. 
							
							For each~${\vSc\in\Cssmall{T}}$ with the single \supporting facility~${\iraised}$ in~${\ys}$, we have
							\begin{alignat*}{2}
							&&~&\ysComp_{\iraised} \copen{{\iraised}}
							\\
							&=&& \ysComp_{\iraised} f_{\iraised} + \ysComp_{\iraised} w_\vSc \sdist{T}{{\iraised}}{\vSc}
							\\
							&\le&& \ysComp_{\iraised} f_{\iraised} + w_\vSc \sdist{T}{{\iraised}}{\vSc}
							\\
							&\le&&\bbig_\vSc \formulaPunctuationSpace,
							\end{alignat*}
							where the last inequality follows from Constraint~\eqref{con:star_budget}. 
							If~${\vSc}$ is the root of~${T'}$, then let~${\vSc'}$ be the son of~${\vSc}$, otherwise, let~${\vSc'}$ be the father of~${\vSc}$.
							We have
							\begin{alignat*}{2}
							&&~&(1-\ysComp_{\iraised})\cclosed{{\iraised}}
							\\
							&=&& (1-\ysComp_{\iraised})  w_{\vSc}(\ell+1/2)\sdist{T}{\vSc}{\vSc'} 
							\\
							&\le&& 8(\ell+\frac{1}{2})(\budDirectConnection{\vSc}+\ell\budRelativeConnection{\vSc}) 
							\\
							&\le&& (4\ell+2)(2\budDirectConnection{\vSc}+2\ell\budRelativeConnection{\vSc}) 
							\\
							&\le&& (4\ell+2)\bbig_\vSc \formulaPunctuationSpace,
							\end{alignat*}
							where the first inequality follows from Property~\ref{prop:rerouting}. 
							Thus,
							\[{\ysComp_{\iraised} \copen{{\iraised}} + (1-\ysComp_{\iraised})\cclosed{{\iraised}} \le (4\ell+3)\bbig_\vSc}\formulaPunctuationSpace.\]
							
							Given~${(4\ell+3)\bbig_\vSc \ge 5\bsmall_\vSc}$ for each~${\vSc \in \Cs{T}}$,
							we infer the following upper bound on the expected cost:
							\begin{alignat*}{2}
							&&~&\sum_{i\in\Fs{T}}  \ysComp_i \copen{i} + (1-\ysComp_i)\cclosed{i}
							\\&\le&&\sum_{i\in\Fs{T}}  (4\ell+3)\bbig_\vSc
							\\&=&&  (4\ell+3)b(T)\formulaPunctuationSpace.\qedHereInAlign
							\end{alignat*}
						\end{proof}
					
			The result on the trees directly extends to the forest~${H}$, as no demand is routed between any two star trees. 
			Here, we define the \emph{gross cost of~${H}$} as the sum of the gross costs of its star trees.
			\begin{corollary}\label{cor:expected-forest-cost}
				The expected gross cost 
				of our solution to the star forest~${H}$ is at most~${(4\ell+3)b(H)}$. 
			\end{corollary}
			Note that the same upper bound also applies on the expected cost of our solution since the gross cost is always an upper bound on the actual cost.
					
			\paragraph{Derandomization.}	
			Let~${\F_H}$ denote the set of all facilities in~${H}$.
			Consider any step of the rounding procedure and let~${\rys}$ be the current opening vector at the beginning of the step.
			From this step on, the rounding procedure guarantees that each facility~${i\in\F_H}$ will be opened with probability~${\rysComp_i}$.
			Hence, from this step on,  
			the expected value of the gross cost is exactly~${\sum_{i\in\F_H}  \rysComp_i \copen{i} + (1-\rysComp_i)\cclosed{i}}$.
			This value is an upper bound for the expected gross cost of at least one of the choices in the current step.
			This leads us directly to the following derandomization algorithm: At each step of the rounding procedure, compute the expected value of the gross cost for each of the two choices of rounding.  
			Greedily make the choice with the smaller expected value, or make an arbitrary choice if both values are equal.
			At the end, we deterministically obtain a solution to~${H}$ whose gross cost is at most 
			the expected gross cost of our rounding procedure. 
			
			Using Corollary~\ref{cor:expected-forest-cost}, we summarize our discussion in the following theorem.
		\begin{theorem}
		\label{thm:algorithm-star-forest}
		For any value at least~${2}$ for the parameter~${\ell}$, there is an efficient algorithm that computes for a given star
		forest~${H}$ a solution with capacity violation at most~${2+3/(1-\ell)}$ and cost at most~${(4\ell + 3)\cdot b(H)}$.
	\end{theorem}
						Combining this result with Theorem~\ref{thm:central_theorem}, we obtain Theorem \ref{thm:main_2_eps}.

				\section{Algorithm for Non-uniform Hard-Capacitated~\PROBkmed}
				\label{sec_ckm}
				In this section, we describe a bi-factor approximation algorithm for
				the~\PROBkmed problem with non-uniform hard capacities that will prove Theorem~\ref{thm:main_3_eps}.  
				Note that 
				this problem is 
				equivalent to~\PROBkfl with uniform opening costs. 
				Moreover, its standard LP relaxation, denoted by~\LPkMed{}, is a special case of~\LPfacLoc{} where all opening costs~${(f_i)_{i\in\F}}$ are set to~${0}$.

				During our algorithm,
				we will obtain, step by step, a series of solutions where the initially fractional openings are more and more restricted until we finally arrive at an integral solution to~\LPkMed{} with bounded capacity violation.
				We will consider the following two types of solutions.
				\begin{definition}
					A solution~${(\widetilde{\x},\widetilde{\y})}$ to~\LPkMed{} is called a \emph{$[\puthalf,1]$-solution} if
					we have~${\widetilde{\yComp}_i \in \{0\} \cup [\puthalf,1]}$ for every~${i \in \F}$. 			
					Similarly, a solution~${(\widetilde{\x},\widetilde{\y})}$ to~\LPkMed{} is called a \emph{$\{\puthalf,1\}$-solution} if we have~${\widetilde{\yComp}_i \in \{0, \puthalf,1\}}$ for every~${i \in \F}$.
				\end{definition}
				
				Recall that a solution with capacity violations has to satisfy all constraints of~\LPkMed{} with the exception of Constraint~\eqref{lp:capacity_ge_total_demand}. In this section, we will consider an even weaker type of solutions.
				\begin{definition}
					A solution~${(\widetilde{\x},\widetilde{\y})}$ 
					to the weaker version of~\LPkMed{}
					where we drop Constraints~\eqref{lp:opening_ge_demand} and~\eqref{lp:capacity_ge_total_demand}  
					is called a \emph{weak solution} to~\LPkMed{}.
				\end{definition}

				Let~${(\xo,\yo)}$ be the optimum solution to~\LPkMed{} that we fixed in Section~\ref{sec:bundles}. 
				Similarly to the algorithm of the previous section, we partition all facilities into \starTails{} and, for each corresponding star instance, we compute a strict solution of volume bounded by a function of~${\ell}$. In doing so, we set the parameter~${\ell=2}$. 
				Note that we do not consider relaxed solutions.  
				We then proceed as follows.

				In Section~\ref{[1-1/l,1]-solution},  
				we modify the solution to each star instance
				by 
				moving openings between facilities 
				such that all \supporting facilities are open or all but one of the facilities are closed.
				The union of these solutions will  
				help us to obtain a~${[\puthalf, 1]}$-solution~${(\x', \y')}$ where fractional facilities have capacity violation at most~${1+\varepsilon}$ and open facilities have capacity violation at most~${2+\varepsilon}$, for any sufficiently small constant~${\varepsilon}$.
				
				Then, in Section~\ref{bank_step}, we construct a weak~${\{\puthalf, 1\}}$-solution~${(\hat{\x}, \hat{\y})}$: 
				By some greedy rule, we either round each opening in~${\y'}$ down to~${\puthalf}$ or up to~${1}$. 
				Thereby, we might violate Constraint~\eqref{lp:opening_ge_demand} and therefore obtain only a weak solution. 
				Also the capacity violation might increase slightly, but not more than up to~${2 + 2\varepsilon}$.
				The connection cost of~${(\hat{\x}, \hat{\y})}$ remains the same as in~${(\x', \y')}$.
				
				Eventually, in Section~\ref{final_rounding}, we round~${(\hat{\x}, \hat{\y})}$ into an integral solution. We do this by
				building so called \emph{facility trees} and cutting them to smaller instances which are easier to round. 
				By this procedure, we obtain an integral solution~${(\bar{\x}, \bar{\y})}$ to~\LPkMed{} with capacity violation~${3 + 3\varepsilon}$.

				\subsection{Obtaining a~${[\puthalf,1]}$-Solution with Capacity Violations.}
				\label{[1-1/l,1]-solution}
				
				In this section, we describe how to obtain a 
				solution~${(\x',\y')}$ with capacity violations such that the opening of every \supporting facility is in~${[\puthalf,1]}$. 
				Let~${\varepsilon}$ be an arbitrary constant with~${0<\varepsilon\le 1}$.  
				
				In the following, we consider a star instance~${S_\vSc}$. 
				By Lemmas~\ref{lem:compute_strict_solution}~and~\ref{lem:at_most_two_fractional}, we compute 
				a strict solution~${\ys}$ to~${S_\vSc}$ with at most two fractional facilities and the bound~${\vol{\ys} \le \fvol{\yo}{\F_\vSc}}$.
				Since~${\ys}$ satisfies~Constraint~\eqref{con:star_demand}, we can completely distribute the demand~${w_\vSc}$ on the facilities in~${\F_\vSc}$ such that each facility~${i}$ serves a demand~${\demand_i}$ of size at most~${\ysComp_iu_i}$. 
				We fix such a distribution of~${w_\vSc}$ and define~${\demand_i}$ as the demand that the facility~${i}$ has to serve; thus~${\sum_{i\in \F_\vSc}\demand_i = w_\vSc}$.
				Note that there is no capacity violation in our distribution and that the cost of moving~${\demand_i}$ to the facility~${i}$ is at most~${\dist{i}{\vSc} \demand_i \le \dist{i}{\vSc} \ysComp_iu_i}$. Hence, by Constraint~\eqref{con:star_strict_budget},  
				the connection cost of our distribution, that is, the cost of sending~${w_\vSc}$ from the star~center~${\vSc}$ to the facilities, is upper bounded by the strict budget~${\bsmall_\vSc}$. 
				
				Next, we compute a new opening vector~${\ys'}$ for the facilities in~${\F_\vSc}$ where either all \supporting facilities are open, or there is one fractional facility and all other facilities are closed.  
				In parallel, we assign each facility a demand~${\demand'_i}$ such that the new distribution of~${w_\vSc}$ has cost in~${\bigOh(\bsmall_\vSc)}$ and capacity violation in~${\bigOh(1+\varepsilon)}$.
				Depending on the size of~${\vol{\ys}}$, we compute~${\ys'}$ in one of two different ways.

				\paragraph{Small Volume.}
				First, we consider the case when~${\vol{\ys}\le 1}$, which is always true for star instances whose \starTails{} have volume at most~${1}$,  
				and might sometimes also hold star instances with \starTails{} of volume greater than~${1}$.  
				If there is only one fractional facility~${\iraised}$,
				we just set~${\ysComp'_\iraised= \min\{1, \fvol{\yo}{\F_\vSc}\}}$ 
				and~${\demand'_\iraised = \demand_\iraised}$. Then~${\ysComp'_\iraised\ge\vol{\ys}= \ysComp_\iraised}$ and we have no capacity violation.
				If there are exactly two fractional facilities~${\iraisedA}$ and~${\iraisedB}$,  
				then we 
				close one of them and
				move 
				its demand and opening to the other one. 
				In fact, we can do so without any increase of capacity violation as the next lemma shows.
				\begin{lemma}\label{no_cap_violation_in_small_bundle}
					Let~${i, i' \in \F_\vSc}$ satisfy~${\ysComp_{i} + \ysComp_{i'} \le 1}$.
					For at least one~${i''\!\in\{i,i'\}}$, there is no capacity violation if we set the opening of~${i''}$ to~${\ysComp_{i}+\ysComp_{i'}}$ and its demand to~${\demand_i + \demand_{i'}}$.
				\end{lemma}
				\begin{proof}
					Take~${i''\!\in\{i,i'\}}$ with~${u_{i''}\! = \max\{u_i, u_{i'}\}}$ and observe~${(\ysComp_{i} + \ysComp_{i'})u_{i''} \geq \demand_i + \demand_{i'}}$.
				\end{proof}
				
				Unfortunately, the connection cost might be unbounded in the lemma above. However, if we allow a slight capacity violation, we can control the cost.
				
				\begin{lemma}\label{small_bundle_cap}
					Let~${\varepsilon'}$ satisfy~${0<\varepsilon'\le 1}$ and
					let~${i, i' \in \F_\vSc}$ satisfy~${\ysComp_{i} + \ysComp_{i'} \le 1}$.
					For at least one~${i''\in\{i,i'\}}$, the following is true:
					If we set the opening of~${i''}$ to~${\ysComp_{i}+\ysComp_{i'}}$ and its demand to~${\demand_i + \demand_{i'}}$,
					then 
					its capacity violation is at most~${1 + \varepsilon'}$ and 
					we get~${\demand_i + \demand_{i'} \le (1+\varepsilon')/\varepsilon' \cdot \demand_{i''}}$.
				\end{lemma}
				\begin{proof}
					If for both choices of~${i''}$ the resulting capacity violation is at most~${1+\varepsilon'}$,
					we choose~${i''\in\{i,i'\}}$ with maximum~${\demand_{i''}}$. 
					Then~${(\demand_i + \demand_{i'})/\demand_{i''} \le 2 \le  (1+\varepsilon')/\varepsilon'}$ and the claim holds.
					
					Now, assume that for one of the choices, say~${i'}$, 
					we have capacity violation~${\gamma}$ greater than~${1 + \varepsilon'}$.
					Then, by Lemma~\ref{no_cap_violation_in_small_bundle}, the other choice for~${i''}$ 
					has no capacity violation at all. 
					Thus, we choose~${i''=i}$.
					Note that~${\gamma}$ satisfies~${\gamma u_{i'} (\ysComp_i + \ysComp_{i'})=(\demand_i + \demand_{i'})}$. Given~${u_{i'}\ysComp_{i'} \ge \demand_{i'}}$, we obtain 
					\[\frac{\ysComp_{i'}}{\ysComp_i + \ysComp_{i'}} \cdot \frac{\demand_i + \demand_{i'}}{\demand_{i'}} \ge \gamma > 1 + \varepsilon' \formulaPunctuationSpace.\]
					This leads to~${(\demand_i + \demand_{i'})/{\demand_{i'}} > 1 + \varepsilon'}$. 
					Together with~${1 - \demand_{i'}/(\demand_i + \demand_{i'}) = {\demand_i}/(\demand_i + \demand_{i'})}$, we obtain \[{\frac{\demand_i + \demand_{i'}}{\demand_i} < \frac{1+\varepsilon'}{\varepsilon'}}\] for the demand increase of~${i}$. 
					
				\end{proof}
				
				By Lemma~\ref{small_bundle_cap} and choosing~${\varepsilon'=\varepsilon}$, we select an appropriate facility~${\iraised\in\{\iraisedA,\iraisedB\}}$. We set its opening~${\ysComp'_{\iraised}=\min\{1, \fvol{\yo}{\F_\vSc}\}}$, 
				its demand~${\rest{\demand}_{\iraised}=w_\vSc}$ and close the other facility. 
				Thus,~${\ysComp'_{\iraised}\ge \vol{\ys}}$ and the capacity bound still holds.
				Given~${\demand'_{\iraised} \le(1+\varepsilon)/{\varepsilon} \cdot \demand_{\iraised}}$ and given~${\demand_{\iraised} \dist{\iraised}{\vSc}\le \ysComp_{\iraised} u_{\iraised} \dist{\iraised}{\vSc} \le \bsmall_\vSc}$ by Constraint~\eqref{con:star_strict_budget}, the connection cost is at most~${(1+\varepsilon)/{\varepsilon} \cdot \bsmall_\vSc}$.

				\paragraph{Big Volume.}Next, we consider the case when~${\vol{\ys}> 1}$, which is only true for star instances whose \starTails{} have volume greater than~${1}$.  
				If there are no fractional facilities, 
				then we just set~${\ys' = \ys}$
				and~${\demand_i' = \demand_i}$ for each~${i\in \F_\vSc}$. 
				
				Otherwise, consider two \supporting facilities~${\iraisedA}$ and~${\iraisedB}$ with smallest openings. Since we have at most two fractional facilities, the remaining \supporting ones are  open. 
				
				\begin{lemma}\label{lem:bigVolumeTwoFac}
					Let~${i, i' \in \F_\vSc}$ 
					satisfy~${\ysComp_{i} + \ysComp_{i'} \geq 1}$. 
					For at least one of the two facilities the following is true: 
					If we open the facility and set its demand to~${\demand_{i} + \demand_{i'}}$,
					then its capacity violation and demand increase by a factor at most~${2}$.  
				\end{lemma}
				\begin{proof}
					Choose a facility~${i''}$ in~${\{i,i'\}}$ with the demand~${\max\{\demand_{i}, \demand_{i'}\}}$ and observe that the capacity violation is at most~${\ysComp_{i''}/(\ysComp_i + \ysComp_{i'}) \cdot (\demand_i + \demand_{i'})/\demand_{i''} \le  (\demand_i + \demand_{i'})/\demand_{i''} \le 2}$ and the connection cost is~${(\demand_i + \demand_{i'})\dist{i''}{\vSc} \le 2 \demand_{i''}\dist{i''}{\vSc}}$.
				\end{proof}
				If~${\ysComp_{\iraisedA} + \ysComp_{\iraisedB} \geq 1}$, we choose one of them by Lemma~\ref{lem:bigVolumeTwoFac} to be open in~${\ys'}$ and close the other one. Otherwise, if~${\ysComp_{\iraisedA} + \ysComp_{\iraisedB} < 1}$, then there is an open facility~${\iraised}$ in the star instance. 
				\begin{lemma}\label{lem:bigVolumeThreeFac}
					Let~${i, i' ,i'' \in \F_\vSc}$ 
					satisfy~${\ysComp_{i}=1}$ and~${\ysComp_{i'} + \ysComp_{i''} < 1}$.
					For at least one of the three facilities  
					the following is true:
					If we open the facility and set its demand to~${\demand_{i} + \demand_{i'} + \demand_{i''}}$,
					then its capacity violation increases by a factor at most~${2+\varepsilon}$ and its demand increases by a factor of a most~${2+4/\varepsilon}$. 
				\end{lemma}
				\begin{proof}
					By Lemma~\ref{small_bundle_cap} and choosing~${\varepsilon'=\varepsilon/2}$, we select one of~${\{i',i''\}}$, say~${i'}$, 
					such that 
					its capacity violation is at most~${1+\varepsilon'}$ and its demand is at most~${(1+\varepsilon')/\varepsilon' \cdot \demand_{i'}}$ 
					
					Then we apply Lemma~\ref{lem:bigVolumeTwoFac} on~${i}$ and~${i'}$ using as the opening and demand of~${i'}$ 
					the quantities~${\ysComp_{i'}+\ysComp_{i''}}$ and~${\demand_{i'}+\demand_{i''}}$, respectively.
					In the worst case, we choose to open~${i'}$ which causes a capacity violation at most~${2 (1+\varepsilon') = 2 + \varepsilon}$ and an total increase of demand by a factor at most~${2(1+\varepsilon')/\varepsilon' = 2 + 4/\varepsilon}$.
				\end{proof}
				
				By Lemma~\ref{lem:bigVolumeThreeFac}, we choose one of the three facilities in~${\{\iraised,\iraisedA,\iraisedB\}}$ to be open in~${\ys'}$ and close the other two. 
				In both cases, we route the demand of the closed facilities to the chosen one. The resulting capacity violation is at most~${2+\varepsilon}$. Since we do not change the openings and demands of any other facilities of the star instance, and only increase the demand of one facility by a factor at most~${2 + 4/\varepsilon}$,
				the total connection cost is at most~${\sum_{i \in \F_\vSc} \dist{i}{\vSc}(2 + 4/\varepsilon) \demand_{i}}$.
				By Constraint~\eqref{con:star_strict_budget}, this cost is at most~${(2 + 4/\varepsilon) \bsmall_\vSc}$.
				Note the inequality~${\vol{\ys'}\le\vol{\ys}\le\fvol{\yo}{\F_\vSc}}$.

				The discussion of the two cases of~${\vol{\ys}}$ can be summarized as follows.
				\begin{lemma}\label{lem:restricted_solution}
					We can compute an opening vector~${\ys'}$ for~${\F_\vSc}$ with~${\vol{\ys'}\le\fvol{\yo}{\F_\vSc}}$ 
					where
					\begin{itemize}
						\item there is only one \supporting facility~${i\in\F_\vSc}$ and~${\ysComp'_i=\fvol{\yo}{\F_\vSc}}$, or
						\item all \supporting facilities are open.
					\end{itemize}
					Furthermore, we can distribute the demand~${w_\vSc}$ on the facilities \supporting in~${\ys'}$ such that 
					\begin{itemize} 
						\item each fractional facility~${i\in\F_\vSc}$ has capacity violation at most~${1+\varepsilon}$, 
						and
						\item each open facility~${i\in\F_\vSc}$ has capacity violation at most~${2+\varepsilon}$. 
					\end{itemize} 
					The connection cost of the distribution 
					is at most~${(2 + 4/\varepsilon) \bsmall_\vSc}$.
				\end{lemma}

				\begin{corollary}
					\label{cor:restr}
					For any~${\varepsilon}$ with~${0<\varepsilon\le 1}$,
					we can efficiently compute a~${[\puthalf, 1]}$-solution~${(\rest{\x}, \rest{\y})}$ with capacity violations
					such that~${\vol{\rest{\y}}\le k}$ holds,
					fractional facilities have capacity violation at most~${1+\varepsilon}$, 
					open facilities have capacity violation at most~${2+\varepsilon}$, and
					the total connection cost is at most~${20/\varepsilon \opt + 16\opt}$. 
				\end{corollary}
				\begin{proof}
					To construct a feasible solution~${(\rest{\x}, \rest{\y})}$ to~\LPkMed{} with the claimed capacity violation bounds, 
					we first apply the procedure of Lemma~\ref{lem:restricted_solution} to compute an opening vector for each star instance.
					Let~${\rest{\y}}$ be the union of all these opening vectors. 
					We have the inequality 
					\[{\vol{\rest{\y}} \le \sum_{\vSc\in \starCenterSet} \fvol{\yo}{\F_\vSc} \le \sum_{i \in \F} \yoComp_i \le k}\formulaPunctuationSpace,\] 
					where the second last inequality follows from the fact that each facility belongs to exactly one star instance. Thus, Constraint~\eqref{lp:volume_k} is fulfilled.
					
						Next, we construct a feasible assignment~${\rest{\x}}$ suitable for~${\rest{\y}}$.
						For every facility~${i\in\F}$, let~${\rest{\demand}_i}$ be its demand given by Lemma~\ref{lem:restricted_solution}.
						For every star instance~${S_\vSc}$ and every client~${\vC\in\C}$, we define~${\xComp_{\vC}^{\vSc}}$ as the total \lpDemand{}~${\sum_{i'\in \F_\vSc} \xoComp_{i'\vC}}$ of~${\vC}$ that is served by the facilities of~${S_\vSc}$ in~${\xo}$. 
						Then we send the fraction~${\rest{\demand}_i / w_\vSc}$ of~${\xComp_{\vC}^{\vSc}}$ to every facility~${i\in\F_{\vSc}}$, that is, we set~${\rest{\xComp}_{i\vC} = \rest{\demand}_i / w_\vSc\cdot \xComp_{\vC}^{\vSc}}$.
						
						Hence, Constraint~\eqref{lp:demand_of_client} holds, as each client~${\vC\in\C}$ is fully served:
						\[\sum_{i \in \F} \rest{\xComp}_{i\vC} =  \sum_{\vSc \in \starCenterSet} \xComp_{\vC}^{\vSc} \sum_{i \in \F_\vSc} \demand_i / w_\vSc  = \sum_{\vSc \in \starCenterSet} \xComp_{\vC}^{\vSc} = 1 \formulaPunctuationSpace.\]
						
						The assignment also implies the satisfaction of Constraint~\eqref{lp:opening_ge_demand}. To see this, consider 
						any client~${\vC\in\C}$, star instance~${\vSc\in\starCenterSet}$ and facility~${i\in\F_{\vSc}}$.
						If~${\rest{\yComp}_i = 1}$, then the constraint immediately holds. Otherwise, by Lemma~\ref{lem:restricted_solution}, we have 
						\[\rest{\yComp}_i = \fvol{\yo}{\F_\vSc} = \sum_{i'\in \F_\vSc} \yoComp_{i'} \ge \sum_{i'\in \F_\vSc} \xoComp_{i'\vC} = \xComp_{\vC}^{\vSc} \ge \rest{\xComp}_{i\vC} \formulaPunctuationSpace. \]
						Observe that the facility~${i}$ serves a total amount of
						\[\sum_{\vC \in \C} \rest{\xComp}_{i\vC} = \rest{\demand}_i / w_\vSc \sum_{\vC \in \C} \xComp_{\vC}^{\vSc} = \rest{\demand}_i\formulaPunctuationSpace.\]
						Hence, the bounds on capacity violation of Lemma~\ref{lem:restricted_solution} still hold. 

					Consequently,~${(\rest{\x}, \rest{\y})}$ is a feasible solution with capacity violations. 
					Recall that Lemma~\ref{lem:restricted_solution} restricts the openings of the facilities.
					For each facility~${i\in\F}$, either~${\rest{\yComp_i}\in\{0,1\}}$, or~${\rest{\yComp_i}=\fvol{\yo}{\F_\vSc}}$, where~${\F_\vSc}$ is the \starTail{} containing~${i}$. Furthermore, by Corollary~\ref{cor:minvolume-bundle} and~${\ell=2}$, we know that every \starTail{}~${\F_\vSc}$ has volume~${\fvol{\yo}{\F_\vSc}}$ at least~${1 - 1/\ell = 1/2}$. 
					Thus,~${(\rest{\x}, \rest{\y})}$ is even a~${[\puthalf, 1]}$-solution.
					
					Regarding the connection cost of~${\rest{\x}}$, we can assume that we first move the demands of the clients to the star centers and then move them from there to the \supporting facilities. 
					The cost of the first step is at most~${(2+2\ell)\opt}$ as Lemma~\ref{lem:move-to-bundle-centers} implies.
					The cost of the second step, for each star center~${\vSc}$, is at most~${(2 + 4/\varepsilon) \bsmall_\vSc}$
					by Lemma~\ref{lem:restricted_solution}. 
					Since we know that the total strict budget over all star instances is upper bounded by~${\sum_{\vSc\in\starCenterSet} \bsmall_\vSc\le(1+2\ell)\opt}$ (Corollary~\ref{cor:total_budget}), 
					the total cost is at most
					\begin{alignat*}{1}
					&(2 + 4/\varepsilon)(1+2\ell) \opt + (2+2\ell)\opt \\
					&=~ 20/\varepsilon \opt + 16\opt\formulaPunctuationSpace. 
					\qedHereInAlign
					\end{alignat*}
					
				\end{proof}
				
				\subsection{Computing a weak~${\{\puthalf,1\}}$-solution.}
				\label{bank_step}
				
				Let~${(\rest{\x}, \rest{\y})}$ be a~${[\puthalf, 1]}$-solution with capacity violations obtained by Corollary~\ref{cor:restr}.
				We will now transform it into a weak~${\{\puthalf, 1\}}$-solution~${(\hat{\x}, \hat{\y})}$. 
				\begin{definition}
					We define 
					\begin{itemize}
						\item~${N_1 = \{i\in \F \mid \rest{\yComp}_i = 1\}}$, and
						\item~${N_2 = \{i\in \F \mid \puthalf \le \rest{\yComp}_i < 1\}}$.
					\end{itemize}
					For each facility~${i \in N_2}$, let~${\rest{\demand}_i = \sum_{\vC\in\C} \rest{\xComp}_{i\vC}}$ be the demand served by~${i}$ and let~${\s{i}}$ be its closest facility in~${N_1 \cup N_2 \elementsetminus{i}}$ (recall Definition~\ref{def:closeness}). 
				\end{definition}
				
				\begin{lemma}\label{lem:gapSolution}
					We can efficiently compute a weak~${\{\puthalf,1\}}$-solution~${(\hat{\x}, \hat{\y})}$ 
					of volume~${\vol{\hat{\y}}\le k}$ and connection cost of at most~ $20/\varepsilon \opt + 16\opt$ 
					where 
					the capacity violation of every facility is at most~${2 + 2\varepsilon}$. 
					Moreover, the following inequality holds:
					
					\[\sum_{i \in N_2} \rest{\demand}_i(1 - \hat{\yComp}_i) \dist{\s{i}}{i} \leq \sum_{i \in N_2} \rest{\demand}_i (1 - \yComp_i') \dist{\s{i}}{i} \formulaPunctuationSpace.\]
				\end{lemma}
				\begin{proof}
					For each facility~${i \in N_1}$, set~${\hat{\yComp}_i = 1}$. 
					If~${k\ge|N_1|+|N_2|}$, then also set~${\hat{\yComp}_i = 1}$ for each facility~${i \in N_2}$.
					Thus, the volume~${\vol{\hat{\y}}}$ is bounded from above by~${k}$, and
					the claimed inequality holds, as the left side adds up to~${0}$.
					
					If~${k<|N_1|+|N_2|}$, then~${|N_2| > 2k - 2|N_1| - |N_2|}$.
					For each facility~${i\in N_2}$, we define~${\demand_i'\dist{\s{i}}{i}}$ as the \emph{weight of the facility~${i}$}.
					Sort all facilities in~${N_2}$ non-increasingly by their weights. 
					Then set the openings of the first~${2k - 2|N_1| - |N_2|}$ facilities of~${N_2}$ to~${1}$,
					and the openings of the remaining facilities to~${\puthalf}$. 
					By this assignment of openings, the volume of~${\hat{\y}}$ is exactly 
					\[\vol{\hat{\y}} = |N_1| + (2k - 2|N_1| - |N_2|) + \frac{1}{2}(|N_2| - (2k - 2|N_1| - |N_2|)) = k \formulaPunctuationSpace.\]
					
					Observe that among all opening vectors~${\vv}$ for~${\F}$ satisfying 
					the two conditions,~(i)~${\fvol{\vv}{N_2}\le k-|N_1|}$, and~(ii)~${\vvComp_i\ge\puthalf}$ for every~${i\in N_2}$,
					vector~${\hat{\y}}$ attains the maximum value for the \emph{objective}~${\sum_{i \in N_2} \demand_i' \hat{\yComp_i} \dist{\s{i}}{i}}$. 
					To see this, take any such opening vector and arbitrarily increase any fractional openings (each to at most~${1}$) until the total volume of the facilities in~${N_2}$ is~${k-|N_1|}$. This does not decrease the objective value. Then iteratively move portions of openings from facilities with small weight to facilities with high weight, until all facilities have openings~${1}$ or~${1/2}$. Also this iterative step does not decrease the objective value. 
					We obtain an opening vector that has the same openings for facilities in~${N_2}$ as~${\hat{\y}}$. Hence, it has the same objective value as~${\hat{\y}}$ and our observation follows. 
					
					Now, observe that~${\rest{\y}}$ fulfills the two conditions~(i) and~(ii). 
					Thus, we have 
					\[\sum_{i \in N_2} \demand_i' \yComp_i' \dist{\s{i}}{i} \leq  \sum_{i \in N_2} \demand_i' \hat{\yComp_i} \dist{\s{i}}{i}\formulaPunctuationSpace,\]
					which is equivalent to the inequality claimed in the lemma statement.

					Next, we set~${\hat{\x} = \rest{\x}}$ and thus have the same cost as~${\rest{\x}}$. 
					Consider the facility openings that decreased. 
					Since these openings 
					changed by a factor not smaller than~${\puthalf}$, 
					we can bound the capacity violation of their facilities by~${2 (1+\varepsilon)}$.
					The capacity violation of other facilities did not increase. Now, open facilities have in worst case 
					capacity violation \[\max\{1+\varepsilon, 2+\varepsilon\}\le2 + 2\varepsilon \formulaPunctuationSpace.\qedHereInAlign\]

				\end{proof}
				
					\begin{definition}
						We define 
						\begin{itemize}
							\item~${\hat{N}_1 = \{i \in \F~|~\hat{\yComp}_i = 1\}}$, and
							\item~${\hat{N}_2 = \{i \in \F~|~\hat{\yComp}_i = \puthalf\}}$.
						\end{itemize}
					\end{definition}
					
					We have~${\hat{N_1} \cup \hat{N_2} = N_1 \cup N_2}$, and in particular~${\hat{N}_2\subseteq N_2}$. 
					Thus,~${\s{i}}$ is well defined for every~${i\in\hat{N_2}}$ and we have~${\s{i}\in\hat{N_1} \cup \hat{N_2}}$.
				
				\subsection{Rounding a weak~${\{\puthalf, 1\}}$-solution~${\hat{\y}}$ to an integral solution~${\bar{\y}}$.}
				\label{final_rounding}
				In the last section, we obtained 
				a weak~${\{\puthalf, 1\}}$-solution~${(\hat{\x}, \hat{\y})}$ by Lemma~\ref{lem:gapSolution}.  
				In
				this section, we describe how to round this solution to an integral
				solution~${(\bar{\x},\bar{\y})}$.  For the sake of easier presentation, we
				assume that the demands of the clients have been moved to the facilities via
				the solution~${(\hat{\x},\hat{\y})}$ so that every facility~${i\in \hat{N}_1\cup\hat{N}_2}$ carries the demand~${\demand_i'}$.  
				We will describe how to obtain an integral opening vector~${\bar{\y}}$ and how to further reroute the demand 
				to facilities that are open in~${\bar{\y}}$. 
				We will give an upper
				bound of~${3+3\varepsilon}$ on the capacity violation and analyze the
				cost of the rerouting.  
				
				Combining the rerouting with the assignment~${\hat{\x}}$, we obtain~${\bar{\x}}$. 
				Altogether, this leads to the solution~${(\bar{\x},\bar{\y})}$ to the original instance (where the demand resides in
				the clients) with capacity violation~${3+3\varepsilon}$. The cost of
				this solution is the total cost of~${(\hat{\x},\hat{\y})}$ and the rerouting.
				
				\paragraph{Building facility trees.} In a similar way as in related works~\cite{CharikarkMedConst1999,Shanfei_Li_cap2014}, we construct a directed forest of in-trees spanning the facilities in~${\hat{N}_2}$.
				For this, we run Procedure \ShortTrees($\hat{N}_2$,~${\hat{N}_1}$) as described in Section~\ref{sec:bundles}. 
				By the construction, nodes in~${\hat{N}_1}$ may appear only as roots (Lemma~\ref{lem:root-son-distance}), and each node~${i\in\hat{N}_2}$ has either a directed edge to its closest node~${\s{i}}$ in~${\hat{N}_1 \cup \hat{N}_2 \elementsetminus{i}}$, or is a root of an in-tree. In the following, we will call the in-trees facility trees.

				\paragraph{Decomposing facility trees to rooted facility stars.} We cut each facility tree~${T}$ into \emph{facility stars} consisting of a root and a group of leafs. To this end, we greedily choose the leaf node~${i}$ that has the largest number of edges on the path to the root of its tree. Then we remove the subtree rooted at~${\s{i}}$. We call the removed subtree a facility star and use~${Q_{\s{i}}}$ to denote it. 		
				See Procedure~\FacilityStars($T$) and \fig~\ref{fig:facility_stars:decomposition}. 
				\begin{procedure}
					\caption{\FacilityStars($T$)}
					\While{there are at least two nodes in~${T}$}{
						choose a leaf node~${i}$ with the largest number of edges on the path from~${i}$ to the root\; 
						consider the subtree rooted at~${\s{i}}$ as a rooted facility star~${Q_{\s{i}}}$, and remove this subtree\;
					}
					\If{only one node~${i}$ is left and~${\hat{\yComp}_i < 1}$}{
						add~${i}$ to the star rooted at~${\s{i}}$ as a child of~${\s{i}}$\;
					}
				\end{procedure}
  \begin{figure}[htb]
  	\captionsetup[subfloat]{captionskip=5pt}
  	\centering
  	\subfloat[A facility tree~${T}$. Shaded in gray are the facility stars that have been found in the while-loop of Procedure \FacilityStars($T$). Note that the root~${i}$ has not been assigned to any star during the while-loop. 
  	If its opening~${\hat{\yComp}_i}$ is less than~${1}$, it will be added to the son~${i'}$ that satisfies~${i'=\s{i}}$. \label{fig:facility_stars:decomposition}]{\quad\quad\quad\quad\includegraphics{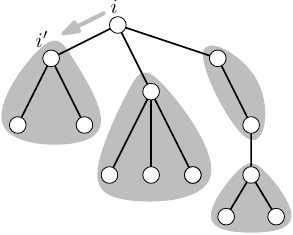}\quad\quad\quad\quad}\hfill
  	\subfloat[A facility star. Shaded in gray is a partition of the fractional facilities into pairs and triples. Here, also the root is a fractional facility. In each tuple, only one facility with the highest demand will be opened. It serves the demand of its closed partners.\label{fig:facility_stars:rounding}]{\quad\quad\quad\includegraphics[page=2]{Fig_Facility-Stars}\quad\quad\quad}\hfill
  	\caption{A facility star decomposed into facility stars, and a facility star partitioned into tuples.}
  	\label{fig:facility_stars}
  \end{figure}

				There is a special case, when only one facility~${i}$ is left in~${T}$. 
				If~${\hat{\yComp}_i = 1}$, we just ignore this node, otherwise~${\hat{\yComp}_i < 1}$ and we have~${i\in\hat{N}_2}$. 
				Note that~${i}$ can only be the root of~${T}$.
				Then, by Lemma~\ref{lem:root-son-distance},~${\s{i}}$ was a son of~${i}$. 
				Since~${\s{i}}$ was removed from~${T}$ but its father~${i}$ was not,~${\s{i}}$  
				must be the root of a facility star.
				Hence, we just add~${i}$ to~${Q_{\s{i}}}$.
				
				\paragraph{Rounding facility stars.} 
				Using the facility stars, we will round~${\hat{\y}}$ and reroute the demand to obtain our integral solution~${(\bar{\x},\bar{\y})}$.
				First, we open all facilities in~${\hat{N}_1}$ that do not belong to any facility star. Each of them is serving its own demand.
				Then we apply the following procedure on each facility star~${Q_\rFS}$ to open at most~${\floor{ \sum_{i \in Q_\rFS} \hat{\yComp}_i }}$ facilities in each of them: 
				
				If~${Q_\rFS}$ contains at least two fractional facilities, we partition all fractional facilities into pairs and triples; see \Fig.~\ref{fig:facility_stars:rounding}. 
				In each pair and triple, we open a facility~${i}$ that has the biggest demand~${\demand_i'}$, close all other facilities in its tuple and route the demand of the closed facilities to~${i}$. 
				If~${Q_\rFS}$ also contains a facility~${i}$ in~${\hat{N}_1}$, we open it and let it serve its own demand.

				If~${Q_\rFS}$ contains exactly one fractional facility~${i}$, then~${Q_\rFS}$ contains also an open facility~${i'}$ of~${\hat{N}_1}$. 
				If~${\demand_{i'}' < 2 \demand_{i}'}$, we open the facility~${i}$, close~${i'}$ and 
				move the demand of~${i'}$ to~${i}$. 
				Otherwise, if~${\demand_{i'}'/2 \geq \demand_{i}'}$, we open the facility~${i'}$, close~${i}$ and move the demand of~${i}$ to~${i'}$.

				Next, we show that our solution~${(\bar{\x},\bar{\y})}$ has at most~${k}$ open facilities and small capacity violation.
				
				\begin{lemma}\label{capacity_violation_3_eps}
					In the integral solution~${\bar{\y}}$, at most~${k}$ facilities are open and every facility has capacity violation 
					at most~${3 + 3\varepsilon}$. 
				\end{lemma}
				\begin{proof}
					Observe that every facility star~${Q_\rFS}$ contains at most~${\floor{\sum_{i \in Q_\rFS} \hat{\yComp}_i}}$ open facilities in~${\bar{\y}}$. 
					Among the facilities that do not belong to any facility star, 
					we open only those that already had opening~${1}$ in~${\hat{\y}}$. 
					So we have~${\sum_{i \in F} \bar{\yComp}_i \leq \sum_{i \in F} \hat{\yComp}_i \le k}$  which implies the first claim.

					In the solution~${(\hat{\x}, \hat{\y})}$, the capacity violation of each facility~${i \in \hat{N}_1 \cup \hat{N}_2}$ is at most~${2 + 2\varepsilon}$. 
					To bound the capacity violation in~${(\bar{\x}, \bar{\y})}$ by~${3 + 3\varepsilon}$, it suffices to bound the increase of capacity violation by the factor~${3/2}$.				
					
					Consider any facility~${i\in\hat{N}_2}$ that we opened in~${\bar{\y}}$. 
					By our choice to open it, we sent at most~${2\demand_i'}$ units of demand to it, either from closed facilities in its tuple, or from the root~${\rFS}$ of its facility star.
					Thus, the demand of~${i}$ increased by a factor at most~${3}$. Given that we simultaneously increased its opening by the sfactor~${2}$, the capacity violation of~${i}$ increased by a factor at most~${3/2}$.
					
					Next, consider any facility in~${i \in \hat{N}_1}$ that we opened. 
					If~${i}$ serves only its own demand, its capacity violation did not increase.
					Otherwise,~${i}$ 
					is also serving the demand of some facility~${i'\in\hat{N}_2}$ and we have~${\demand_{i}'/2 \geq \demand_{i'}'}$. Thus, the demand of~${i}$ increased by a factor at most~${3/2}$, and so its capacity violation.
				\end{proof}
				
				To this end, we bound the cost of our solution~${(\bar{\x},\bar{\y})}$. For this, we need to bound the rerouting cost. We will do it in two steps. First we provide an upper bound that depends on the facility demands and the distances within facility stars. In the second step, we relate this upper bound to~${\opt}$. 
				
				\begin{lemma}\label{cost_of_rounding}
					The cost of rerouting the demand from the facilities that are \supporting in~${\hat{\y}}$ to the facilities that are open in~${\bar{\y}}$ is at most~${2\sum_{i \in \hat{N}_2} \demand_i'\dist{\s{i}}{i}}$.
				\end{lemma}		
				\begin{proof}
					Consider any facility star~${Q_\rFS}$.
					To bound the rerouting cost, we assume, by triangle inequality, that demand is rerouted only along the edges of~${Q_\rFS}$.
					For each such edge~${(i,\rFS)}$, we have~${\rFS=\s{i}}$ and~${i\in\hat{N}_2}$. We don't have~${i\in\hat{N}_1}$ as nodes of~${\hat{N}_1}$ can appear only as a roots. 
					Thus, it suffices to show that each edge~${(i,\rFS)}$ carries at most~${2\demand_i'}$ units of demand.
					Consider any such edge~${(i,\rFS)}$.
					If~${i}$ is closed, it sends its demand~${\demand_i'}$ along~${(i,\rFS)}$ to~${\rFS}$ (from where the demand might be further routed) and no other demand is routed along~${(i,\rFS)}$.
					If~${i}$ is opened, it receives at most~${2\demand_i'}$ units of demand as discussed in the proof of Lemma~\ref{capacity_violation_3_eps}, and no other demand is routed along~${(i,\rFS)}$.
				\end{proof}
				
				Before we further bound the term~${\sum_{i \in \hat{N}_2} \demand_i'\dist{\s{i}}{i}}$, we first show the following helpful inequality.

				\begin{lemma}\label{lem:helpful_inequality}
					For every~${\vC\in \C}$, we have 
					\[\sum_{i \in N_2} \xComp'_{i\vC}  (1 - \yComp'_i) \dist{\s{i}}{i} \leq 2 \sum_{i \in N_1 \cup N_2}  \xComp'_{i\vC} \dist{i}{\vC} \formulaPunctuationSpace.\] 
				\end{lemma}
				\begin{proof}
					Fix any~${\vC\in\C}$.
					By Constraint~\eqref{lp:opening_ge_demand}, we have \[{1 - \yComp'_i \leq 1 - \xComp'_{i\vC} ~=\! \sum_{i' \in N_1 \cup N_2 \elementsetminus{i}} \xComp'_{i'\vC}}\] for every~${i\in N_2}$, and thus 
					\begin{alignat*}{2}
					&&~&\sum_{i \in N_2}  \xComp'_{i\vC}  (1 - \yComp'_i) \dist{\s{i}}{i}
					\\&\le&& \sum_{i \in N_2}   \xComp'_{i\vC} \sum_{i' \in N_1 \cup N_2 \elementsetminus{i}} \xComp'_{i'\vC} \dist{\s{i}}{i} \formulaPunctuationSpace. 
					\end{alignat*}
					By the definition of~${\s{i}}$,~${\dist{\s{i}}{i} \leq \dist{i'}{i}}$ for every~${i' \in N_1 \cup N_2 \elementsetminus{i}}$. Using this, we can further upper bound the expression above by
					\begin{alignat*}{3}
					&&~& \sum_{i \in N_2}   \xComp'_{i\vC} \sum_{i' \in N_1 \cup N_2 \elementsetminus{i}}\xComp'_{i'\vC}\dist{i'}{i}&&
					\\&\leq&& \sum_{i \in N_2}   \xComp'_{i\vC} \sum_{i' \in N_1 \cup N_2 \elementsetminus{i}} \xComp'_{i'\vC}(\dist{i'}{\vC} &&+ \dist{i}{\vC}) \tag{Triangle inequality}
					\\&=&&\sum_{i \in N_2}  \xComp'_{i\vC} \sum_{i' \in N_1 \cup N_2 \elementsetminus{i}} \xComp'_{i'\vC} \dist{i'}{\vC} 
					&&+  \sum_{i \in N_2} \xComp'_{i\vC} \sum_{i' \in N_1 \cup N_2 \elementsetminus{i}} \xComp'_{i'\vC} \dist{i}{\vC}
					\\&\le&&\sum_{i' \in N_1 \cup N_2}  \xComp'_{i'\vC} \dist{i'}{\vC} \sum_{i \in N_2} \xComp'_{i\vC}
					&&+ \sum_{i \in N_2}   \xComp'_{i\vC} \dist{i}{\vC} \sum_{i' \in N_1 \cup N_2 \elementsetminus{i}} \xComp'_{i'\vC} 
					\\&\le&&\sum_{i' \in N_1 \cup N_2}  \xComp'_{i'\vC} \dist{i'}{\vC} &&+ \sum_{i \in N_2} \xComp'_{i\vC} \dist{i}{\vC}
					\\&\leq&& 2  \sum_{i \in N_1 \cup N_2} \xComp'_{i\vC} \dist{i}{\vC}\formulaPunctuationSpace, && 
					\end{alignat*}
					where the second last inequality follows from~${\sum_{i \in N_1 \cup N_2} \xComp'_{i\vC} = 1}$ for each~${\vC \in \C}$, given by Constraint~\eqref{lp:demand_of_client}.
				\end{proof}
				
				We are ready to relate the rerouting cost to~${\opt}$.
				
				\begin{lemma}\label{cost_of_star}
					The sum~${\sum_{i \in \hat{N}_2} \demand_i'\dist{\s{i}}{i}}$ 
					is at most~${80/\varepsilon\opt +\, 64\opt}$.
				\end{lemma}
				\begin{proof}
					
					For each~${i \in\hat{N}_2}$, we have~${\hat{\yComp_i} = \puthalf}$, so
					\[\demand_i'\dist{\s{i}}{i} = 2 \demand_i' (1 - \hat{\yComp_i}) fvol{\s{i}}{i}\formulaPunctuationSpace.\]
					Thus, 
					\begin{alignat*}{2}
					&&~&\hphantom{2} \sum_{i \in \hat{N}_2} \demand_i'\dist{\s{i}}{i}
					\\&=&&2 \sum_{i\in\hat{N}_2} \demand_i' (1 - \hat{\yComp_i})\dist{\s{i}}{i} 
					\\&\le&&2 \sum_{i \in N_2}  \demand_i' (1 - \hat{\yComp_i}) \dist{\s{i}}{i} \tag{$\hat{N}_2 \subseteq N_2$}
					\\&\leq&~& 2 \sum_{i \in N_2} \demand_i' (1 - \yComp_i') \dist{\s{i}}{i}  \tag{Lemma~\ref{lem:gapSolution}}
					\\&=&& 2 \sum_{\vC\in \C} \sum_{i \in N_2} \xComp'_{i\vC} (1 - \yComp_i') \dist{\s{i}}{i} \tag{Definition of~${(\demand'_i)_{i\in\F}}$}
					\\&\le&& 4 \sum_{\vC \in \C}\sum_{i \in N_1 \cup N_2} \xComp_{i\vC}'\dist{i}{\vC} \tag{Lemma~\ref{lem:helpful_inequality}}
					\\&\leq&& 4 \left(20/\varepsilon \opt \,+\, 16\opt\right) \tag{Corollary \ref{cor:restr}}
					\\&=&& 80/\varepsilon \opt \,+\, 64\opt\formulaPunctuationSpace.\qedHereInAlign
					\end{alignat*} 
				\end{proof}
				
				Now we have all the ingredients to bound the cost of our solution and to prove Theorem~\ref{thm:main_3_eps}.
				\begin{proof}[Proof of Theorem~\ref{thm:main_3_eps}]
					From Lemma \ref{lem:gapSolution}, we know that the cost of solution~${(\hat{\x}, \hat{\y})}$ is at most~${20/\varepsilon \opt + \, 16\opt}$. 
					This corresponds also to the cost of sending the demand from the clients to the facilities \supporting in~${\hat{\y}}$.
					Using Lemmas \ref{cost_of_rounding} and \ref{cost_of_star}, we bound the cost of rerouting the demand to the facilities open in~${\bar{\y}}$ by
					\[{2(80/\varepsilon\opt + 64\opt)}\formulaPunctuationSpace.\] 
					Summing this up, we obtain
					\[180/\varepsilon \opt \,+\,  144\opt \]
					as an upper bound for the cost of our solution~${(\bar{\x}, \bar{\y})}$ with capacity violation~${3+3\varepsilon}$.
				\end{proof}

			\section{Concluding Remarks and Open Questions}
			In this \kapitel, we gave the first approximation algorithms for hard-capacitated~\PROBkfl problems, where we considered either non-uniform capacities or non-uniform opening costs.
			Both algorithms are based on the standard LP relaxation, a reduction to single-demand-node instances, and a tree structure to guide the distribution of demand. 
			
			Our results imply two insights on the integrality gap of the standard~LP:
			For uniform capacities, the~${2}$ barrier on capacity violation is tight up to an arbitrarily small constant. For non-uniform capacities, the barrier is located between~${2}$ and~${3}$; it would be interesting to pinpoint it tighter. 
		
			It also remains open to construct an algorithm for the generalization of the two settings above, that is, for hard-capacitated~\PROBkfl where both, the capacities and openings, are non-uniform.
			For such an algorithm, it would be appealing to base it on a new LP relaxation like the one introduced by Li~\cite{LiSODA2016}

			Eventually, the big open question is whether capacitated \PROBkmed admits a constant-factor approximation algorithm.
			
			\section*{Acknowledgments}
			We are very thankful to Neelima Gupta for pointing out some inaccuracies in the previous version~\cite{ByrkaCapKmed2015} of this paper.

		\end{document}